\documentclass{lmcs}
\pdfoutput=1
\usepackage[utf8]{inputenc}

\makeatletter
\def\enddoc@text{%
}
\makeatother

\usepackage{lastpage}
\lmcsdoi{19}{2}{15}
\lmcsheading{}{\pageref{LastPage}}{}{}%
{Aug.~20,~2021}{Jun.~07,~2023}{}

\usepackage{hyperref}

\usepackage[appendix=append,bibliography=common]{apxproof}

\usepackage{etoolbox}
\makeatletter
\apptocmd{\@enddocumenthook}{
    \vspace{-2cm}
    \nobreak%
    \insert\copyins{\hsize.58\textwidth
\vbox to 0pt{\vskip12 pt%
      \fontsize{6}{7\p@}\normalfont\upshape
      \everypar{}%
      \noindent\fontencoding{T1}%
  \headertextsf{This work is licensed under the Creative Commons
  Attribution License. To view a copy of this license, visit
  \texttt{https://creativecommons.org/licenses/by/4.0/} or send a
  letter to Creative Commons, 171 Second St, Suite 300, San Francisco,
    CA 94105, USA, or Eisenacher Strasse~2, 10777 Berlin, Germany}\vss}
      \par
      \kern\z@}%
}
\makeatother

\usepackage{mathtools}
\usepackage{bm}

\usepackage{tikz}
\usetikzlibrary{automata,arrows}

\usepackage[all]{xy}
\usepackage{url}
\usepackage[colorinlistoftodos,textsize=tiny,color=orange!70%
,disable%
]{todonotes}
\usepackage[most]{tcolorbox}
\usepackage{subcaption}
\usepackage{amssymb,amsmath,stmaryrd}
\newtheoremrep{theorem}[defi]{Theorem}
\newtheoremrep{lemma}[defi]{Lemma}
\newtheoremrep{proposition}[defi]{Proposition}
\newtheoremrep{corollary}[defi]{Corollary}

\renewcommand{\int}{\ensuremath{\mathbb{Z}}}

\newcommand{\one}{\ensuremath{1}}
\newcommand{\zero}{\ensuremath{0}}

\newcommand{\comp}[1]{\ensuremath{\overline{#1}}}

\newcommand{\symmclose}[1]{\ensuremath{\mathit{Sym}(#1)}}

\newcommand{\inc}[2]{\iota_{#1}^{#2}}

\newcommand{\Pow}[1]{\ensuremath{\mathbf{2}^{#1}}}

\newcommand{\Ytop}[2]{[{#1}]^{#2}}
\newcommand{\Ybot}[2]{[{#1}]_{#2}}

\newcommand{\interval}[2]{\ensuremath{[{#1}, {#2}]}}

\newcommand{\norm}[1]{\ensuremath{|\!|{#1}|\!|}}

\newcommand{\mins}{\min\nolimits}
\newcommand{\maxs}{\max\nolimits}

\newcommand{\suc}{\eta}

\makeatletter
\newcommand{\superimpose}[2]{%
  {\ooalign{$#1\@firstoftwo#2$\cr\hfil$#1\@secondoftwo#2$\hfil\cr}}}
\makeatother
\newcommand{\oominus}{\mathrel{\mathpalette\superimpose{{\ominus}{\div}}}}

\newcommand{\dual}[1]{\ensuremath{{#1}^\mathit{op}}}

\newcommand{\monM}{\mathbb{M}}

\definecolor{dmagenta}{rgb}{0.81,0,0.81}
\definecolor{dcyan}{rgb}{0,0.6,0.6}
\definecolor{dgreen}{rgb}{0.09, 0.45, 0.27}
\definecolor{dred}{rgb}{0.55, 0.0, 0.0}

\newcommand{\dgreen}{\color{dgreen}}
\newcommand{\blue}{\color{blue}}
\newcommand{\red}{\color{dred}}

\newcommand{\mystrutab}{\raisebox{0ex}[0.4cm]{}}
\newcommand{\mystrutbl}{\raisebox{-2ex}[0.4cm]{}}


\allowdisplaybreaks

\begin{document}

\title{Fixpoint Theory - Upside Down}

\thanks{This work is supported
  by the MIUR project PRIN2017- ASPRA, Grant No. 201784YSZ5, and the
  DFG projects BEMEGA (project number 260261790) and SpeQt (project
  number 434050016).}

\author[P.~Baldan]{Paolo Baldan}[a]

\author[R.~Eggert]{Richard Eggert}[b]

\author[B.~K\"onig]{Barbara K\"onig}[b]

\author[T.~Padoan]{Tommaso Padoan}[a]
\address{Universit\`a di Padova, Italy}
\email{baldan@math.unipd.it, padoan@math.unipd.it}

\address{Universit\"at Duisburg-Essen, Germany}
\email{richard.eggert@uni-due.de, barbara\_koenig@uni-due.de}

\begin{abstract}
  Knaster-Tarski's theorem, characterising the greatest fixpoint of a
  monotone function over a complete lattice as the largest
  post-fixpoint, naturally leads to the so-called coinduction proof
  principle for showing that some element is below the greatest
  fixpoint (e.g., for providing bisimilarity witnesses). The dual
  principle, used for showing that an element is above the least
  fixpoint, is related to inductive invariants.
  In this paper we provide proof rules which are similar in spirit but
  for showing that an element is above the greatest fixpoint or,
  dually, below the least fixpoint. The theory is developed for
  non-expansive monotone functions on suitable lattices of the form
  $\monM^Y$, where $Y$ is a finite set and $\monM$ an MV-algebra, and
  it is based on the construction of (finitary) approximations of the
  original functions. We show that our theory applies to a wide range
  of examples, including termination probabilities, metric transition
  systems, behavioural distances for probabilistic automata and
  bisimilarity. Moreover
  it allows us to determine original algorithms for solving simple
  stochastic games.
\end{abstract}

\maketitle

\section{Introduction}
\label{se:introduction}

Fixpoints are ubiquitous in computer science as they provide
a meaning to inductive and coinductive definitions (see, e.g.,~\cite{San:IntroBisCoind,NNH:PPA}). A monotone
function $f : L\to L$ over a complete lattice
$(L,\sqsubseteq)$, by Knaster-Tarski's
theorem~\cite{t:lattice-fixed-point}, admits a least fixpoint
$\mu f$ and greatest fixpoint $\nu f$ which are characterised as the least pre-fixpoint and the greatest post-fixpoint, respectively. This immediately
gives well-known proof principles for showing that a lattice element
$l\in L$ is \emph{below} $\nu f$  or \emph{above} $\mu f$
\[
  \frac{l \sqsubseteq f(l)}{l\sqsubseteq \nu f}
  \qquad\qquad
  \frac{f(l)\sqsubseteq l}{\mu f\sqsubseteq l}
\]

On the other hand, showing that a given element $l$ is \emph{above}
$\nu f$ or \emph{below} $\mu f$ is more difficult. One can think of
using the characterisation of least and largest fixpoints via Kleene's
iteration. E.g., the largest fixpoint is the least element of the
(possibly transfinite) descending chain obtained by iterating $f$ from
$\top$. Then showing that $f^i(\top) \sqsubseteq l$ for some $i$,
one concludes that $\nu f \sqsubseteq l$. This proof principle is
related to the notion of ranking functions. However, this is a less
satisfying notion of witness since $f$ has to be applied $i$ times,
and this can be inefficient or unfeasible when $i$ is an infinite
ordinal.

The aim of this paper is to present an alternative proof rule for this
purpose for functions over lattices of the form $L = \monM^Y$
where $Y$ is a finite set and $\monM$ is an MV-chain, i.e., a totally
ordered complete lattice endowed with suitable operations of sum and
complement. This allows us to capture several examples, ranging from
ordinary relations for dealing with bisimilarity to behavioural
metrics, termination probabilities and simple stochastic
games.

Assume $f : \monM^Y \to \monM^Y$ monotone and consider the question of
proving that some fixpoint $a : Y\to \monM$ is the largest
fixpoint $\nu f$.
The idea is to show that there is no ``slack'' or ``wiggle room'' in
the fixpoint $a$ that would allow us to further increase it. This is
done by associating with every $a : Y\to \monM$ a function
$f^\#_a$ on $\Pow{Y}$ whose greatest fixpoint gives us the
elements of $Y$ where we have a potential for increasing $a$ by adding
a constant. If no such potential exists, i.e.\ $\nu f^\#_a$ is empty,
we conclude that $a$ is $\nu f$. A similar function $f_\#^a$
(specifying decrease instead of increase) exists for the case of least
fixpoints. Note that the premise is $\nu f_\#^a = \emptyset$,
i.e. the witness remains coinductive. The proof rules are:
\[
  \frac{f(a) = a \qquad \nu f^\#_a = \emptyset}{\nu
    f = a}
  \qquad
  \frac{f(a) = a \qquad \nu f_\#^a = \emptyset}{\mu
    f = a}
\]

For applying the rule we compute a greatest fixpoint
on $\Pow{Y}$, which is finite, instead of working on the potentially infinite
$\monM^Y$. The rule does not work for all monotone functions
$f : \monM^Y\to \monM^Y$, but we show that whenever $f$ is
non-expansive the rule is valid. Actually, it is not only sound,
but also reversible, i.e., if $a=\nu f$ then
$\nu f_a^{\#} = \emptyset$, providing an if-and-only-if
characterisation of whether a given fixpoint corresponds to the
  greatest fixpoint.

Quite interestingly, under 
the same assumptions on $f$, using a
restricted function $f_a^*$, the rule can be used, more generally,
when $a$ is just a \emph{pre-fixpoint} ($f(a)\sqsubseteq a$) and it
allows to conclude that $\nu f\sqsubseteq a$. A dual result holds for
\emph{post-fixpoints} in the case of least fixpoints.
\[
  \frac{f(a)\sqsubseteq a \qquad \nu f^*_a = \emptyset}{\nu
    f\sqsubseteq a} \qquad \frac{a \sqsubseteq f(a) \qquad \nu f_*^a
    = \emptyset}{a \sqsubseteq \mu f}
\]

As already mentioned, the theory above applies to many interesting
scenarios: witnesses for non-bisimilarity, algorithms for simple
stochastic games~\cite{condon92}, lower bounds for termination
probabilities and behavioural metrics in the setting of
probabilistic~\cite{bblm:on-the-fly-exact-journal} and metric
transition systems \cite{afs:linear-branching-metrics} and
probabilistic automata~\cite{bblmtv:prob-bisim-distance-automata}. In
particular we were inspired by, and generalise, the self-closed
relations of~\cite{f:game-metrics-markov-decision}, also used
in~\cite{bblmtv:prob-bisim-distance-automata}.

\paragraph*{Motivating example.}

Consider a Markov chain $(S, T, \suc)$ with a finite set of states $S$,
where $T\subseteq S$ are the terminal states and every state
$s\in S\backslash T$ is associated with a probability
distribution
$\suc(s)\in \mathcal{D}(S)$.\footnote{By $\mathcal{D}(S)$ we denote the set of
  all maps $p : S\to [0,1]$ such that $\sum_{s\in S} p(s) = 1$.}
Intuitively, $\suc(s)(s')$ denotes the probability of state $s$
choosing $s'$ as its successor.
Assume that, given a fixed state $s\in S$, we want to determine
the termination probability of $s$, i.e. the probability of 
eventually reaching any terminal state from $s$.
As a concrete example, take the Markov chain given in
Fig.~\ref{fig:markov-chain}, where $u$ is the only terminal state.

The termination probability arises as the least fixpoint of a function
$\mathcal{T}$ defined as in Fig.~\ref{fig:markov-chain}. The values
of $\mu \mathcal{T}$ are indicated in green (left value).

Now consider the function $t$ assigning to each state the termination
probability written in red (right value). It is not difficult to see
that $t$ is another fixpoint of $\mathcal{T}$, in which states $y$ and
$z$ convince each other incorrectly that they terminate with
probability~$1$, resulting in a vicious cycle that gives
``wrong'' results. We want to show that $\mu \mathcal{T} \neq t$
without knowing $\mu\mathcal{T}$.
Our idea is to compute the set of states that still has some ``wiggle
room'', i.e., those states which could reduce their termination
probability by $\delta$ if all their successors did the same. This
definition has a coinductive flavour and it can be computed as a
greatest fixpoint on the finite powerset $\Pow{S}$ of states, instead
of on the infinite lattice $\interval{0}{1}^S$.

We hence consider a function
$\mathcal{T}_\#^t : \Pow{\Ytop{S}{t}}\to \Pow{\Ytop{S}{t}}$,
dependent on $t$, defined as follows. Let $\Ytop{S}{t}$ be the support
of $t$, i.e., the set of all states $s$ such that $t(s) > 0$, where
a reduction in value is in principle
possible. Then a state $s\in \Ytop{S}{t}$ is in $\mathcal{T}_\#^t(S')$ iff
$s\not\in T$ and for all $s'$ for which $\suc(s)(s') > 0$ it holds that
$s'\in S'$, i.e.\ all successors of $s$ are in $S'$.

The greatest fixpoint of $\mathcal{T}_\#^t$ is $\{y,z\}$. The fact
that it is not empty means that there is some ``wiggle room'', i.e.,
the value of $t$ can be reduced on the elements $\{y,z\}$ and thus $t$
cannot be the least fixpoint of $f$. Moreover, the intuition that $t$
can be improved on $\{y,z\}$ can be made precise, leading to the
possibility of performing the improvement and search for the least
fixpoint from there.

\paragraph*{Contributions.} In the paper we formalise the theory
outlined above, showing that the proof rules work for non-expansive
monotone functions $f$ on lattices of the form $\monM^Y$, where $Y$ is
a finite set and $\monM$ a (potentially infinite) MV-algebra
(Section~\ref{se:nonexpansive-approximation}
and Section~\ref{se:proof-rules}). Additionally, given a decomposition of
$f$ we show how to obtain the corresponding approximation
compositionally (Section~\ref{se:de-composing}). Then, in order to show
that our approach covers a wide range of examples and allows us to
derive useful and original algorithms, we discuss various
applications: termination probability, behavioural distances for
metric transition systems and probabilistic automata, bisimilarity
(Section~\ref{se:applications}) and simple stochastic games
(Section~\ref{se:ssgs}).

Further proofs and material can be found in the
appendix.

\begin{figure}[t]
  \centering
  \begin{subfigure}{0.55\textwidth}
    \centering
	  \normalsize
	  \begin{eqnarray*}
	    && \mathcal{T} : [0,1]^S \to [0,1]^S \\
	    && \mathcal{T}(t)(s) = \left\{
	    \begin{array}{ll}
	      1 & \mbox{if $s\in T$} \\
	      \sum\limits_{s'\in S} \suc(s)(s') \cdot t(s') & \mbox{otherwise}
	    \end{array}
	    \right.
	  \end{eqnarray*}
  \end{subfigure}%
  \begin{subfigure}{0.45\textwidth}
    \centering
	  \begin{tikzpicture}[->, >=stealth, nodes={minimum
	      size=1.8em}, node distance=0.8cm]
	    \node [circle,draw](x) [label=below:{\dgreen $\frac{1}{2}$}/{\red $1$}] {$x$};
	    \node [circle,draw,accepting](t) [right=of x]
	    [label=below:{\dgreen $1$}/{\red $1$}] {$u$};
	    \node [circle,draw](y) [left=of x] [label=below:{\dgreen
	      $0$}/{\red $1$}] {$y$};
	    \node [circle,draw](z) [left=of y][label=below:{\dgreen $0$}/{\red $1$}] {$z$};
	    \path [thick] (x) edge [loop above] node {$\frac{1}{3}$} (x);
	    \draw [->,thick] (x) to node [above]{$\frac{1}{3}$} (t);
	    \draw [->,thick] (x) to node [above]{$\frac{1}{3}$} (y);
	    \draw [->,thick,bend left] (y) to node [below]{$1$} (z);
	    \draw [->,thick,bend left] (z) to node [above]{$1$} (y);
	  \end{tikzpicture}
  \end{subfigure}
  \caption{Function $\mathcal{T}$ (left) and a Markov
    chain with two fixpoints of $\mathcal{T}$ (right)}
  \label{fig:markov-chain}
\end{figure}

\section{Lattices and MV-algebras}
\label{se:setup}

In this section, we review some basic notions used in the paper,
concerning complete lattices and MV-algebras~\cite{Mun:MV}.

A preordered or partially ordered set $(P, \sqsubseteq)$
is often denoted simply as $P$, omitting the order
relation.
Given $x, y \in P$, with $x \sqsubseteq y$, we denote by
$\interval{x}{y}$ the interval
$\{ z \in P \mid x \sqsubseteq z \sqsubseteq y\}$.
The \emph{join} and the \emph{meet} of a
subset $X \subseteq P$ (if they exist) are denoted $\bigsqcup X$
and $\bigsqcap X$, respectively.
 
A \emph{complete lattice} is a partially ordered set
$(L, \sqsubseteq)$ such that each subset $X \subseteq L$ admits a
join $\bigsqcup X$ and a meet $\bigsqcap X$. A complete lattice
$(L, \sqsubseteq)$ always has a least element
$\bot = \bigsqcup \emptyset$ and a greatest element
$\top = \bigsqcap \emptyset$.

A function $f : L \to L$ is \emph{monotone} if for all
$l, l' \in L$, if $l \sqsubseteq l'$ then
$f(l) \sqsubseteq f(l')$. By Knaster-Tarski's
theorem~\cite[Theorem~1]{t:lattice-fixed-point}, any monotone
function on a complete lattice has a least and a greatest fixpoint,
denoted respectively $\mu f$ and $\nu f$, characterised as the meet
of all pre-fixpoints, respectively the join of all post-fixpoints:
$\mu f = \bigsqcap \{ l \mid f(l) \sqsubseteq l \}$
and 
$\nu f = \bigsqcup \{ l \mid l \sqsubseteq f(l) \}$.

Let $(C, \sqsubseteq)$, $(A, \leq)$ be complete lattices. A
\emph{Galois connection} is a pair of monotone functions
$\langle \alpha, \gamma\rangle$ such that $\alpha : C \to A$,
$\gamma : A \to C$ and for all $a \in A$ and $c \in C$:
\begin{center}
  $\alpha(c) \leq a$ \quad iff \quad $c \sqsubseteq \gamma (a)$.
\end{center}
Equivalently, for all $a \in A$ and $c \in C$, (i)~$c \sqsubseteq \gamma (\alpha(c))$ and
(ii)~$\alpha(\gamma(a)) \leq a$.
In this case we will write
$\langle \alpha, \gamma\rangle : C \to A$.
For a Galois connection
$\langle \alpha, \gamma \rangle: C \to A$, the function
$\alpha$ is called the left (or lower) adjoint and $\gamma$ the right
(or upper) adjoint.

Galois connections are at the heart of abstract
interpretation~\cite{cc:ai-unified-lattice-model,CC:TLA}. In
particular, when $\langle \alpha, \gamma \rangle$ is a Galois
connection, given $f^C : C \to C$ and $f^A : A \to A$, monotone
functions, if $f^C \circ \gamma \sqsubseteq \gamma \circ f^A$, then
$\nu f^C \sqsubseteq \gamma(\nu f^A)$. If the equality
$f^C \circ \gamma = \gamma \circ f^A$ holds, a condition sometimes referred to
as \emph{$\gamma$-completeness}, then greatest fixpoints are preserved along
the connection, i.e., $\nu f^C = \gamma(\nu f^A)$.

Given a set $Y$ and a complete lattice $L$, the set
of functions $L^Y = \{ f \mid f : Y \to L \}$, endowed with
pointwise order, i.e., for $a, b \in L^Y$, $a \sqsubseteq b$ if
$a(y) \sqsubseteq b(y)$ for all $y\in Y$, is a complete lattice.

In the paper we will mostly work with lattices of the form $\monM^Y$
where $\monM$ is a special kind of lattice with a rich algebraic
structure, i.e. an MV-algebra~\cite{Mun:MV}.

\begin{defi}[MV-algebra]
  \label{de:mv}
  An \emph{MV-algebra} is a tuple
  $\monM = (M, \oplus, \zero, \comp{(\cdot)})$ where
  $(M, \oplus, \zero)$ is a commutative monoid and
  $\comp{(\cdot)} : M \to M$ maps each element to its
  \emph{complement}, such that for all $x, y \in M$
  \begin{enumerate}
  \item \label{de:mv:1}
    $\comp{\comp{x}} = x$

  \item \label{de:mv:2}
    $x \oplus \comp{\zero} = \comp{\zero}$
    
  \item \label{de:mv:3}
    $\comp{(\comp{x} \oplus y)} \oplus y = \comp{(\comp{y} \oplus x)} \oplus x$.
  \end{enumerate}
  We denote $\one = \comp{\zero}$, multiplication
  $x \otimes y = \comp{\comp{x} \oplus \comp{y}}$ and subtraction
  $x \ominus y = x \otimes \comp{y}$.
\end{defi}

Note that by using the derived operations, axioms~(2) and~(3) above
can be written as
\begin{enumerate}
  \setcounter{enumi}{1}
\item $x \oplus \one = \one$
\item $(y \ominus x) \oplus x = (x \ominus y) \oplus y$
\end{enumerate}

MV-algebras are endowed with a natural order.

\begin{defi}[natural order]
  \label{de:order}
  Let $\monM = (M, \oplus, \zero, \comp{(\cdot)})$ be an MV-algebra. The
  \emph{natural order} on $\monM$ is defined, for $x, y \in M$, by
  $x \sqsubseteq y$ if $x \oplus z= y$ for some $z \in M$. When $\sqsubseteq$ is
  total $\monM$ is called an \emph{MV-chain}.
\end{defi}

The natural order gives an MV-algebra a lattice structure where
$\bot = \zero$, $\top =\one$, $x \sqcup y = (x \ominus y) \oplus y$
and
$x \sqcap y = \comp{\comp{x} \sqcup \comp{y}} = x \otimes (\comp{x}
\oplus y)$. We call the MV-algebra \emph{complete}, if it is a
complete lattice. This is not true in general, e.g.,
$([0,1] \cap \mathbb{Q}, \leq)$.

\begin{exa}
  \label{ex:mv-algebra}
  A prototypical example of an MV-algebra is
  $([0,1],\oplus,0,\comp{(\cdot)})$ where $x\oplus y = \min\{x+y,1\}$
  and $\comp{x} = 1-x$ for $x,y\in [0,1]$. This means that
  $x\otimes y = \max\{x+y-1,0\}$ and $x\ominus y = \max\{0,x-y\}$
  (truncated subtraction). The operators $\oplus$ and $\otimes$ are
  also known as strong disjunction and conjunction in {\L}ukasiewicz
  logic~\cite{m:lukasiewicz-mv}. The natural order is $\le$ (less
  or equal) on the reals.

  Another example is $(\{0,\dots,k\},\oplus,0,\comp{(\cdot)})$ where
  $n\oplus m = \min\{n+m,k\}$ and $\comp{n} = k-n$ for
  $n,m\in \{0,\dots,k\}$. We are in particular interested in the
    case $k=1$. Both MV-algebras are complete and MV-chains.

  Boolean algebras (with disjunction and complement) also form
  MV-algebras that are complete, but in general not MV-chains.
\end{exa}

MV-algebras are the algebraic semantics of {\L}ukasiewicz logic. They
can be shown to correspond to intervals of the kind $\interval{0}{u}$
in suitable groups, i.e., abelian lattice-ordered groups with a strong
unit~$u$~\cite{Mun:MV}.

We next review some properties of MV-algebras. They are taken from or
easy consequences of properties in~\cite{Mun:MV} and will be used
throughout the paper.

\begin{lemmarep}[properties of MV-algebras]
  \label{le:mvprop}
  Let $\monM = (M, \oplus, \zero, \comp{(\cdot)})$ be an MV-algebra. For
  all $x, y, z \in M$

  \begin{enumerate}[align=left]
    
  \item \label{le:mvprop:1}
    $x \oplus \comp{x} = 1$
    
  \item \label{le:mvprop:2}
    $x \sqsubseteq y$ \quad iff \quad $\comp{x} \oplus y = \one$ \quad iff \quad $x \otimes \comp{y}=\zero$ \quad iff \quad $y = x \oplus (y \ominus x)$

  \item \label{le:mvprop:3}
    $x \sqsubseteq y$ iff $\comp{y} \sqsubseteq \comp{x}$

  \item\label{le:mvprop:4}
    $\oplus$, $\otimes$ are monotone in both arguments, $\ominus$ monotone in the first and antitone in the second argument.

  \item \label{le:mvprop:5}
    if $x \sqsubset y$ then $0 \sqsubset y \ominus x$;

  \item\label{le:mvprop:6}
    $(x \oplus y) \ominus y \sqsubseteq x$
    
  \item\label{le:mvprop:7}
    $z \sqsubseteq x \oplus y$ if and only if $z \ominus x \sqsubseteq
    y$.
    
  \item \label{le:mvprop:8}
    if $x \sqsubset y$ and $z \sqsubseteq \comp{y}$ then $x \oplus z \sqsubset y \oplus z$;

  \item \label{le:mvprop:9}
    $y \sqsubseteq \comp{x}$ if and only if $(x \oplus y) \ominus y = x$;
   
  \item \label{le:mvprop:10}
     $x \ominus (x \ominus y) \sqsubseteq y$ and if $y \sqsubseteq x$ then $x \ominus (x \ominus y) = y$.

   \item \label{le:mvprop:11} Whenever $\monM$ is an MV-chain,
     $x \sqsubset y$ and $0 \sqsubset z$ imply
     $(x\oplus z) \ominus y \sqsubset z$
  \end{enumerate}
\end{lemmarep}

\begin{proof}
  The proof of properties (\ref{le:mvprop:1}), (\ref{le:mvprop:2}),
  (\ref{le:mvprop:3}), (\ref{le:mvprop:4}) can be found directly
  in~\cite{Mun:MV}. For the rest:

  \begin{enumerate}[align=left]
    \setcounter{enumi}{4}
  
  \item Immediate consequence of (\ref{le:mvprop:2}). In fact, given
    $x \sqsubset y$, if we had $y \ominus x = 0$ then by
    (\ref{le:mvprop:2}),
    $y = x \oplus (y \ominus x) = x \oplus 0 = x$, contradicting the hypothesis.
    
  \item Observe that
    $(x \oplus y) \ominus y = \comp{\comp{(x \oplus y)} \oplus y} =
    \comp{(\comp{x} \ominus y) \oplus y} = \comp{(y \ominus \comp{x})
      \oplus \comp{x}} \sqsubseteq \comp{\comp{x}} = x$, where the
    last inequality is motivated by the fact that
    $\comp{x} \sqsubseteq (y \ominus \comp{x}) \oplus \comp{x}$ and
    point (\ref{le:mvprop:3}).

  \item The direction from left to right is an immediate consequence of
    (\ref{le:mvprop:6}). In fact, if $z \sqsubseteq x \oplus y$ then
    $z \ominus x \sqsubseteq (x \oplus y) \ominus x \sqsubseteq y$.

    The other direction goes as follows: if $z\ominus x\sqsubseteq y$,
    then -- by monotonicity (\ref{le:mvprop:4}) --
    $(z\ominus x)\oplus x\sqsubseteq y\oplus x = x\oplus y$. The left
    hand side can be rewritten to $(x\ominus z)\oplus z\sqsupseteq z$.
    
  \item Assume that $x \sqsubset y$ and $z \sqsubseteq \comp{y}$. We
    know, by property (\ref{le:mvprop:4}) that
    $x \oplus z \sqsubseteq y \oplus z$. Assume by contradiction that
    $x \oplus z = y \oplus z$. Then we have
    \begin{align*}
      \comp{x}
      & \sqsubseteq & \mbox{[by properties~(\ref{le:mvprop:3})
                      and~(\ref{le:mvprop:6})]}\\
      & \sqsubseteq \comp{(x \oplus z) \ominus z}
                    & \quad \mbox{[since $x \oplus z =  y \oplus z$]}\\
      & \sqsubseteq \comp{(y \oplus z) \ominus z}
                    & \quad \mbox{[definition of $\ominus$]}\\
      & = (\comp{y} \ominus z) \oplus z 
                    & \quad \mbox{[since $z \sqsubseteq \comp{y}$ and property~(\ref{le:mvprop:2})]}\\
      & = \comp{y}
    \end{align*}
    And with point (\ref{le:mvprop:3}) this is a contradiction.

  \item Assume $y \sqsubseteq \comp{x}$. We know
    $(x \oplus y) \ominus y \sqsubseteq x$. If it were
    $(x \oplus y) \ominus y \sqsubset x$, then
    $((x \oplus y) \ominus y) \oplus y \sqsubset x \oplus y$, with
    (\ref{le:mvprop:8}). Since the left-hand side is equal to
    $(y\ominus (x \oplus y)) \oplus (x \oplus y) \sqsupseteq x \oplus
    y$, this is a contradiction.
    
    For the other direction assume that $(x\oplus y)\ominus y =
    x$. Hence we have
    $x = (x\oplus y)\ominus y = \comp{\comp{(x\oplus y)}\oplus y}$. By
    complementing on both sides we obtain
    $\comp{x} = \comp{(x\oplus y)}\oplus y$ which implies that
    $y\sqsubseteq \comp{x}$.
      
  \item Observe that, by (\ref{le:mvprop:7}), we have
    $\comp{y} \sqsubseteq \comp{x} \oplus (\comp{y} \ominus \comp{x})
    = \comp{x} \oplus (x \ominus y) = \comp{x\ominus (x\ominus
      y)}$. Therefore, by (\ref{le:mvprop:3}),
    $x \ominus (x \ominus y) \sqsubseteq y$, as desired.
      
    For the second part, assume $y \sqsubseteq x$ and thus, by
    (\ref{le:mvprop:3}), $\comp{x} \sqsubseteq \comp{y}$. Using
    (\ref{le:mvprop:2}), we obtain
    $\comp{y} = \comp{x} \oplus (\comp{y} \ominus \comp{x}) = \comp{x}
    \oplus \comp{y \oplus \comp{x}} = \comp{x} \oplus (x \ominus
    y)$. Hence
    $y = \comp{\comp{x} \oplus (x \ominus y)} = x \ominus (x \ominus
    y)$.

  \item We first observe that, given $u,v \in \monM$, $u \sqsubseteq v \oplus (u\ominus
    v)$. This is a direct consequence of axiom~(3) of MV-algebras and
    the definition of natural order.

    Second, in an MV-chain if $u, v \sqsupset 0$, then
    $u\ominus v \sqsubset u$. In fact, if $u \sqsubseteq v$ and thus
    $u\ominus v = 0 \sqsubset u$. If instead, $v\sqsubset u$ we have
    $0 \sqsubset v$ and
    $u\ominus v \sqsubseteq 1\ominus v = \comp{v}$, hence by
    (\ref{le:mvprop:8}) it holds that
    $0 \oplus (u\ominus v) \sqsubset v \oplus (u\ominus v)$. Recalling
    that $v\sqsubset u$ and thus by
    (\ref{le:mvprop:2}),
    $(u\ominus v)\oplus v = u$, we conclude $u\ominus v \sqsubset u$.

    Now
    \begin{align*}
      & (x\oplus z) \ominus y \\
      & \sqsubseteq (x \oplus (y\ominus x) \oplus (z \ominus
        (y\ominus x))) \ominus y & \mbox{[by first obs.\ above]} \\
      & = (y \oplus (z \ominus (y\ominus x))) \ominus y & \mbox{[since
        $x\sqsubseteq y$, by (\ref{le:mvprop:2})]} \\
      & \sqsubseteq z \ominus (y\ominus x) &
        \mbox{[by (\ref{le:mvprop:6})]} \\
      & \sqsubset z & \mbox{[by second obs.\ above, since
          $z\sqsupset 0$} \\
      && \mbox{and $y\ominus x\sqsupset 0$ by
         (\ref{le:mvprop:5})]} \tag*{\qedhere}
    \end{align*}
    \qedhere
  \end{enumerate}
\end{proof}

Note that we adhere to the following convention: whenever brackets are
missing, we always assume that we associate from left to right.  So
$a\oplus b \ominus c$ should be read as $(a\oplus b)\ominus c$ and not
as $a\oplus (b\ominus c)$, which is in general different.

\section{Non-expansive functions and their approximations}
\label{se:nonexpansive-approximation}

As mentioned in the introduction, our interest is for fixpoints of
monotone functions $f: \monM^Y\to \monM^Y$, where $\monM$ is an
MV-chain and $Y$ is a finite set. We will see that for non-expansive
functions we can over-approximate the sets of points in which a given
$a \in \monM^Y$ can be increased in a way that is preserved by the
application of $f$. This will be the core of the proof rules outlined
earlier.%

\subsection{Non-expansive functions on MV-algebras.}

For defining non-expansiveness it is convenient to introduce a norm, which can be seen as an adaptation of the standard $l_\infty$ norm.

\begin{defi}[norm]
  \label{de:norm}
  Let $\monM$ be an MV-chain and let $Y$ be a finite set. Given
  $a \in \monM^Y$ we define its \emph{norm} as
  $\norm{a} = \max \{ a(y) \mid y \in Y\}$.
\end{defi}

Given a finite set $Y$ we extend $\oplus$ and $\otimes$ to $\monM^Y$
pointwise. E.g. if $a, b \in \monM^Y$, we write $a \oplus b$ for the
function defined by $(a \oplus b)(y) = a(y) \oplus b(y)$ for all
$y \in Y$.
Given $Y' \subseteq Y$ and $\delta \in \monM$, we write $\delta_{Y'}$
for the function defined by $\delta_{Y'}(y) = \delta$ if $y \in Y'$
and $\delta_{Y'}(y) = \zero$, otherwise.
Whenever this does not generate confusion,
we write $\delta$ instead of $\delta_Y$.

As shown in the lemma below, $\norm{\cdot}$ has the standard properties of a norm.
Moreover, it is clearly monotone, i.e., if $a \sqsubseteq b$
then $\norm{a} \sqsubseteq \norm{b}$.

\begin{lemmarep}[properties of the norm]
  \label{le:norm}
  Let $\monM$ be an MV-chain and let $Y$ be a finite set. Then $\norm{\cdot} : \monM^Y \to \monM$ satisfies, for all $a, b \in \monM^Y$, $\delta \in \monM$
  \begin{enumerate}
  \item \label{le:norm:1}
    $\norm{a \oplus b} \sqsubseteq \norm{a} \oplus \norm{b}$,

  \item \label{le:norm:2}
    $\norm{\delta \otimes a} = \delta \otimes \norm{a}$ and

  \item \label{le:norm:3}
    $\norm{a} = \zero$ implies that $a$ is the constant $\zero$.
  \end{enumerate}
\end{lemmarep}

\begin{proof}
  Concerning (1), let $\norm{a \oplus b}$ be realised on some element
  $y \in Y$, i.e., $\norm{a \oplus b} = a(y) \oplus b(y)$. Since
  $a(y) \sqsubseteq \norm{a}$ and $b(y) \sqsubseteq \norm{b}$, by
  monotonicity of $\oplus$ we deduce that
  $\norm{a \oplus b} \sqsubseteq \norm{a} \oplus \norm{b}$.

  Concerning (2), note that
  \begin{align*}
    \norm{\delta \otimes a}
    & = \max \{ \comp{\comp{\delta} \oplus
      \comp{a(y)}} \mid y \in Y\}\\
    & = \comp{\min \{ \comp{\delta} \oplus \comp{a(y)} \mid y \in Y \}}\\
    & = \comp{\comp{\delta} \oplus \min \{\comp{a(y)} \mid y \in Y \}}\\
    & = \comp{\comp{\delta} \oplus \comp{\max \{ a(y) \mid y \in Y \}}}\\
    & = \comp{\comp{\delta} \oplus \comp{\norm{a}}}\\
    & = \delta \otimes \norm{a}\
  \end{align*}

  Finally, point (3) is straightforward, since $\zero$ is the bottom of $\monM$.
\end{proof}

We next introduce non-expansiveness. Despite the fact that we will
eventually be interested in endo-functions $f : \monM^Y \to \monM^Y$, in
order to allow for a compositional reasoning we work with functions
where domain and codomain can be different.

\begin{defi}[non-expansiveness]
  \label{de:non-expansiveness}
  Let $f: \monM^Y\to \monM^Z$ be a function, where $\monM$ is an MV-chain and $Y, Z$ are
  finite sets. We say that it is \emph{non-expansive} if for all
  $a, b \in \monM^Y$ it holds that
  $\norm{f(b) \ominus f(a)} \sqsubseteq \norm{b \ominus a}$.
\end{defi}

Note that $(a,b)\mapsto\norm{a\ominus b}$ is the supremum lifting of a
directed version of Chang's distance~\cite{Mun:MV}.
It is easy to see that all non-expansive functions on MV-chains are
monotone (see Lemma~\ref{le:non-exp-monotonic} in the
appendix). Moreover, when $\monM = \{0,1\}$, i.e., $\monM$ is the
two-point boolean algebra, the two notions coincide.

\begin{toappendix}
\begin{lemma}[non-expansiveness implies monotonicity]
  \label{le:non-exp-monotonic}
  Let $\monM$ is an MV-chain and let $Y, Z$ be finite sets.
  Every non-expansive function $f: \monM^Y\to \monM^Z$ is monotone.
\end{lemma}

\begin{proof}
  Let $a, b \in \monM^Y$ be such that $a \sqsubseteq b$. Therefore, by
  Lemma~\ref{le:mvprop}(\ref{le:mvprop:2}), $a(y) \ominus b(y) = 0$
  for all $y \in Y$, hence $a \ominus b = 0$. Thus
  $\norm{f(a) \ominus f(b)} \sqsubseteq \norm{a \ominus b} = 0$. In
  turn this implies that for all $z \in Z$,
  $f(a)(z) \ominus f(b)(z) = 0$. Hence
  Lemma~\ref{le:mvprop}(\ref{le:mvprop:2}), allows us to conclude
  $f(a)(z) \sqsubseteq f(b)(z)$ for all $z \in Z$, i.e.,
  $f(a) \sqsubseteq f(b)$, as desired.
\end{proof}

The next lemma provides a useful equivalent characterisation of
non-expansiveness.

\begin{lemma}[characterisation of non-expansiveness]
  \label{le:non-expansiveness}
  Let $f: \monM^Y\to \monM^Z$ be a monotone function, where $\monM$ is
  an MV-chain and $Y, Z$ are finite sets. Then $f$ is non-expansive
  iff for all $a \in \monM^Y$, $\theta \in \monM$ and $z \in Z$ it
  holds
  $f(a \oplus \theta)(z) \ominus f(a)(z) \sqsubseteq
  \theta$.
\end{lemma}

\begin{proof}
  Let $f$ be non-expansive and let $a \in \monM^Y$ and
  $\theta \in \monM$. We have that for all $z \in Z$
    \begin{align*}
      & f(a \oplus \theta)(z) \ominus f(a)(z) \sqsubseteq
      &\\
      & \quad \sqsubseteq \norm{f(a \oplus \theta) \ominus f(a)}
      & \mbox{[by definition of norm]}\\
      & \quad \sqsubseteq \norm{(a \oplus \theta) \ominus a}
      & \mbox{[by hypothesis]}\\
      & \quad \sqsubseteq \norm{\lambda y. \theta}
      &  \mbox{[by Lemma~\ref{le:mvprop}(\ref{le:mvprop:6}) and monotonicity of norm]}\\      
      & \quad = \theta
      &  \mbox{[by definition of norm]}
    \end{align*}

    \medskip

    Conversely, assume that for all $a \in \monM^Y$,
    $\theta \in \monM$ and $z \in Z$ it holds
    $f(a \oplus \theta)(z) \ominus f(a)(z) \sqsubseteq
    \theta$. For $a, b \in \monM^Y$, first observe that for all
    $y \in Y$ it holds
    $b(y) \ominus a(y) \sqsubseteq \norm{b \ominus a}$, hence, if we
    let $\theta = \norm{b \ominus a}$, we have
    $b \sqsubseteq a \oplus \theta$ and thus, by monotonicity,
    $f(b) \ominus f(a) \sqsubseteq f(a \oplus \theta) \ominus
    f(a)$. Thus

    \begin{align*}
      & \norm{f(b) \ominus f(a)} \sqsubseteq
      &\\
      & \quad \sqsubseteq \norm{f(a + \theta) \ominus f(a)} =\\
      & \quad \mbox{[by the observation above and monotonicity of norm]}\\
      & \quad = \max \{ f(a + \theta)(z) \ominus f(a)(z) | z \in Z \}
      & \mbox{[by definition of norm]}\\
      & \quad \sqsubseteq \theta
      & \mbox{[by hypothesis]}\\
      & \quad = \norm{b \ominus a}
      &  \mbox{[by the choice of $\theta$]} \tag*{\qedhere}
    \end{align*}
\end{proof}
\end{toappendix}

\begin{toappendix}

  \begin{lemma}[composing non-expansive functions]
    \label{le:comp-non-expansive}
    Let $\monM$ be an MV-chain and let $Y, W, Z$ be finite sets. If
    $g : \mathbb{M}^Y \to \mathbb{M}^W$ and
    $h : \mathbb{M}^W \to \mathbb{M}^Z$ are non-expansive then
    $h \circ g : \mathbb{M}^Y \to \mathbb{M}^Z$ is non-expansive.
  \end{lemma}

\begin{proof}
  Straightforward. We have for any $a, b\in \mathbb{M}^Y$ that
  \begin{align*}
    & \norm{ h(g(b)) \ominus h(g(a))} \sqsubseteq\\
    & \quad \sqsubseteq \norm{ g(b) \ominus g(a)}
    & \mbox{[by non-expansiveness of $h$]}\\
    & \quad \sqsubseteq \norm{ b \ominus a} & \mbox{[by
      non-expansiveness of $g$]} \tag*{\qedhere}
  \end{align*}
\end{proof}
\end{toappendix}

\subsection{Approximating the propagation of increases.}

Let $f : \monM^Y \to \monM^Z$ be a monotone function and take
$a, b \in \monM^Y$ with $a \sqsubseteq b$. We are interested in
the difference $b(y) \ominus a(y)$ for some $y\in Y$ and on how the
application of $f$ ``propagates'' this difference.
The reason is that, understanding that no increase can be
propagated will be crucial to establish when a fixpoint of a
non-expansive function $f$ is actually the largest one, and, more
generally, when a (pre-)fixpoint of $f$ is above the largest fixpoint.

In order to formalise the above intuition, we rely on tools from
abstract interpretation. In particular, the following pair of
functions, which, under a suitable condition, form a Galois
connection, will play a major role. For this purpose we fix
$a\in\monM^Y$, $\delta\in\monM$. The left adjoint
$\alpha_{a,\delta}$ takes as input a set $Y'\subseteq Y$ and, for $y\in Y'$,
it increases the values $a(y)$ by $\delta$, while the right adjoint
$\gamma_{a,\delta}$ takes as input a function $b\in \monM^Y$,
$b \in \interval{a}{a \oplus \delta}$ and checks for which parameters
$y\in Y$ the value $b(y)$ exceeds $a(y)$ by $\delta$.

We also define $\Ybot{Y}{a}$, the subset of elements in $Y$ where
$a(y)$ is not $1$ and thus there is a potential to increase, and
$\delta_a$, which gives us the least of such increases (i.e., the largest increase that can be used on all elements in $\Ybot{Y}{a}$ without ``overflowing'').

\begin{defi}[functions to sets, and vice versa]
  \label{de:galois}
  Let $\monM$ be an MV-algebra and let $Y$ be a finite set. Define the
  set $\Ybot{Y}{a} = \{ y\in Y \mid a(y) \neq 1\}$ (support of
  $\comp{a}$) and
  $\delta_a = \min\{ \comp{a(y)} \mid y \in \Ybot{Y}{a}\}$ with
  $\min\emptyset = 1$.
  
  For $0\sqsubset\delta \in \monM$ we consider the functions
  $\alpha_{a,\delta} : \Pow{\Ybot{Y}{a}} \to \interval{a}{a \oplus
    \delta}$ and
  $\gamma_{a,\delta} : \interval{a}{a \oplus \delta} \to
  \Pow{\Ybot{Y}{a}}$, defined, for $Y' \in \Pow{\Ybot{Y}{a}}$ and
  $b \in \interval{a}{a \oplus \delta}$, by
  \[ \alpha_{a,\delta}(Y') = a \oplus \delta_{Y'} \qquad
    \gamma_{a,\delta}(b) = \{ y\in \Ybot{Y}{a} \mid b(y) \ominus a(y)
    \sqsupseteq \delta \}. \]
\end{defi}

\begin{toappendix}
  \begin{lemma}[well-definedness]
    The functions $\alpha_{a,\delta}$, $\gamma_{a,\delta}$ from
    Def.~\ref{de:galois} are well-defined and monotone.
\end{lemma}

\begin{proof}
  The involved functions $\alpha_{a,\delta}$ and $\gamma_{a,\delta}$
  are well-defined. In fact, for $Y' \subseteq \Ybot{Y}{a}$, clearly
  $\alpha_{a,\delta} = a \oplus \delta_{Y'} \in \interval{a}{a \oplus
    \delta}$. Moreover, for $b \in \interval{a}{a \oplus \delta}$ we
  have $\gamma_{a,\delta}(b) \subseteq \Ybot{Y}{a}$. In fact, if
  $y \not\in \Ybot{Y}{a}$ then $a(y) = \one$, hence $b(y) = \one$ and
  thus $b(y) \ominus a(y) = \zero \not\sqsupseteq \delta$, and thus
  $y \not\in \gamma_{a,\delta}(b)$. Moreover, they are clearly
  monotone.
\end{proof}
\end{toappendix}

When $\delta$ is sufficiently small, the pair
$\langle \alpha_{a,\delta}, \gamma_{a,\delta} \rangle$ is a Galois
connection.

\begin{lemmarep}[Galois connection]
  \label{le:galois} 
  Let $\monM$ be an MV-algebra
  and $Y$ be a finite set.
  For $\zero \neq \delta \sqsubseteq \delta_a$, the pair
  $\langle \alpha_{a,\delta}, \gamma_{a,\delta} \rangle : \Pow{\Ybot{Y}{a}}
  \to \interval{a}{a \oplus \delta}$ is a Galois connection.
  \begin{center}
    \begin{tikzpicture}[->]
      \node (y) {$\Pow{\Ybot{Y}{a}}$}; \node [right=of y] (b)
      {$\interval{a}{a \oplus \delta}$}; \draw [->,thick,bend left] (y) to
      node [above]{$\alpha_{a,\delta}$} (b); \draw [->,thick,bend left]
      (b) to node [below]{$\gamma_{a,\delta}$} (y);
    \end{tikzpicture}
  \end{center}
\end{lemmarep}

\begin{proof}
  For all $Y' \in \Pow{\Ybot{Y}{a}}$  it holds
\begin{center}
  $\gamma_{a,\delta}(\alpha_{a,\delta} (Y'))  = \gamma_{a,\delta}(a \oplus
  \delta_{Y'}) = Y'$.
  \end{center}
  In fact, for all $y \in Y'$, $(a\oplus\delta_{Y'})(y)=a(y) \oplus \delta$. Moreover, and
  by the choice of $\delta$ and definition of $\Ybot{Y}{a}$, we have
  $\delta \sqsubseteq \delta_a \sqsubseteq \comp{a(y)}$, by
  Lemma~\ref{le:mvprop}(\ref{le:mvprop:9}), we have
  $(a\oplus\delta_{Y'})(y) \ominus a(y) = \delta$ hence
  $y \in \gamma_{a,\delta}(\alpha_{a,\delta} (Y'))$. Conversely, if
  $y \not\in Y'$, then $(a\oplus\delta_{Y'})(y)=a(y)$, and thus
  $(a\oplus\delta_{Y'})(y) \ominus a(y) = \zero \not\sqsupseteq \delta$.

  Moreover, for all
  $b \in \interval{a}{a \oplus \delta}$ we have
  \begin{center}
    $\alpha_{a,\delta} (\gamma_{a,\delta}(b)) = a \oplus \delta_{\gamma_{a,\delta}(b)} \sqsubseteq b$
  \end{center}
  In fact, for all $y \in Y$, if $y \in \gamma_{a,\delta}(b)$, i.e., 
  $\delta \sqsubseteq b(y) \ominus a(y)$ then
  $(a \oplus \delta_{\gamma_{a,\delta}(b)})(y) = a(y) \oplus \delta \sqsubseteq a(y) \oplus (b(y) \ominus
  a(y)) = b(y)$, by
  Lemma~\ref{le:mvprop}(\ref{le:mvprop:2}). If instead, $y \not\in \gamma_{a,\delta}(b)$, then
  $(a \oplus \delta_{\gamma_{a,\delta}(b)}(b))(y) = a(y) \sqsubseteq b(y)$.
\end{proof}

Observe that differently from what normally happens in abstract
interpretation, the component $\alpha$ of the Galois connection, i.e.,
the left adjoint, transforms abstract values (sets) into concrete ones
(functions) and thus it plays the role of a concretisation function.

\begin{exa}
  \label{ex:running-1}
  We illustrate the definitions with small examples whose sole
  purpose is to get a better intuition. (See
  Fig.~\ref{fig:galois-functions} for a visual representation.)
  Consider the MV-chain $\monM = [0,1]$, a set
  $Y = \{y_1,y_2,y_3,y_4\}$ and a function $a\colon Y\to [0,1]$
  with $a(y_1) = 0.2$, $a(y_2) = 0.4$, $a(y_3) = 0.9$, $a(y_4) =
  1$. In this case $\Ybot{Y}{a} = \{y_1,y_2,y_3\}$ and $\delta_a = 0.1$.

  Choose $\delta = 0.1$ and $Y' = \{y_1,y_3\}$. Then
  $\alpha_{a,\delta}(Y')$ is a function that maps $y_1 \mapsto 0.3$,
  $y_2 \mapsto 0.4$, $y_3 \mapsto 1$, $y_4 \mapsto 1$.

  We keep $\delta = 0.1$ and consider a function $b\colon Y\to [0,1]$ with $b(y_1) = 0.3$, $b(y_2) = 0.45$, $b(y_3) = b(y_4) =
  1$. Then $\gamma_{a,\delta}(b) = \{y_1,y_3\}$. 
  
  \begin{figure}
    \centering
    ${\dgreen a}\ =\ \raisebox{-30pt}{\scalebox{0.65}{\input{galois-a.tex}}}$ 

    $\alpha_{a,\delta}\colon$
    $Y'=\{y_1,y_3\}\quad \mapsto\quad 
    \raisebox{-30pt}{\scalebox{0.65}{\input{galois-a-plus-theta.tex}}}$
    
    $\gamma_{a,\delta}\colon$
    ${\blue b} = \raisebox{-30pt}{\scalebox{0.65}{\input{galois-b.tex}}} \quad
    \mapsto\quad  Y'=\{y_1,y_3\}$
    
    \caption{Visual representation of $\alpha_{a,\delta}$ and
      $\gamma_{a,\delta}$}
    \label{fig:galois-functions}
  \end{figure}
\end{exa}

Whenever $f$ is non-expansive, it is easy to see that it restricts to
a function
$f : \interval{a}{a \oplus \delta} \to \interval{f(a)}{f(a) \oplus \delta}$
for all $\delta \in \monM$.

\begin{toappendix}
\begin{lemma}[restricting non-expansive functions to intervals]
  \label{le:restriction}
  Let $\monM$ be an MV-chain, let $Y, Z$ be finite sets
  $f: \monM^Y\to \monM^Z$ be a non-expansive function. Then $f$
  restricts to a function
  $f_{a,\delta} : \interval{a}{a \oplus \delta} \to
  \interval{f(a)}{f(a)\oplus \delta}$, defined by
  $f_{a,\delta}(b) = f(b)$.
\end{lemma}

\begin{proof}
  Given $b \in \interval{a}{a \oplus \delta}$, by monotonicity of $f$
  we have that $f(a) \sqsubseteq f(b)$. Moreover,
  $f(b) \sqsubseteq f(a \oplus \delta) \sqsubseteq f(a) \oplus
  \delta$, where the last passage is motivated by
  Lemma~\ref{le:non-expansiveness}.
\end{proof}

In the following we will simply write $f$ instead of $f_{a,\delta}$.

\end{toappendix}

As mentioned before, a crucial result shows that for all non-expansive
functions, under the assumption that $Y,Z$ are finite and the order on
$\monM$ is total, we can suitably approximate the propagation of
increases. In order to state this result, a useful tool is a notion
of approximation of a function.

\begin{defi}[$(\delta,a)$-approximation]
  \label{de:approximation}
  Let $\monM$ be an MV-chain, let $Y$, $Z$ be finite sets and let
  $f: \monM^Y \to \monM^Z$ be a non-expansive function. For
  $a \in \monM^Y$ and 
  any $\delta\in \monM$ 
   we define
  $f_{a,\delta}^{\#} : \Pow{\Ybot{Y}{a}} \to \Pow{\Ybot{Z}{f(a)}}$ as
  $f_{a,\delta}^{\#} = \gamma_{f(a), \delta} \circ f \circ
  \alpha_{a,\delta}$.
\end{defi}

Given $Y' \subseteq \Ybot{Y}{a}$, its image
$f_{a,\delta}^{\#}(Y') \subseteq \Ybot{Z}{f(a)}$ is the set of points
$z \in \Ybot{Z}{f(a)}$ such that
$\delta \sqsubseteq f(a \oplus \delta_{Y'})(z) \ominus f(a)(z)$, i.e.,
the points to which $f$ propagates an increase of the function $a$
with value $\delta$ on the subset $Y'$.

\begin{exa}
  \label{ex:running-2}
  We continue with Example~\ref{ex:running-1} and consider the
  function $f\colon [0,1]^Y\to [0,1]^Y$ with $f(b) = b \ominus 0.3$
  for every $b \in [0,1]^Y$,
  which can easily be seen to be non-expansive. We again consider
  $a\colon Y\to [0,1]$ and $\delta=0.1$ as in
  Example~\ref{ex:running-1}, and $Y' = \{y_1,y_2,y_3\}$.
  The maps $a$, $\alpha_{a,\delta}(Y')$,
  $f(a)$ and $f(\alpha_{a,\delta}(Y'))$ are given in the table below
  and we obtain
  $f_{a,\delta}^\#(Y') =
  \gamma_{f(a),\delta}(f(\alpha_{a,\delta}(Y'))) = \{y_2,y_3\}$, that
  is only the increase at $y_2$ and $y_3$ can be propagated, while the
  value of $y_1$ is too low and $y_4$ is not even contained in
  $\Ybot{Y}{a}$ (the domain of $f_{a,\delta}^\#$), since its value is
  already $1.0$ and there is no slack left. That is, we obtain those
  elements of $Y$ for which the last two lines in the table below
  differ by $0.1$.

  \begin{center}
    \begin{tabular}{l|c|c|c|c|}
      & $y_1$ & $y_2$ & $y_3$ & $y_4$ \\ \hline
      $a$ & $0.2$ & $0.4$ & $0.9$ & $1.0$ \\
      $\alpha_{a,\delta}(Y')$ & $0.3$ & $0.5$ & $1.0$ & $1.0$ \\
      $f(a)$ & $0.0$ & $0.1$ & $0.6$ & $0.7$ \\
      $f(\alpha_{a,\delta}(Y'))$ & $0.0$ & $0.2$ & $0.7$ & $0.7$ \\ \hline
    \end{tabular}
  \end{center}

  In general we have $f_{a,\delta}^\#(Y') = Y'\cap \{y_2,y_3\}$ if
  $\delta\le \delta_a = 0.1$, $f_{a,\delta}^\#(Y') = Y'\cap \{y_2\}$
  if $0.1 < \delta\le 0.6$ and $f_{a,\delta}^\#(Y') = \emptyset$ if
  $0.6 < \delta$.
\end{exa}

We now show that $f_{a,\delta}^{\#}$ is antitone in the
parameter $\delta$, a non-trivial result.

\begin{lemmarep}[antitonicity]
  \label{le:approximation-monotonic}
  Let $\monM$ be an MV-chain, let $Y$, $Z$ be finite sets, let
  $f: \monM^Y \to \monM^Z$ be a non-expansive function and let
  $a \in \monM^Y$. For $\theta, \delta \in \monM$, if 
  $\theta \sqsubseteq \delta$
   then
  $f_{a,\delta}^{\#} \subseteq f_{a,\theta}^{\#}$.
\end{lemmarep}

\begin{proof}
  Let $Y' \subseteq \Ybot{Y}{a}$ and let us prove that
  $f_{a,\delta}^{\#}(Y') \subseteq f_{a,\theta}^{\#}(Y')$. Take
  $z \in f_{a,\delta}^{\#}(Y')$. This means that
  $\delta \sqsubseteq f(a \oplus \delta_{Y'})(z) \ominus f(a)(z)$.
  
  We have
  \begin{align*}
    & \delta \sqsubseteq f(a \oplus \delta_{Y'})(z) \ominus f(a)(z) \\
    & \quad \mbox{[by hypothesis]}\\
    & = f(a \oplus \theta_{Y'} \oplus (\delta \ominus \theta)_{Y'})(z) \ominus f(a)(z) \\
    & = f(a \oplus \theta_{Y'} \oplus (\delta \ominus \theta)_{Y'})(z) \ominus
    f(a \oplus \theta_{Y'})(z) \oplus
    f(a \oplus \theta_{Y'})(z) 
    \ominus f(a)(z) \\
    & \sqsubseteq \norm{f(a \oplus \theta_{Y'} \oplus (\delta \ominus \theta)_{Y'}) \ominus f(a \oplus \theta_{Y'})}
    \oplus
    f(a \oplus \theta_{Y'})(z) \ominus f(a)(z) \\
    & \quad \mbox{[by definition of norm and monotonicity of $\oplus$]}\\
    & \sqsubseteq \norm{a \oplus \theta_{Y'} \oplus (\delta \ominus \theta)_{Y'} \ominus (a \oplus \theta_{Y'})}
    \oplus
    f(a \oplus \theta_{Y'})(z) \ominus f(a)(z) \\
    & \quad \mbox{[by non-expansiveness of $f$ and monotonicity of $\oplus$]}\\
    & \sqsubseteq \norm{(\delta \ominus \theta)_{Y'}} \oplus f(a \oplus \theta_{Y'})(z) \ominus f(a)(z)\\
    & \sqsubseteq (\delta \ominus \theta) \oplus f(a \oplus \theta_{Y'})(z) \ominus f(a)(z)\\
    & \quad \mbox{[by definition of norm]}
  \end{align*}
  If we subtract $\delta \ominus \theta$ on both sides, we get
  $\delta \ominus (\delta \ominus \theta) \sqsubseteq f(a \oplus
  \theta_{Y'})(z) \ominus f(a)(z)$, and, as above, since, by
  Lemma~\ref{le:mvprop}(\ref{le:mvprop:10}),
  $\delta \ominus (\delta \ominus \theta) = \theta$ we conclude
  \begin{center}
    $\theta \sqsubseteq f(a \oplus \theta_{Y'})(z) \ominus f(a)(z)$
  \end{center}
  which means $z \in f_{a,\theta}^{\#}(Y')$.
\end{proof}

Since $f_{a,\delta}^{\#}$ increases when $\delta$ decreases and there
are only finitely many such functions, there must be a value
$\inc{a}{f}$ such that all functions $f_{a,\delta}^{\#}$ for
$0\sqsubset \delta\sqsubseteq\inc{a}{f}$ are equal. The resulting function will be the approximation of interest. 

We next show how $\inc{a}{f}$ can be determined. We start by observing
that for each $z \in \Ybot{Z}{f(a)}$ and $Y' \subseteq \Ybot{Y}{a}$ there is a largest increase
$\theta$ such that $z \in f_{a,\theta}^{\#}(Y')$.

\begin{lemmarep}[largest increase for a point]
  \label{le:approx-max-theta}
  Let $\monM$ be a complete MV-chain, let $Y$, $Z$ be finite sets, let
  $f: \monM^Y \to \monM^Z$ be a non-expansive function and fix
  $a \in \monM^Y$. For all $z \in \Ybot{Z}{f(a)}$ and $Y' \subseteq \Ybot{Y}{a}$ the set
  $\{ \theta \in \monM \mid z \in f_{a,\theta}^\#(Y') \}$
  has a maximum, that we denote by $\inc{a}{f}(Y',z)$.
\end{lemmarep}

\begin{proof}
  Let
  $V = \{ \theta  \in \monM \mid z \in f_{a,\theta}^\#(Y')
  \}$. Expanding the definition we have that
  \begin{center}
    $V= \{ \theta  \in \monM \mid \theta \sqsubseteq f(a
    \oplus \theta_{Y'})(z) \ominus f(a)(z)\}$.
  \end{center}
  If we let $\eta = \sup V$, for all $\theta \in V$, since
  $\theta_{Y'} \sqsubseteq \eta_{Y'}$, clearly, by monotonicity
  \begin{center}
     $\theta \sqsubseteq f(a\oplus \eta_{Y'})(z) \ominus f(a)(z)$
   \end{center}
   and therefore, by definition of supremum,
   $\eta \sqsubseteq f(a\oplus \eta_{Y'})(z) \ominus f(a)(z)$, i.e.,
   $\eta \in V$ is a maximum, as desired.
\end{proof}

We can then provide an explicit definition of $\inc{a}{f}$ and of the
approximation of a function.

\begin{lemmarep}[$a$-approximation for a function]
  \label{le:approximation-2}
  Let $\monM$ be a complete MV-chain, let $Y, Z$ be finite sets and let
  $f: \monM^Y\to \monM^Z$ be a non-expansive function. Let
  \begin{center}
    $\inc{a}{f} = \min \{ \inc{a}{f}(Y',z) \mid Y' \subseteq
    \Ybot{Y}{a}\ \land\ z \in \Ybot{Z}{f(a)}\ \land\ \inc{a}{f}(Y',z)
    \neq 0\} \cup \{ \delta_a\}$.
  \end{center}
  Then for all $0 \neq \delta \sqsubseteq \inc{a}{f}$ it holds that
  $f_{a,\delta}^\# = f_{a,\inc{a}{f}}^\#$.

  The function $f_{a,\inc{a}{f}}^\#$ is called the
  \emph{$a$-approximation} of $f$ and it is denoted by $f_a^\#$.
\end{lemmarep}

\begin{proof}
  Since $\delta \sqsubseteq \inc{a}{f}$, by
  Lemma~\ref{le:approximation-monotonic} we have
  $f_{a,\delta}^\#\supseteq f_{a,\inc{a}{f}}^\#$. For the other
  inclusion let $Y' \subseteq \Ybot{Y}{a}$. We have
  \[ f_{a,\delta}^\#(Y') = \{z\in \Ybot{Z}{f(a)} \mid f(a\oplus
    \delta_{Y'})(z)\ominus f(a)(z) \sqsupseteq \delta\} \] by
  definition. Assume that there exists $z\in f_{a,\delta}^\#(Y')$
  where
  $f(a\oplus (\inc{a}{f})_{Y'})(z)\ominus f(a)(z) \not\sqsupseteq
  \inc{a}{f}$. But this is a contradiction, since $\inc{a}{f}$ is the
  minimum of all such non-zero values.
\end{proof}

In the following, we show that indeed, for all non-expansive functions, the
$a$-approximation properly approximates the propagation of increases.
Given an MV-chain $\monM$ and a finite set $Y$, we first observe that
each function $b \in \monM^Y$ can be expressed as a suitable sum of
functions of the shape $\delta_{Y'}$.

\begin{lemmarep}[standard form]
  \label{le:standard-form}
  Let $\monM$ be an MV-chain and let $Y$ be a finite set. Then for any
  $b \in \monM^Y$ there are $Y_1, \ldots, Y_n \subseteq Y$ with
  $Y_{i+1} \subseteq Y_{i}$ for $i \in \{ 1, \ldots, n-1\}$ and
  $\delta^i \in \monM$, $0 \neq \delta^i \sqsubseteq \comp{\bigoplus_{j=1}^{i-1} \delta^j}$ for
  $i \in \{ 1, \ldots, n\}$ such that
  \begin{center}
    $b = \bigoplus_{i=1}^n \delta^i_{Y_i}$
    \quad and \quad $\norm{b} = \bigoplus_{i=1}^n \delta^i$.
  \end{center}
  where we assume that an empty sum evaluates to $\zero$.
\end{lemmarep}

\begin{proof}
  Given $b \in \monM^Y$, consider $V = \{ b(y) \mid y \in Y \}$. If
  $V$ is empty, then $Y$ is empty and thus $b = 1_Y$, i.e., we can
  take $n=1$, $\delta^1 = 1$ and $Y_1 = Y$. Otherwise, if
  $Y\neq \emptyset$, then $V$ is a finite non-empty set. Let
  $V = \{ v_1, \ldots, v_n \}$, with $v_i \sqsubseteq v_{i+1}$ for
  $i \in \{ 1, \ldots, n-1\}$. For $i \in \{ 1, \ldots, n\}$ define
  $Y_i = \{ y \in Y \mid v_i \sqsubseteq b(y) \}$. Clearly,
  $Y_1 \supseteq Y_2 \supseteq \ldots \supseteq Y_n$. Moreover let
  $\delta^1 = v_1$ and $\delta^{i+1} = v_{i+1} \ominus v_{i}$ for
  $i \in \{ 1, \ldots, n-1\}$.

  Observe that for each $i$, we have
  $v_i = \bigoplus_{j=1}^i \delta^i$, as it can easily shown by
  induction. Hence
  $\delta^{i+1} = v_{i+1} \ominus v_i = v_{i+1} \ominus
  \bigoplus_{j=1}^i \delta^i \sqsubseteq 1 \ominus \bigoplus_{j=1}^i
  \delta^i = \comp{\bigoplus_{j=1}^i \delta^i}$.

  We now show that $b =\bigoplus_{i=1}^n \delta^i_{Y_i}$ by induction on
  $n$.
  \begin{itemize}
  \item 
    If $n=1$ then $V=\{v_1\}$ and thus $b$ is a constant function
    $b(y) = v_1$ for all $y \in Y$. Hence $Y_1 = Y$ and thus
    $b = \delta^1_Y = \delta^1_{Y_1}$, as desired.

  \item 
    If $n>1$, let $b' \in \monM^Y$ defined by $b'(y) = b(y)$ for
    $y \in Y \backslash Y_n$ and $b'(y) = v_{n-1}$ for $y \in Y_n$. Note
    that $\{ b'(y) \mid y \in Y \} = \{ v_1, \ldots, v_{n-1} \}$. Hence,
    by inductive hypothesis, $b' = \bigoplus_{i=1}^{n-1}
    \delta^i_{Y_i}$. Moreover, $b'(y) = b \oplus \delta^n_{Y_n}$, and thus we
    conclude.
  \end{itemize}

  Finally observe that the statement requires $\delta^i \neq 0$ for
  all $i$. We can enjoy this property by just omitting the first
  summand when $v_1 = 0$.
\end{proof}

The above characterisation allows us to show a technical property of the functions in the interval $\interval{a}{a \oplus \delta}$ of interest.

\begin{lemmarep}
  \label{le:approx-fb}
  Let $\monM$ be an MV-chain, let $Y$, $Z$ be finite sets and let
  $f: \monM^Y \to \monM^Z$ be a non-expansive function. Let
  $a \in \monM^Y$.
  For $b \in \interval{a}{a \oplus \delta}$, let 
  $b \ominus a = \bigoplus_{i=1}^n \delta^i_{Y_i}$ be a standard form
  for $b \ominus a$.
  If $\gamma_{f(a),\delta}(f(b)) \neq \emptyset$
  then $Y_n = \gamma_{a,\delta}(b)$ and
  $\gamma_{f(a),\delta}(f(b)) \subseteq
  f_{a,\delta^n}^\#(Y_n)$.
\end{lemmarep}

\begin{proof}
  By hypothesis $\gamma_{f(a),\delta}(f(b)) \neq \emptyset$. Let
  $z \in \gamma_{f(a),\delta}(f(b))$. This means that
  $\delta \sqsubseteq f(b)(z) \ominus f(a)(z)$.
  First observe that
    \begin{align*}
      & \delta \sqsubseteq f(b)(z) \ominus f(a)(z)
      & \mbox{[by hypothesis]}\\
      & \quad \sqsubseteq \norm{f(b) \ominus f(a)}
      & \mbox{[by definition of norm]}\\
      & \quad \sqsubseteq \norm{b \ominus a}
      & \mbox{[by non-expansiveness of $f$]}\\
      & \quad \sqsubseteq \delta
      & \mbox{[since $b \in \interval{a}{a\oplus \delta}$]}
    \end{align*} Hence
    \begin{center}
      $\norm{f(b) \ominus f(a)} = \delta = \norm{b \ominus a} =  \bigoplus_{i=1}^n \delta^i$.
    \end{center}

    Also observe that, since $\delta^n \neq 0$, we have
    $(b \ominus a) (z) = \delta$ iff $z\in Y_n$. In fact, if
    $z \in Y_n$ then $z \in Y_i$ for all $i \in \{1, \ldots, n\}$ and
    thus
    $(b \ominus a)(z) = \bigoplus_{i=1}^n \delta^i_{Y_i}(z) =
    \bigoplus_{i=1}^n \delta^i =\delta$. Conversely, if
    $z \not\in Y_n$, then
    $(b \ominus a) (z) \sqsubseteq \bigoplus_{i=1}^{n-1} \delta^i
    \sqsubset \delta$. In fact, $0 \sqsubset \delta^n$ and 
    $\bigoplus_{i=1}^{n-1} \delta^i \sqsubseteq \comp{\delta^n}$. Thus by Lemma~\ref{le:mvprop}(\ref{le:mvprop:8}),
    $ \bigoplus_{i=1}^{n-1} \delta^i \sqsubset \delta^n \oplus \bigoplus_{i=1}^{n-1} \delta^i 
    = \bigoplus_{i=1}^{n} \delta^i = \delta$.
    Hence $Y_n = \gamma_{a,\delta}(b)$.

    \smallskip

    Let us now show that
    $\gamma_{f(a),\delta}(f(b)) \subseteq
    f_{a,\delta^n}^\#(Y_n)$. Given $z \in \gamma_{f(a),\delta}(f(b))$,
    we show that $z \in f_{a,\delta^n}^\#(Y_n)$. Observe that
    \begin{align*}
      & \delta \sqsubseteq f(b)(z) \ominus f(a)(z) =\\
      & \quad \mbox{[by hypothesis]}\\     
      & = f(a \oplus(b \ominus a))(z) \ominus f(a)(z) =\\      
      & \quad \mbox{[by Lemma~\ref{le:mvprop}(\ref{le:mvprop:2}), since $a \sqsubseteq b$]}\\
      & = f(a \oplus \bigoplus_{i=1}^n \delta^i_{Y_i})(z) \ominus f(a)(z) =\\
      & \quad \mbox{[by construction]}\\
      & = f(a \oplus \bigoplus_{i=1}^n \delta^i_{Y_i}))(z) \ominus f(a \oplus \delta^n_{Y_n})(z) \oplus f(a \oplus \delta^n_{Y_n})(z) \ominus f(a)(z)\\
      & \quad \mbox{[by Lemma~\ref{le:mvprop}(\ref{le:mvprop:2}), since $f(a \oplus \delta^n_{Y_n})(z) \sqsubseteq f(a \oplus \bigoplus_{i=1}^n \delta^i_{Y_i})(z)$]}\\
      & \sqsubseteq \norm{f(a \oplus \bigoplus_{i=1}^n \delta^i_{Y_i}) \ominus f(a \oplus \delta^n_{Y_n})} \oplus f(a \oplus \delta^n_{Y_n})(z) \ominus f(a)(z)\\
      & \quad \mbox{[by definition of norm and monotonicity of $\oplus$]}\\
      & \sqsubseteq \norm{a \oplus \bigoplus_{i=1}^n \delta^i_{Y_i} \ominus (a \oplus \delta^n_{Y_n})} \oplus f(a \oplus \delta^n_{Y_n})(z) \ominus f(a)(z)\\
      & \quad \mbox{[by non-expansiveness of $f$ and monotonicity of $\oplus$]}\\
      & = \norm{a \oplus \delta^n_{Y_n} \oplus \bigoplus_{i=1}^{n-1} \delta^i_{Y_i}\ominus (a \oplus \delta^n_{Y_n})} \oplus f(a \oplus \delta^n_{Y_n})(z) \ominus f(a)(z)\\
      & \quad \mbox{[by algebraic manipulation]}\\
      & \sqsubseteq \norm{\bigoplus_{i=1}^{n-1} \delta^i_{Y_i} } \oplus f(a \oplus \delta^n_{Y_n})(z) \ominus f(a)(z)\\
      & \quad \mbox{[by Lemma~\ref{le:mvprop}(\ref{le:mvprop:6}) and monotonicity of norm]}\\
      & \sqsubseteq \bigoplus_{i=1}^{n-1} \delta^i \oplus f(a \oplus \delta^n_{Y_n})(z) \ominus f(a)(z)\\
      & \quad \mbox{[by Lemma~\ref{le:norm}(\ref{le:norm:1}) and the fact that $\norm{\delta^i_{Y_i}} = \delta^i$]}\\
      & =  (\delta \ominus \delta^n) \oplus f(a \oplus \delta^n_{Y_n})(z) \ominus f(a)(z)\\
      & \quad \mbox{[by construction, since $\delta^n = \comp{\bigoplus_{i=1}^{n-1} \delta^i}$]}\\
    \end{align*}
     If we subtract
    $\delta \ominus \delta^n$ on both sides, we get
    $\delta \ominus (\delta \ominus \delta^n) \sqsubseteq f(a \oplus
    \delta^n_{Y_n})(z) \ominus f(a)(z)$, i.e., since, by
    Lemma~\ref{le:mvprop}(\ref{le:mvprop:10}),
    $\delta \ominus (\delta \ominus \delta^n) = \delta^n$ we conclude
    \begin{center}
      $\delta^n \sqsubseteq f(a \oplus \delta^n_{Y_n})(z) \ominus f(a)(z)$.
    \end{center}
    Hence
    $z \in \gamma_{f(a),\delta^n}(f(\alpha_{a,\delta^n}(Y_n)) =
    f_{a,\delta^n}^\#(Y_n)$, which is the desired result.
\end{proof}

We can finally prove the main result about legitimacy of the approximation.

\begin{theorem}[approximation of non-expansive functions]
  \label{th:approximation}
  Let $\monM$ be a complete MV-chain, let $Y, Z$ be finite sets and let
  $f: \monM^Y\to \monM^Z$ be a non-expansive function.
  Then for all $\zero \sqsubset\delta \in \monM$:
  \begin{enumerate}[a.]

  \item
    \label{th:approximation:1}
    $\gamma_{f(a),\delta} \circ f \subseteq f^{\#}_a \circ \gamma_{a,\delta}$

   \item
    \label{th:approximation:2}
    for $\delta
      \sqsubseteq \delta_a$: $\delta
      \sqsubseteq \inc{a}{f}$ iff
    $\gamma_{f(a),\delta} \circ f = f^{\#}_a \circ
      \gamma_{a,\delta}$
    \end{enumerate}
    \[
      \xymatrix@R=6mm{
        [a,a\oplus\delta] \ar[d]_{f}
        \ar[r]^{\gamma_{a,\delta}} \ar@{}[dr] | {\sqsubseteq} & \Pow{\Ybot{Y}{a}} \ar[d]^{f^{\#}_a} \\
        [f(a),f(a)\oplus \delta] \ar[r]_(.6){\gamma_{f(a),\delta}} &
        \Pow{\Ybot{Z}{f(a)}}
      }
    \]
  \end{theorem}

\begin{proof}

  \begin{enumerate}[a.]

  \item Let $b \in \interval{a}{a \oplus \delta}$. First note that
    whenever $\gamma_{f(a),\delta}(f(b)) = \emptyset$, %
    the desired
    inclusion obviously holds.

    If instead $\gamma_{f(a),\delta}(f(b)) \neq \emptyset$, let
    $b \ominus a = \bigoplus_{i=1}^n \delta^i_{Y_i}$ be a standard
    form with $\delta^n \neq 0$. First observe that, by
    Lemma~\ref{le:approx-fb}, we have $Y_n = \gamma_{a,\delta^n}(b)$ and
    \begin{equation}
      \label{eq:th:approximation:1}
      \gamma_{f(a),\delta}(f(b)) \subseteq f_{a,\delta^n}^\#(Y_n).
    \end{equation}

    For all $z \in f_{a,\delta^n}^\#(Y_n)$, by definition of
    $\inc{a}{f}(Y_n,z)$ we have that
    $0 \sqsubset \delta_n \sqsubseteq \inc{a}{f}(Y_n,z)$, therefore
    $\inc{a}{f} \sqsubseteq \inc{a}{f}(Y_n,z)$. Moreover,
    $z \in f_{a,\inc{a}{f}(Y_n,z)}^\#(Y_n) \subseteq
    f_{a,\inc{a}{f}}^\#(Y_n) = f_a^{\#}(Y_n)$, where the last
    inequality is motivated by Lemma~\ref{le:approximation-monotonic}
    since $\inc{a}{f} \sqsubseteq \inc{a}{f}(Y_n,z)$.
    Therefore,
    $f_{a,\delta^n}^\#(Y_n) \subseteq f_a^{\#}(\gamma_{a,\delta} (b))$,
    which combined with (\ref{eq:th:approximation:1}) gives the
    desired result.

  \item For (\ref{th:approximation:2}), we first show the direction
    from left to right. Assume that $\delta \sqsubseteq
    \inc{a}{f}$. By (a) clearly,
    $\gamma_{f(a),\delta} \circ f (b) \subseteq f^{\#}_a \circ
    \gamma_{a,\delta} (b)$. For the converse inclusion, note that:
    \begin{align*}
      & f_a^{\#} (\gamma_{a,\delta}(b)) & \mbox{[by definition of $f_a^{\#}$]}\\
      & \quad = f_{a,\inc{a}{f}}^\# (\gamma_{a,\delta}(b)) & \mbox{[by Lemma~\ref{le:approximation-monotonic}, since $\delta\sqsubseteq \inc{a}{f}$]}\\
      & \quad \subseteq f_{a,\delta}^{\#} (\gamma_{a,\delta}(b)) & \mbox{[by definition of $f_{a,\delta}^{\#}$]}\\
      & \quad = \gamma_{f(a),\delta}(f(\alpha_{a,\delta}(\gamma_{a,\delta}(b)))) & \mbox{[since $\alpha_{a,\delta} \circ \gamma_{a,\delta}(b) \sqsubseteq b$]}\\
      & \quad \subseteq \gamma_{f(a),\delta}(f(b))
    \end{align*}
    as desired.

    For the other direction, assume
    $\gamma_{f(a),\delta} \circ f(b) = f_a^\# \circ
    \gamma_{a,\delta}(b)$ holds for all $b \in [a,a\oplus \delta]$.
    Now, for every $Y' \subseteq \Ybot{Y}{a}$ we have
    $f_{a,\delta}^\#(Y') = \gamma_{f(a),\delta} \circ f \circ
    \alpha_{a,\delta} (Y') = f_a^\# \circ \gamma_{a,\delta} \circ
    \alpha_{a,\delta} (Y')$. We also have
    $\gamma_{a,\delta} \circ \alpha_{a,\delta} (Y') = Y'$ (see proof
    of Lemma~\ref{le:galois}), thus
    $f_{a,\delta}^\# (Y') = f_a^\#(Y')$. For any $\delta$ with
    $\inc{a}{f}\sqsubset \delta \sqsubseteq \delta_a$ there exists
    $Y'\subseteq \Ybot{Y}{a}$ and $z\in \Ybot{Z}{f(a)}$ with
    $z \in f_a^\# (Y')$ but $z \notin f_{a,\delta}^\#(Y')$, by
    definition of $\inc{a}{f}$. Therefore
    $\delta \sqsubseteq \inc{a}{f}$ has to hold.
    \qedhere
  \end{enumerate}
\end{proof}

Note that if $Y=Z$ and $a$ is a fixpoint of $f$, i.e., $a = f(a)$,
then condition~(\ref{th:approximation:1}) above corresponds exactly to
soundness in the sense of abstract
interpretation~\cite{cc:ai-unified-lattice-model}. Moreover, when
$\delta \sqsubseteq \delta_a$ and thus
$\langle \alpha_{a,\delta}, \gamma_{a,\delta} \rangle$ is a Galois
connection, $f_{a,\delta}^\# = \gamma_{a,\delta} \circ f \circ \alpha_{a,\delta}$
is the best correct approximation of $f$. In particular, when $\delta \sqsubseteq \inc{a}{f}$, such a best correct approximation is $f_a^\#$, the $a$-approximation of $f$, i.e., it becomes independent from $\delta$,
and condition~(\ref{th:approximation:2}) corresponds to
($\gamma$-)completeness~\cite{GRS:MAIC} (see also Section~\ref{se:setup}).

\section{Proof rules}
\label{se:proof-rules}

In this section we formalise the proof technique outlined in the
introduction for showing that a fixpoint is the largest and, more
generally, for checking over-approximations of greatest fixpoints of
non-expansive functions.

\subsection{Proof rules for fixpoints}

Consider a monotone function $f : \monM^Y \to \monM^Y$  for some finite set $Y$. We first focus on the problem of establishing whether some 
given fixpoint $a$ of $f$ coincides with $\nu f$ (without explicitly
knowing $\nu f$), and, in case it does not, finding an ``improvement'', i.e., a post-fixpoint of $f$, larger than $a$.
To this aim we need a technical lemma.

\begin{lemmarep}
  \label{le:technical-soundness}
  Let $\monM$ be a complete
  MV-chain, $Y$ a finite set and $f : \monM^Y\to \monM^Y$ be a
  non-expansive function.
  Let $a \in \monM^Y$ be a pre-fixpoint of $f$ (i.e.,
  $f(a)\sqsubseteq a$), let
  $f_a^{\#}: \Pow{\Ybot{Y}{a}} \to \Pow{\Ybot{Y}{f(a)}}$ be the
  $a$-approximation of $f$.
  Assume $\nu f \not\sqsubseteq a$ and let
  $Y' = \{ y \in \Ybot{Y}{a} \mid \nu f (y) \ominus a(y) = \norm{\nu f \ominus
    a}\}$. Then for all $y \in Y'$ it holds $a(y) = f(a)(y)$ and $Y' \subseteq f_a^{\#}(Y')$.
\end{lemmarep}

\begin{proof}
  Let $\delta =  \norm {\nu f \ominus a}$. Assume $\nu f\not\sqsubseteq a$, i.e., there exists $y\in Y$ such
  that $\nu f(y)\not\sqsubseteq a(y)$. Since the order is total, this
  means that $a(y)\sqsubset \nu f(y)$. Hence, by
  Lemma~\ref{le:mvprop}(\ref{le:mvprop:5}),
  $\nu f(y) \ominus a(y) \sqsupset \zero$.
  Then $\delta = \norm {\nu f \ominus a} \sqsupset \zero$. Moreover,
  for all $y \in Y'$,
  $\comp{a(y)} = 1 \ominus a(y) \sqsupseteq \nu f(y) \ominus a(y) =
  \delta$.

  \smallskip

  First, observe that
  \begin{equation}
    \label{eq:soundness:1}
    \nu f \sqsubseteq a \oplus \delta,
  \end{equation}
  since for all $y\in Y$ $\nu f(y)\ominus a(y) \sqsubseteq \delta$ by
  definition of $\delta$ and then (\ref{eq:soundness:1}) follows from
  Lemma~\ref{le:mvprop}(\ref{le:mvprop:7}).

  \smallskip

  Concerning the first part, let $y \in Y'$. Since $a$ is a
  pre-fixpoint, $f(a)(y) \sqsubseteq a(y)$. Assume by contradiction
  that $f(a)(y) \sqsubset a(y)$. Then we have
  \begin{align*}
    & f(a\oplus\delta)(y) =\\
    & \quad \mbox{[by Lemma~\ref{le:mvprop}(\ref{le:mvprop:2}), since $f$ is monotone and thus $f(a) \sqsubseteq f(a\oplus \delta)$]}\\
    & = f(a)(y) \oplus (f(a\oplus\delta)(y) \ominus f(a)(y))\\
    & \quad \mbox{[since $f$ is non-expansive, by Lemma~\ref{le:non-expansiveness}, hence $f(a\oplus\delta)(y)\ominus f(a)(y) \sqsubseteq \delta$]}\\
     & \sqsubseteq f(a)(y) \oplus \delta\\
    & \quad \mbox{[by $f(a)(y) \sqsubset a(y)$, $\delta \sqsubseteq \comp{a(y)}$ and Lemma~\ref{le:mvprop}(\ref{le:mvprop:6})]}\\
    & \sqsubset a(y) \oplus \delta\\
    & \quad \mbox{[by Lemma~\ref{le:mvprop}(\ref{le:mvprop:2}) since $a(y) \sqsubseteq \nu f(y)$ and $\delta = \nu f(y) \ominus a(y)$]}\\
    & = \nu f(y)\\
    &=  f(\nu f)(y)\\
    & \quad \mbox{[since $\nu f \sqsubseteq a \oplus \delta$
      (\ref{eq:soundness:1}) and $f$ monotone]}\\
    & \sqsubseteq f(a\oplus\delta)(y)
  \end{align*}
  i.e., a contradiction. Hence it must be $a(y) = f(a)(y)$.

  \smallskip

  For the second part, in order to show $Y' \subseteq f_a^{\#}(Y')$,
  we let $b = \nu f\sqcup a$. By using (\ref{eq:soundness:1}) we
  immediately have that $b \in \interval{a}{a \oplus \delta}$.

  \smallskip

  We next prove that
  \begin{center}
    $Y' = \gamma_{a,\delta}(b)$.
  \end{center}
  We show separately the two inclusions. If $y \in Y'$ then
  $a(y) \sqsubset \nu f(y)$ and thus
  $b(y) = a(y) \sqcup \nu f(y) = \nu f (y)$ and thus
  $b(y) \ominus a(y) = \nu f(y) \ominus a(y) =
  \delta$. Hence $y \in \gamma_{a,\delta}(b) $. Conversely, if
  $y \in \gamma_{a,\delta}(b)$, then $a(y) \sqsubset \nu f(y)$. In
  fact, if it were $a(y) \sqsupseteq \nu f(y)$, then, by definition of
  $b$ we would have $b(y) = a(y)$ and
  $b(y) \ominus a(y) = \zero \not\sqsupseteq \delta$. Therefore,
  $b(y) =\nu f(y)$ and thus
  $\nu f(y) \ominus a(y) = b(y) \ominus a(y) \sqsupseteq \delta$,
  whence $y \in Y'$.

  We can now conclude. In fact, since $f$ is non-expansive, by
  Theorem~\ref{th:approximation}(\ref{th:approximation:1}), we have
  \[ \gamma_{f(a),\delta}(f(b)) \subseteq f^\#_a(Y'). \]

  Moreover
  $Y' \subseteq \gamma_{f(a),\delta}(f(b))$. In fact, let $y \in Y'$, i.e.,
  $y \in \Ybot{Y}{a}$ and $\delta \sqsubseteq b(y) \ominus
  a(y)$. 
  Since $a(y) = f(a)(y)$, we have that $y \in \Ybot{Y}{f(a)}$. In order to conclude that $y \in \gamma_{f(a),\delta}(f(b))$ it is left to show that 
  $\delta \sqsubseteq f(b)(y)\ominus f(a)(y)$. We have
  \begin{align*}
    f(b)(y) \ominus f(a)(y)
    & =  f(b)(y) \ominus a(y) & \mbox{[since $y \in Y'$]}\\
    & =  f(\nu f \sqcup a)(y) \ominus a(y) & \mbox{[definition of $b$]}\\
    & \sqsupseteq (f(\nu f)(y) \sqcup f(a)(y)) \ominus a(y) & \mbox{[properties of $\sqcup$]}\\
    & = (\nu f(y) \sqcup a(y)) \ominus a(y) & \mbox{[since $\nu f$ fixpoint and $y \in Y'$]}\\
    & = b(y) \ominus a(y) & \mbox{[definition of $b$]}\\
    & \sqsupseteq \delta & \mbox{[since $y \in Y'$]}
  \end{align*}

  Combining the two inclusions, we have $Y' \subseteq f_a^\#(Y')$, as desired.
\end{proof}

Observe that, when $a$ is a fixpoint, clearly
$\Ybot{Y}{a} = \Ybot{Y}{f(a)}$, and thus the $a$-approximation of $f$
(Lemma~\ref{le:approximation-2}) is an endo-function
$f_a^\# : \Ybot{Y}{a} \to \Ybot{Y}{a}$ and $Y'$ is its
post-fixpoint. Then, we have the following result, which relies on the
fact that, due to Theorem~\ref{th:approximation} and properties of Galois connections, $\gamma_{a,\delta}$ maps
the greatest fixpoint of $f$ to to the greatest fixpoint of
 $f_a^\#$.

\begin{theorem}[soundness and completeness for fixpoints]
  \label{th:fixpoint-sound-compl}
  Let $\monM$ be a complete
  MV-chain, $Y$ a finite set and $f : \monM^Y\to \monM^Y$ be a
  non-expansive function.
  Let $a \in \monM^Y$ be a fixpoint of $f$. Then 
  $\nu f_a^{\#} = \emptyset$ if and only if $a = \nu f$.
\end{theorem}

\begin{proof}
  Let $a$ be a fixpoint of $f$ and assume that $a = \nu f$. For
  $\delta = \inc{a}{f} \sqsubseteq \delta_a$, according to
  Lemma~\ref{le:galois}, we have a Galois connection:
  \begin{center}
    \begin{tikzpicture}[->]
      \node (y) {$\Pow{\Ybot{Y}{a}}$};
      \node [right=of y] (b) 
      {$\interval{a}{a+\delta}$}; 
      \draw [->,thick,bend left] (y) to node [above]{$\alpha_{a,\delta}$} (b); 
      \draw [->,thick,bend left] (b) to node [below]{$\gamma_{a,\delta}$}
      (y);
      \draw [->,thick] (y) edge [loop left] node {$f^\#_a$} (y);	
      \draw [->,thick] (b) edge [loop right] node {$f_{a,\delta}$} ();
    \end{tikzpicture}
  \end{center}
  Since $a$ is a fixpoint, then $\Ybot{Y}{f(a)} = \Ybot{Y}{a}$ and, by
  Theorem~\ref{th:approximation}(\ref{th:approximation:2}),
  $\gamma_{a,{\delta}} \circ f = \gamma_{f(a),{\delta}} \circ f =
  f_a^{\#}\circ \gamma_{a,{\delta}}$.

  Therefore by~\cite[Proposition~14]{CC:TLA},
  $\nu f^\#_a = \gamma_{a,{\delta}}(\nu f)$. Recall that
  $\gamma_{a,{\delta}}(\nu f) = \{ y \in Y \mid {\delta} \sqsubseteq
  \nu f(y) \ominus a(y) \}$. Since $a = \nu f$ and
  ${\delta} \sqsupset \zero$, we know that
  $\gamma_{a,{\delta}}(\nu f) = \emptyset$ and we conclude
  $\nu f^\#_a = \emptyset$, as desired.

  \medskip

  Conversely, in order to prove that if $\nu f^\#_a = \emptyset$ then
  $a= \nu f$, we prove the contrapositive. Assume that
  $a \neq \nu f$. Since $a$ is a fixpoint and $\nu f $ is the largest,
  this means that $a \sqsubset \nu f$ and thus
  $\norm{\nu f \ominus a} \neq 0$. Consider
  $Y' = \{ y \in \Ybot{Y}{a} \mid \nu f (y) \ominus a(y) = \norm{\nu f \ominus
    a} \} \neq \emptyset$. By Lemma~\ref{le:technical-soundness},
  $Y'$ is a post-fixpoint of $f_a^{\#}$, i.e., $Y' \subseteq f_a^{\#}(Y')$,
  and thus $\nu f^\#_a \supseteq Y'$ which implies
  $\nu f^\#_a \neq \emptyset$, as desired.
\end{proof}

Whenever $a$ is a fixpoint, but not yet the largest fixpoint of $f$, from the result above $\nu f^\#_a\neq \emptyset$. Intuitively, $\nu f^\#_a$ is the set of points where $a$ can still be ``improved''. More precisely,
we can show that $a$ can be increased on the points in $\nu f^\#_a$ producing a post-fixpoint of $f$.
In order to determine how much $a$ can be increased we proceed similarly to what we have done for defining $\inc{a}{f}$ (Lemma~\ref{le:approximation-2}), but restricting the attention to $\nu f_a^{\#}$ instead of considering the full $\Ybot{Y}{a}$.
While $\inc{a}{f}$ could always be used, by restricting to $\nu f_a^{\#}$ we are able to find a better, that is larger, value which is still correct.

\begin{defi}[largest increase for a subset]
  \label{de:greatest-increase-set}
  Let $\monM$ be a complete MV-chain and let $f: \monM^Y\to \monM^Y$
  be a non-expansive function, where $Y$ is a finite set and let
  $a \in \monM^Y$. For $Y' \subseteq Y$, we define
  $\delta_{a}(Y') = \min\{ \comp{a(y)} \mid y \in Y'\}$ and
  $\inc{a}{f}(Y') = \min \{ \inc{a}{f}(Y',y)  \mid y \in Y' \}$.
\end{defi}

\begin{exa}
  \label{ex:running-3}
  We intuitively explain the computation of the values in the
  definition above. Let $g\colon [0,1]^Y \to [0,1]^Y$ with
  $g(b) = b \oplus 0.1$, where the set $Y$ and the function
  $a \in [0,1]^Y$ are as in Example~\ref{ex:running-1}.

  Let $Y' = \{y_1,y_2\}$. Then $\delta_a(Y') = 0.6$ and
  $\inc{a}{g}(Y') = 0.5$, i.e., since $g$ adds $0.1$, we can
  propagate an increase of at most $0.5$.
\end{exa}

We next prove that when $a \in \monM^Y$ is a fixpoint of $f$ and
$Y' = \nu f_a^{\#}$, the value $\inc{a}{f}(Y')$ is the largest
increase $\delta$ below $\delta_{a}(Y')$ such that
$a\oplus \delta_{Y'}$ is a post-fixpoint of $f$.

\begin{proposition}[from a fixpoint to larger post-fixpoint]
  \label{pr:increases}
  Let $\monM$ be a complete MV-chain, $f : \monM^Y\to \monM^Y$ a
  non-expansive function, $a \in \monM$ a fixpoint of $f$, and let
  $Y' = \nu f_a^{\#}$ be the greatest fixpoint of the corresponding
  $a$-approximation. Then $\inc{a}{f} \sqsubseteq \inc{a}{f}(Y') \sqsubseteq
  \delta_{a}(Y')$. Moreover, for all $\theta \sqsubseteq \inc{a}{f}(Y')$
  the function $a \oplus \theta_{Y'}$ is a post-fixpoint of $f$, while
  for $\inc{a}{f}(Y') \sqsubset \theta \sqsubseteq \delta_{a}(Y')$ it
  is not.
\end{proposition}

\begin{proof}
  We first show that $\inc{a}{f} \sqsubseteq \inc{a}{f}(Y')$.
  By Lemma~\ref{le:approximation-2} and since $a = f(a)$, we have
  that $\inc{a}{f} = \min \{ \inc{a}{f}(Y'',y) \mid Y'' \subseteq
  \Ybot{Y}{a}\ \land\ y \in \Ybot{Y}{a}\ \land\ \inc{a}{f}(Y'',y)
  \neq 0\} \cup \{ \delta_a\}$.
  Moreover, we have $Y' = \nu f_a^{\#} \subseteq \Ybot{Y}{a}$
  and $\inc{a}{f}(Y',y) \neq 0$, for every $y \in Y'$,
  since $\inc{a}{f}(Y',y) = \max \{ \delta \in \monM \mid y \in
  f_{a,\delta}^{\#}(Y')\}$ and
  $y \in Y' = \nu f_a^{\#} = f_{a}^{\#}(\nu f_a^{\#}) =
  f_{a,\inc{a}{f}}^{\#}(Y')$,
  hence $\inc{a}{f}(Y',y) \sqsupseteq \inc{a}{f} \sqsupset 0$.
  Therefore, the minimum in $\inc{a}{f}(Y')$ is computed on
  a subset of the values on which the one in $\inc{a}{f}$ is,
  and so the former must be larger or equal to the latter.

  Next, we prove that $\inc{a}{f}(Y') \sqsubseteq
  \delta_{a}(Y')$. Observe that for all $y \in Y'$ and
  $\delta \in \monM$, if $y \in f_{a,\delta}^{\#}(Y')$,
  by definition of $f_{a,\delta}^{\#}$, it holds that
  $\delta \sqsubseteq f(a \oplus \delta_{Y'})(y) \ominus f(a)(y) =
  f(a \oplus \delta_{Y'})(y) \ominus a(y) \sqsubseteq 1 \ominus a(y)
  = \comp{a(y)}$, where the second equality is motivated by the fact
  that $a$ is a fixpoint. Therefore for all $y \in Y'$ we have
  $\max \{ \delta \in \monM \mid y \in f_{a,\delta}^{\#}(Y') \}
  \sqsubseteq \comp{a(y)}$ and thus 
  $\inc{a}{f}(Y') = \min_{y \in Y'} \max \{ \delta \in \monM \mid y
  \in f_{a,\delta}^{\#}(Y')\}  \sqsubseteq \min_{y  \in Y'}  \comp{a(y)}
  = \delta_{a}(Y')$, as desired.

  \medskip

  Given $\theta \sqsubseteq \inc{a}{f}(Y')$, let us prove that
  $a \oplus \theta_{Y'}$ is a post-fixpoint of $f$, i.e., 
  $a \oplus \theta_{Y'} \sqsubseteq f(a \oplus \theta_{Y'})$. 

  If $y \in Y'$, since $\theta \sqsubseteq \inc{a}{f}(Y')$, by
  definition of $\inc{a}{f}(Y')$, we have
  $\theta \sqsubseteq \max \{ \delta \in \monM \mid y \in
  f_{a,\delta}^{\#}(Y')\}$ and thus, by antitonicity of
  $f_{a,\delta}^{\#}$ with respect to $\delta$, we have
  $y \in f_{a,\theta}^{\#}(Y')$. This means that
  $\theta \sqsubseteq f(a \oplus \theta_{Y'})(y) \ominus f(a)(y) = f(a
  \oplus \theta_{Y'})(y) \ominus a(y)$, where the last passage uses
  the fact that $a$ is a fixpoint. Adding $a(y)$ on both sides and
  using Lemma~\ref{le:mvprop}(\ref{le:mvprop:2}), we obtain
  $a(y) \oplus \theta \sqsubseteq (f(a \oplus \theta_{Y'})(y) \ominus
  a(y)) \oplus a(y) = f(a \oplus \theta_{Y'})(y)$. Since $y \in Y'$,
  $(a \oplus \theta_{Y'})(y) = a(y) \oplus \theta$ and thus
  $(a \oplus \theta_{Y'})(y) \sqsubseteq f(a \oplus \theta_{Y'})(y)$, as desired.

  If instead, $y \not\in Y'$, clearly
  $(a \oplus \theta_{Y'})(y) = a(y) = f(a)(y) \sqsubseteq f(a \oplus
  \theta_{Y'})(y)$, where we again use the fact that $a$ is a fixpoint
  and monotonicity of $f$.

  \medskip

  Lastly, we have to show that if
  $\inc{a}{f}(Y') \sqsubset \theta \sqsubseteq \delta_{a}(Y')$, then
  $a \oplus \theta_{Y'}$ is a not a post-fixpoint of $f$.
  By definition of $\inc{a}{f}(Y')$, from the fact that
  $\inc{a}{f}(Y') \sqsubset \theta$, we deduce that
  $\max \{ \delta \in \monM \mid y \in f_{a,\delta}^{\#}(Y')\}
  \sqsubset \theta$ for some $y \in Y'$ and thus
  $y \not\in f_{a,\theta}^\#(Y')$.

  By definition of $f_{a,\theta}^\#$ and totality of $\sqsubseteq$,
  the above means
  $\theta \sqsupset f(a \oplus \theta_{Y'}) (y) \ominus f(a)(y) = f(a
  \oplus \theta_{Y'}) (y) \ominus a(y)$, since $a$ is a fixpoint of
  $f$. Since $\theta \sqsubseteq \delta_{a}(Y')$, we can add $a(y)$ on
  both sides and, by Lemma~\ref{le:mvprop}(\ref{le:mvprop:8}), we
  obtain $a(y) \oplus \theta \sqsupset f(a \oplus \theta_{Y'})
  (y)$. Since $y \in Y'$, the left-hand side is
  $(a \oplus \theta_{Y'})(y)$. Hence we conclude that indeed
  $a \oplus \theta_{Y'}$ is not a post fixpoint.
\end{proof}

Using these results one can perform an alternative fixpoint iteration
where we iterate to the largest fixpoint from below: start with a
post-fixpoint $a_0 \sqsubseteq f(a_0)$ (which is clearly below
$\nu f$) and obtain, by (possibly transfinite)
iteration, an
ascending chain that in the order converges to $a$,\footnote{Note that throughout the paper the term ``convergence'' used on complete MV-chains will always implicitly refer to convergence in the natural order.} the least fixpoint above
$a_0$. Now, letting $Y' = \nu f_a^{\#}$, check whether
$Y' = \emptyset$. If so, by Theorem~\ref{th:fixpoint-sound-compl}
we know we have reached $\nu f = a$.
If not, $\alpha_{a,\inc{a}{f}(Y')}(Y') = a\oplus (\inc{a}{f}(Y'))_{Y'}$ is
again a post-fixpoint (cf. Proposition~\ref{pr:increases}) and we
continue this procedure until
-- for some ordinal --
we reach the
largest fixpoint $\nu f$, for which we have
$\nu f_{\nu f}^{\#} = \emptyset$.

In order to make the above procedure as efficient as possible, one would like to consider, whenever a fixpoint $a$ is reached, the largest possible increase $\iota$ which is valid, i.e.\ such that $a\oplus \iota$ is again a post-fixpoint of $f$.
Thus the question naturally arises asking whether $\inc{a}{f}(Y')$ is such largest valid increase.
From Proposition~\ref{pr:increases}, it immediately follows that $\inc{a}{f}(Y')$ is the largest valid increase below $\delta_{a}(Y')$,
but it can be seen that there can be larger valid increases above $\delta_{a}(Y')$
(an explicit example is provided later in Example~\ref{ex:markov-no-greatest}, for the dual case of least fixpoints).
However, while the set of valid increases below $\delta_{a}(Y')$ is downward-closed, as proved in Proposition~\ref{pr:increases},
this is not the case for those above $\delta_{a}(Y')$.
Hence, we believe that the most efficient approach would be to search for $\inc{a}{f}(Y')$, or some satisfying approximation, via a binary search bounded by $\delta_{a}(Y')$.

\subsection{Proof rules for pre-fixpoints}

Interestingly, the soundness result in
Theorem~\ref{th:fixpoint-sound-compl} can be generalised to the case
in which $a$ is a pre-fixpoint instead of a fixpoint.
In this case, the $a$-approximation for a function
$f : \monM^Y \to \monM^Y$ is a function $f_a^{\#} : \Ybot{Y}{a} \to \Ybot{Y}{f(a)}$
where domain and codomain are different, hence it would not be
meaningful to look for fixpoints. However, as explained below, it can be restricted to an
endo-function. 

\begin{theorem}[soundness for pre-fixpoints]
  \label{th:soundness}
  Let $\monM$ be a complete
  MV-chain, $Y$ a finite set and $f : \monM^Y\to \monM^Y$ be a
  non-expansive function.
  Given a pre-fixpoint $a \in \monM^Y$ of $f$, let
  $\Ybot{Y}{a=f(a)} = \{ y \in \Ybot{Y}{a} \mid a(y) = f(a)(y) \}$.
  Let us 
  define $f^*_a : \Ybot{Y}{a=f(a)} \to \Ybot{Y}{a=f(a)}$ as
  $f^*_a(Y') = f_a^{\#}(Y') \cap \Ybot{Y}{a=f(a)}$, where
  $f_a^{\#}: \Pow{\Ybot{Y}{a}} \to \Pow{\Ybot{Y}{f(a)}}$ is the $a$-approximation of $f$.
  If $\nu f_a^* = \emptyset$
  then $\nu f \sqsubseteq a$.
\end{theorem}

\begin{proof}
  We prove the contrapositive, i.e., we show that
  $\nu f\not\sqsubseteq a$ allows us to derive that
  $\nu f^*_a \neq \emptyset$.

  Assume $\nu f\not\sqsubseteq a$, i.e., there exists $y\in Y$ such
  that $\nu f(y)\not\sqsubseteq a(y)$. Since the order is total, this
  means that $a(y)\sqsubset \nu f(y)$. Hence, by
  Lemma~\ref{le:mvprop}(\ref{le:mvprop:5}),
  $\nu f(y) \ominus a(y) \sqsupset \zero$.
  Then $\delta = \norm {\nu f \ominus a} \sqsupset \zero$.

  Consider
  $Y' = \{ y \in Y_a \mid \nu f (y) \ominus a(y) = \norm{\nu f \ominus
    a} \} \neq \emptyset$. By Lemma~\ref{le:technical-soundness},
  $Y'$ is a post-fixpoint of $f_a^{\#}$, i.e., $Y' \subseteq f_a^{\#}(Y')$,
  and thus $Y' \subseteq \nu f^\#_a$. Moreover, for all $y \in Y'$,
  $a(y) =f(a)(y)$, i.e., $Y' \subseteq \Ybot{Y}{a=f(a)}$. Therefore we
  conclude $Y' \subseteq f_a^\#(Y') \cap \Ybot{Y}{a=f(a)} = f_a^*(Y')$,
  i.e., $Y'$ is a post-fixpoint also for $f_a^*$, and thus
  $\nu f_a^* \supseteq Y' \neq \emptyset$, as desired.
\end{proof}

The reason why we can limit our attention to the set of points where
$a(y) = f(a)(y)$ is as follows.  Observe that, since $a$ is a
pre-fixpoint and $\ominus$ is antitone in the second argument, $\nu f \ominus a \sqsubseteq \nu f \ominus f(a)$.
Thus
$\norm{\nu f \ominus a} \sqsubseteq \norm{\nu f \ominus f(a)} = \norm{
  f(\nu f) \ominus f(a)} \sqsubseteq \norm{\nu f \ominus a}$, where
the last passage is motivated by non-expansiveness of $f$. Therefore
$\norm{\nu f \ominus a} = \norm{\nu f \ominus f(a)}$.  From this we
can deduce that, if $\nu f$ is strictly larger than $a$ on some
points, surely some of these points are in $\Ybot{Y}{a=f(a)}$. In
particular, all points $y_0$ such that
$\nu f(y_0) \ominus a(y_0) = \norm{\nu f \ominus a}$ are necessarily
in $\Ybot{Y}{a=f(a)}$. Otherwise, we would have
$f(a)(y_0) \sqsubset a(y_0)$ and thus
$\norm{\nu f \ominus a} = \nu f(y_0) \ominus a(y_0) \sqsubset \nu
f(y_0) \ominus f(a)(y_0) \sqsubseteq \norm{\nu f \ominus f(a)}$
(cf. Lemma~\ref{le:technical-soundness}).

\begin{rem}
Completeness does not generalise to pre-fixpoints, i.e., it is not true that
if $a$ is a pre-fixpoint of $f$ and $\nu f\sqsubseteq a$, then
$\nu f^*_a = \emptyset$. A pre-fixpoint might contain slack even
though it is above the greatest fixpoint. A counterexample is
in Example~\ref{ex:bisim-2}.
\end{rem}

\subsection{The dual view for least fixpoints}
\label{se:dual-view-app}

The theory developed so far can be easily dualised to check
under-approximations of least fixpoints. Given a complete MV-algebra
$\monM = (M, \oplus, \zero, \comp{(\cdot)})$ and a non-expansive
function $f : \monM^Y \to \monM^Y$, in order to show that
a post-fixpoint $a \in \monM^Y$ is such that $a \sqsubseteq \mu f$ we
can in fact simply work in the dual MV-algebra,
$\dual{\monM} = (M, \sqsupseteq, \otimes, \comp{(\cdot)}, \one,
\zero)$.

Since $\oplus$ could be the ``standard'' operation on $\monM$, it is
convenient to formulate the conditions using $\oplus$ and $\ominus$
and the original order.
The notation for the dual case is
obtained from that of the original case, referred to as the \emph{primal case} throughout the paper, exchanging
subscripts and superscripts.

The pair of functions
$\langle \alpha^{a,\theta}, \gamma^{a,\theta} \rangle$ is as
follows. Let $a : Y \to \monM$ and $\zero \sqsubset\theta\in
\monM$. The set $\Ytop{Y}{a} = \{ y \in Y \mid a(y) \neq \zero \}$ and
$\delta^a = \min \{ a(y) \mid y \in \Ytop{Y}{a} \}$

The target of the approximation is $\interval{a}{a\otimes\theta}$ in the reverse order, hence $\interval{a\otimes\theta}{a}$ in the original order. Recall that $a\otimes\theta = \comp{\comp{a} \oplus \comp{\theta}} = a \ominus \comp{\theta}$. Hence we obtain
\begin{center}
  \begin{tikzpicture}[->]
    \node (y) {$\Pow{\Ytop{Y}{a}}$};
    \node [right=of y] (b) 
    {$\interval{a \ominus \comp{\theta}}{a}$}; 
    \draw [->,thick,bend left] (y) to node [above]{$\alpha^{a,\theta}$} (b); 
    \draw [->,thick,bend left] (b) to node [below]{$\gamma^{a,\theta}$} (y);
  \end{tikzpicture}
\end{center}
  
For $Y' \in \Pow{\Ytop{Y}{a}}$ we define 
\begin{center}
  $\alpha^{a,\theta}(Y') = a \otimes \theta_{Y'} = a \ominus \comp{\theta}_{Y'}$
\end{center}
Instead $\gamma^{a,\theta}(b) = \{y \in Y \mid \theta \sqsupseteq b(y) \oominus a(y) \}$ where $\oominus$ is the subtraction in the dual MV-algebra. Observe that $x \oominus y = \comp{\comp{x} \otimes y} = \comp{x \oplus \comp{y}} = \comp{y \ominus x}$. Hence $\theta \sqsupseteq b(y) \oominus a(y)$ iff $a(y) \ominus b(y) \sqsupseteq \comp{\theta}$.  Thus for $b \in \interval{a\ominus\comp{\theta}}{a}$ we have
\begin{center}
  $\gamma^{a,\theta}(b) = \{y \in Y \mid \theta \sqsupseteq b(y) \oominus a(y) \} = \{y \in Y \mid a(y) \ominus b(y) \sqsupseteq \comp{\theta}\}$.
\end{center}

Let $f: \monM^Y\to \monM^Z$ be a monotone
function. The norm becomes $\norm{a} = \min \{ a(y) \mid y \in Y\}$. Non-expansiveness (Def.~\ref{de:non-expansiveness}) in the dual MV-algebra becomes: for all $a, b \in \monM^Y$,
$\norm{f(b) \oominus f(a)} \sqsupseteq \norm{b \oominus a}$, which in turn
is
\begin{center}
  $\min \{ \comp{f(a) \ominus f(b)} \mid y \in Y \} \sqsupseteq \min
  \{ \comp{a(y) \ominus b(y)} \mid y \in Y \}$
\end{center}
i.e., $\norm{f(a) \ominus f(b)} \sqsubseteq \norm{a \ominus b}$, which coincides with non-expansiveness in the original MV-algebra.

Observe that, instead of taking a generic $\theta \sqsubset 1$ and
then working with $\bar{\theta}$, we can directly take
$0 \sqsubset \theta$ and replace everywhere $\bar{\theta}$ with
$\theta$.

While the approximation of a function in the primal case are denoted
$f_a^\#$, the approximations in the dual case will be denoted by
$f^a_\#$. %

We can also dualise Proposition~\ref{pr:increases} and obtain that, whenever
$a$ is a fixpoint and $Y' = \nu f^a_{\#} \neq \emptyset$, then
$a \ominus \theta_{Y'}$ is a pre-fixpoint, where
$\theta = \iota_a^f(Y')$ is suitably defined, dualising
Def.~\ref{de:greatest-increase-set}.

\section{(De)Composing functions and approximations}
\label{se:de-composing}

Given a non-expansive function $f$ and a (pre/post-)fixpoint $a$, it
is often non-trivial to determine the corresponding
approximations. However, non-expansive functions enjoy good closure
properties (closure under composition, and closure under disjoint
union) and we will see that the same holds for the corresponding
approximations. Furthermore, it turns out that the functions needed in
the applications can be obtained from just a few templates. This gives
us a toolbox for assembling approximations with relative ease.

We start by introducing some basic functions, which will be used as
the building blocks for the functions needed in the applications.
Note that below we consider distributions on MV-chains of which the
probability distributions introduced earlier are a special case.

\begin{defi}[basic functions]
  \label{de:basic-functions}
  Let $\monM$ be an MV-chain and let $Y$, $Z$ be finite sets.
  \begin{enumerate}
    
  \item \emph{Constant:} For a fixed $k\in\monM^Z$, we define
    $c_k : \monM^Y\to \monM^Z$ by
    \[ c_k(a) = k \]

  \item \emph{Reindexing:} For $u : Z\to Y$, we define
    $u^* : \monM^Y\to\monM^Z$ by
    \[ u^*(a) = a \circ u. \]

  \item \emph{Min/Max:} For $\mathcal{R} \subseteq Y\times Z$,
    we define
    $\mins_\mathcal{R}, \maxs_\mathcal{R} : \monM^Y\to\monM^Z$ by
    \[ \mins_\mathcal{R}(a)(z) = \mins_{y\mathcal{R}z} a(y) \qquad 
      \maxs_\mathcal{R}(a)(z) = \maxs_{y\mathcal{R}z} a(y) \] 
    
  \item \emph{Average:} Call a function $p : Y \to \monM$ a
    distribution when for all $y \in Y$, it holds
    $\comp{p(y)} = \bigoplus_{y' \in Y \backslash \{y\}} p(y')$ and
    let $\mathcal{D}(Y)$ be the set of distributions. Assume that
    $\monM$ is endowed with an additional operation $\odot$ such that
    $(\monM, \odot, 1)$ is a commutative monoid, for $x, y \in \monM$,
    $x \odot y \sqsubseteq x$, and $x \odot y = 0$ iff $x=0$ or $y=0$,
    and $\odot$ weakly distributes over $\oplus$, i.e., for all
    $x, y, z \in \monM$ with $y \sqsubseteq \comp{z}$,
    $x \odot (y \oplus z) = x \odot y \oplus x \odot z$. For
    a finite set $D \subseteq \mathcal{D}(Y)$, we define
    $\mathrm{av}_D : \monM^Y\to\monM^D$ by
    \[ \mathrm{av}_D(a)(p) = \bigoplus_{y\in Y} p(y)\odot a(y) \]
  \end{enumerate}
\end{defi}

A particularly interesting subcase of (3) is when we take as relation
the \emph{belongs to} relation $\in\ \subseteq Y\times\Pow{Y}$. In
this way we obtain functions for selecting the minimum and the
maximum, respectively, of an input function over a set
$Y' \subseteq Y$, that is, the functions
$\mins_\in, \maxs_\in : \monM^Y \to \monM^{\Pow{Y}}$, defined as
\[ \mins_\in(a)(Y') = \min\limits_{y\in Y'} a(y) \qquad \qquad
  \maxs_\in(a)(Y') = \max\limits_{y\in Y'} a(y) \]

The usual probability
distributions arise as a special case of $\mathcal{D}(Y)$ in (4) with $\monM = [0,1]$ where $\odot$ is the standard
multiplication.

Also note that in the definition of $\mathit{av}_D$, the operation
$\odot$ is necessarily monotone. In fact, if $y \sqsubseteq y'$ then,
by Lemma~\ref{le:mvprop}(\ref{le:mvprop:2}), we have
$y' = y \oplus (y' \ominus y)$. Therefore
$x \odot y \sqsubseteq x \odot y \oplus x \odot (y' \ominus y) = x
\odot (y \oplus (y' \ominus y)) = x \odot y'$, where the second
passage holds by weak distributivity.

\begin{toappendix}

The basic functions can be shown to be non-expansive.

\begin{proposition}
  \label{pr:fct-nonexpansive}
  The basic functions from Def.~\ref{de:basic-functions} are all
  non-expansive.
\end{proposition}

\begin{proof}
  \mbox{}
  
  \begin{itemize}
  \item \emph{Constant functions:} immediate.
    
  \item \emph{Reindexing:} Let $u : Z \to Y$. For all
    $a,b \in \monM^Y$, we have
    \begin{align*}
      & \norm{u^*(b) \ominus u^*(a)} & \\
      & = \max_{z \in Z} (b(u(z)) \ominus a(u(z))) & \\
      & \sqsubseteq \max_{y \in Y} (b(y) \ominus a(y)) & \mbox{[since $u(Z) \subseteq Y$]}\\
      & = \norm{b \ominus a} & \mbox{[by def.\ of norm]}
    \end{align*}
    
  \item \emph{Minimum:} Let $\mathcal{R} \subseteq Y \times Z$ be a
    relation. For all $a,b \in \monM^Y$, we have
    \[ \norm{\mins_\mathcal{R}(b) \ominus \mins_\mathcal{R}(a)} =
      \max_{z \in Z} (\min_{y\mathcal{R}z} b(y) \ominus
      \min_{y\mathcal{R}z} a(y)) \] Observe that
    \[ \max_{z \in Z} (\min_{y\mathcal{R}z} b(y) \ominus
      \min_{y\mathcal{R}z} a(y)) = \max_{z \in Z'}
      (\min_{y\mathcal{R}z} b(y) \ominus \min_{y\mathcal{R}z} a(y)) \]
    where $Z' = \{z\in Z \mid \exists\, y\in Y.\, y\mathcal{R}z\}$,
    since on every other $z \in Z \backslash Z'$ the difference
    would be $\zero$. Now, for every $z \in Z'$, take $y_z \in Y$ such
    that $y_z\mathcal{R}z$ and
    $a(y_z) = \min\limits_{y\mathcal{R}z} a(y)$. Such a $y_z$ is
    guaranteed to exist whenever $Y$ is finite. Then, we have
    \begin{align*}
      & \max_{z \in Z'} (\min_{y\mathcal{R}z} b(y) \ominus \min_{y\mathcal{R}z} a(y)) & \\
      & \sqsubseteq \max_{z \in Z'} (b(y_z) \ominus a(y_z)) & \mbox{[$\ominus$ monotone in first arg.]}\\
      & \sqsubseteq \max_{z \in Z'} \norm{b \ominus a} & \mbox{[by def.\ of norm]}\\
      & = \norm{b \ominus a} & \mbox{[$\norm{b \ominus a}$ is independent from $z$]}
    \end{align*}
    
  \item \emph{Maximum:} Let $\mathcal{R} \subseteq Y \times Z$ be a
    relation. For all $a,b \in \monM^Y$ we have
    \begin{align*}
      & \norm{\maxs_\mathcal{R}(b) \ominus \maxs_\mathcal{R}(a)}\\
      & = \max_{z \in Z} (\max_{y\mathcal{R}z} b(y) \ominus \max_{y\mathcal{R}z} a(y))\\
      & \sqsubseteq \max_{z \in Z} (\max_{y\mathcal{R}z} ((b(y) \ominus a(y)) \oplus a(y)) \ominus \max_{y\mathcal{R}z} a(y))\\
      & \qquad\mbox{[since $(b(y) \ominus a(y)) \oplus a(y) = a(y) \sqcup b(y)$ and $\ominus$ monotone in first arg.]}\\[1mm]
      & \sqsubseteq \max_{z \in Z} ((\max_{y\mathcal{R}z} (b(y) \ominus a(y)) \oplus \max_{y\mathcal{R}z} a(y)) \ominus \max_{y\mathcal{R}z} a(y))\\
      & \qquad\qquad\mbox{[by def.\ of $\max$ and monotonicity of $\oplus$]}\\[1mm]
      & \sqsubseteq \max_{z \in Z} \max_{y\mathcal{R}z} (b(y) \ominus a(y)) \qquad\qquad \mbox{[by Lemma~\ref{le:mvprop}(\ref{le:mvprop:6})]}\\
      & \sqsubseteq \max_{z \in Z} \max_{y\mathcal{R}z} \norm{b \ominus a} \qquad\qquad \mbox{[by def.\ of norm]}\\
      & = \norm{b \ominus a} \qquad\qquad \mbox{[since $\norm{b \ominus a}$ is independent]}
    \end{align*}
    
  \item \emph{Average:}
    We first note that, when $p : Y \to \monM$, with $Y$ finite, is a
    distribution, then an inductive argument based on weak
    distributivity, allows one to show that for all $x \in \monM$,
    $Y' \subseteq Y$,
    $x \odot \bigoplus_{y\in Y'} p(y) = \bigoplus_{y \in Y'} x \odot
    p(y)$.

    For all
    $a,b \in \monM^Y$ we have
    \begin{align*}
      & \norm{\mathrm{av}_D(b) \ominus \mathrm{av}_D(a)}\\
      & = \max_{p \in D} (\bigoplus_{y\in Y} p(y) \odot b(y) \ominus \bigoplus_{y\in Y} p(y) \odot a(y))\\
      & \sqsubseteq \max_{p \in D} (\bigoplus_{y\in Y} p(y) \odot ((b(y) \ominus a(y)) \oplus a(y)) \ominus \bigoplus_{y\in Y} p(y) \odot a(y))\\
      & \qquad\mbox{[by monotonicity of $\odot,\oplus,\ominus$ and $(b(y) \ominus a(y)) \oplus a(y) = a(y) \sqcup b(y)$]}\\[1mm]
      & = \max_{p \in D} (\bigoplus_{y\in Y} (p(y) \odot (b(y) \ominus a(y))) \oplus (p(y) \odot a(y)) \ominus \bigoplus_{y\in Y} p(y) \odot a(y))\\
      & \qquad\qquad\mbox{[since $b(y) \ominus a(y) \sqsubseteq 1
        \ominus a(y) = \comp{a(y)}$, and $\odot$ weakly distributes over $\oplus$]}\\
      & = \max_{p \in D} ((\bigoplus_{y\in Y} p(y) \odot (b(y) \ominus a(y)) \oplus \bigoplus_{y\in Y} p(y) \odot a(y)) \ominus \bigoplus_{y\in Y} p(y) \odot a(y))\\
      & \sqsubseteq \max_{p \in D} \bigoplus_{y\in Y} p(y) \odot (b(y) \ominus a(y)) \qquad\qquad \mbox{[by Lemma~\ref{le:mvprop}(\ref{le:mvprop:6})]}\\
      & \sqsubseteq \max_{p \in D} \bigoplus_{y\in Y} p(y) \odot \norm{b \ominus a} \qquad\qquad \mbox{[by def.\ of norm and monotonicity of $\odot$]}\\
      & = \max_{p \in D} \norm{b \ominus a} \odot \bigoplus_{y\in Y}
      p(y) \qquad \mbox{[since $p$ is a distribution and $\odot$ weakly distributes over $\oplus$]}\\
      & = \max_{p \in D} (\norm{b \ominus a} \odot \one) \qquad\qquad \mbox{[since $p$ is a distribution over $Y$]}\\
      & = \norm{b \ominus a} \qquad\qquad \mbox{[since \norm{b \ominus
          a} is independent from $p$]} \tag*{\qedhere}
    \end{align*}
    \qedhere
  \end{itemize}
\end{proof}

The next result determines the approximations associated with the
basic functions.

\begin{proposition}[approximations of basic functions]
  \label{pr:associated-approximations}
  Let $\monM$ be an MV-chain, $Y, Z$ be finite sets and let $a \in
  \monM^Y$.
  \begin{itemize}
  \item \emph{Constant:} for $k : \monM^Z$, the approximations
    $(c_k)^\#_a : \Pow{\Ybot{Y}{a}} \to \Pow{\Ybot{Z}{c_k(a)}}$,
    $(c_k)_\#^a : \Pow{\Ytop{Y}{a}} \to \Pow{\Ytop{Z}{c_k(a)}}$
    are
    \[ (c_k)^\#_a(Y') = \emptyset = (c_k)_\#^a(Y') \]
  \item \emph{Reindexing:} for $u : Z \to Y$, the approximations
    $(u^*)^\#_a : \Pow{\Ybot{Y}{a}} \to \Pow{\Ybot{Z}{u^*(a)}}$,
    $(u^*)_\#^a : \Pow{\Ytop{Y}{a}} \to \Pow{\Ytop{Z}{u^*(a)}}$
    are
    \[ (u^*)^\#_a(Y') = (u^*)_\#^a(Y') = u^{-1}(Y') = \{z \in
      \Ybot{Z}{u^*(a)} \mid u(z) \in Y'\} \]
  \item \emph{Min:} for $\mathcal{R} \subseteq Y \times Z$, the
    approximations
    $(\mins_\mathcal{R})_a^\#  : \Pow{\Ybot{Y}{a}} \to
    \Pow{\Ybot{Z}{\min_\mathcal{R}(a)}}$,
    $(\mins_\mathcal{R})^a_\#  : \Pow{\Ytop{Y}{a}} \to
    \Pow{\Ytop{Z}{\min_\mathcal{R}(a)}}$ are given below, where
    $\mathcal{R}^{-1}(z) = \{y\in Y \mid y\mathcal{R}z\}$:
    \[ (\mins_\mathcal{R})_a^\#(Y') = \{z \in
      \Ybot{Z}{\mins_\mathcal{R}(a)} \mid
      \arg\min_{y\in\mathcal{R}^{-1}(z)} a(y) \subseteq Y'\} \]
    \[ (\mins_\mathcal{R})^a_\#(Y') = \{z \in
      \Ytop{Z}{\mins_\mathcal{R}(a)} \mid
      \arg\min_{y\in\mathcal{R}^{-1}(z)} a(y)\cap Y' \neq
      \emptyset\} \]
  \item \emph{Max:} for $\mathcal{R} \subseteq Y \times Z$, the approximations
    $(\maxs_\mathcal{R})_a^\#  : \Pow{\Ybot{Y}{a}} \to \Pow{\Ybot{Z}{\max_\mathcal{R}(a)}}$, $(\maxs_\mathcal{R})^a_\#  : \Pow{\Ytop{Y}{a}} \to \Pow{\Ytop{Z}{\max_\mathcal{R}(a)}}$ are
    \[ (\maxs_\mathcal{R})_a^\#(Y') = \{z \in
      \Ybot{Z}{\maxs_\mathcal{R}(a)} \mid
      \arg\max_{y\in\mathcal{R}^{-1}(z)} a(y)\cap Y' \neq
      \emptyset\} \]
    \[ (\maxs_\mathcal{R})^a_\#(Y') = \{z \in
      \Ytop{Z}{\maxs_\mathcal{R}(a)} \mid
      \arg\max_{y\in\mathcal{R}^{-1}(z)} a(y) \subseteq Y'\} \]
  \item \emph{Average:} for a finite $D\subseteq \mathcal{D}(Y)$, the
    approximations
    $(\mathrm{av}_D)_a^\#  : \Pow{\Ybot{Y}{a}} \to
    \Pow{\Ybot{D}{\mathrm{av}_D(a)}}$,
    $(\mathrm{av}_D)^a_\# \colon \Pow{\Ytop{Y}{a}} \to
    \Pow{\Ytop{D}{\mathrm{av}_D(a)}}$ are
    \[ (\mathrm{av}_D)_a^\#(Y') = \{p \in
      \Ybot{D}{\mathrm{av}_D(a)} \mid \mathit{supp}(p) \subseteq
      Y'\} \]
    \[ (\mathrm{av}_D)^a_\#(Y') = \{p \in \Ytop{D}{\mathrm{av}_D(a)}
      \mid \mathit{supp}(p) \subseteq Y'\}, \] where
    $\mathit{supp}(p) = \{y\in Y\mid p(y) > 0\}$ for
    $p\in \mathcal{D}(Y)$.
  \end{itemize}
\end{proposition}

\begin{proof}
  We only consider the primal cases, the dual ones are analogous.
  
  Let $a \in \monM^Y$.
  \begin{itemize}
  \item \emph{Constant:} for all
    $0\sqsubset \theta \sqsubseteq \delta_a$ and $Y'\subseteq \Ybot{Y}{a}$ we
    have
    \begin{align*}
      (c_k)_{a,\theta}^\#(Y')&= \gamma_{c_k(a),\theta}\circ c_k \circ
      \alpha_{a,\theta}(Y') \\
      &= \{ z\in \Ybot{Z}{c_k(a)} \mid \theta\sqsubseteq c_k(a\oplus \theta_{Y'})(z)\ominus c_k(a)(z)\} \\
      &=\{ z\in \Ybot{Z}{c_k(a)} \mid \theta\sqsubseteq k\ominus k\} = \{ z\in Z \mid \theta\sqsubseteq 0\} = \emptyset %
    \end{align*}
    Hence all values $\inc{a}{f}(Y',z)$ are equal to $0$ and we have
    $\inc{a}{f} = \delta_a$. Replacing $\theta$ by
    $\inc{a}{f}$ we obtain $(c_k)_a^\#(Y') = \emptyset$.
  \item \emph{Reindexing:} for all
    $0 \sqsubset \theta \sqsubseteq \delta_a$ and
    $Y' \subseteq \Ybot{Y}{a}$ we have
    \begin{align*}
      (u^*)_{a,\theta}^\#(Y') & = \gamma_{u^*(a),\theta} \circ u^*
      \circ \alpha_{a,\theta}(Y') \\
      & = \{z \in \Ybot{Z}{u^*(a)} \mid
        \theta \sqsubseteq (a \oplus \theta_{Y'})(u(z)) \ominus
        a(u(z))\}. 
    \end{align*}

    We show that this corresponds to
    $u^{-1}(Y') = \{z \in Z \mid u(z) \in Y'\}$. It is easy to see
    that for all $z \in u^{-1}(Y')$, we have
    \[ (a \oplus \theta_{Y'})(u(z)) \ominus a(u(z)) = \theta = a(u(z))
      \ominus (a \ominus \theta_{Y'})(u(z)) \] since $u(z) \in Y'$ and
    $\theta \sqsubseteq \delta_a$. Since
    $u(z)\in Y' \subseteq \Ybot{Y}{a}$, we have
    $u^*(a)(z) = a(u(z)) \neq 1$ and hence $z\in \Ybot{Z}{u^*(a)}$. On
    the other hand, for all $z \not\in u^{-1}(Y')$, we have
    \[ (a \oplus \theta_{Y'})(u(z)) = a(u(z)) = (a \ominus
      \theta_{Y'})(u(z)) \] since $u(z) \notin Y'$, and so
    \[ (a \oplus \theta_{Y'})(u(z)) \ominus a(u(z)) = a(u(z)) \ominus
      (a \ominus \theta_{Y'})(u(z)) = 0 \sqsubset \theta. \] Therefore
    $(u^*)_{a,\theta}^\#(Y') = u^{-1}(Y')$.

    We observe that for $Y'\subseteq \Ybot{Y}{a}$,
    $z\in \Ybot{Z}{u^*(a)}$ either
    $u^*(a\oplus \theta_{Y'})(z)\ominus u^*(a)(z) \sqsubset \theta$
    for all $0 \sqsubset \theta\sqsubseteq \delta_a$ -- and in this
    case $\inc{a}{u^*}(Y',z) = 0$ -- or
    $u^*(a\oplus \theta_{Y'})(z)\ominus u^*(a)(z) = \theta$ for all
    $0 \sqsubset\theta\sqsubseteq \delta_a$ -- and in this case
    $\inc{a}{u^*}(Y',z) = \delta_a$. By taking the minimum over all
    non-zero values, we get $\inc{a}{u^*} = \delta_a$.

    And finally we observe that
    $(u^*)_a^\#(Y') = (u^*)_{a,\inc{a}{u^*}}^\#(Y') = u^{-1}(Y')$.
  \item \emph{Minimum:} let $0\sqsubset \theta\sqsubseteq
    \delta_a$. For all $Y' \subseteq \Ybot{Y}{a}$ we have
    \begin{align*}
      (\mins_\mathcal{R})_{a,\theta}^\#(Y') & =
      \gamma_{\mins_\mathcal{R}(a),\theta} \circ \mins_\mathcal{R}
      \circ \alpha_{a,\theta}(Y')\\
      & = \{z \in \Ybot{Z}{\mins_\mathcal{R}(a)} \mid \theta
      \sqsubseteq \min_{y\mathcal{R}z} (a \oplus \theta_{Y'})(y)
      \ominus \min_{y\mathcal{R}z} a(y)\}
    \end{align*}
    We compute the value $V = \min_{y\mathcal{R}z} (a
    \oplus \theta_{Y'})(y) \ominus \min_{y\mathcal{R}z} a(y)$ and
    consider the following cases:
    \begin{itemize}
    \item Assume that there exists $\hat{y}\in
      \arg\min_{y\in\mathcal{R}^{-1}(z)} a(y)$ where $\hat{y}\not\in
      Y'$.
      
      Then $(a\oplus \theta_{Y'})(\hat{y}) = a(\hat{y}) \sqsubseteq
      a(y) \sqsubseteq (a\oplus
      \theta_{Y'})(y)$ for all
      $y\in\mathcal{R}^{-1}(z)$, which implies that
      $\min_{y\mathcal{R}z} (a \oplus \theta_{Y'})(y) =
      a(\hat{y})$. We also have $\min_{y\mathcal{R}z} a(y) =
      a(\hat{y})$ and hence $V = 0$.
    \item Assume that
      $\arg\min_{y\in\mathcal{R}^{-1}(z)} a(y)\subseteq Y'$ and
      $\theta\sqsubseteq a(y) \ominus a(\hat{y})$ in all cases where
      $\hat{y}\in \arg\min_{y\in\mathcal{R}^{-1}(z)} a(y)$,
      $y\not\in Y'$ and $y\mathcal{R}z$.

      Since $\arg\min_{y\in\mathcal{R}^{-1}(z)} a(y)\subseteq Y'$
      we observe that
      \begin{eqnarray*}
        \min_{y\mathcal{R}z} (a \oplus \theta_{Y'})(y) & = &
        \min\{\min_{y\in \arg\min_{y\in\mathcal{R}^{-1}(z)} a(y)}
        (a(y)\oplus \theta), \min_{y\mathcal{R}z, y\not\in Y'} a(y)\}
      \end{eqnarray*}
      We can omit the values of all $y$ with $y\mathcal{R}z$,
      $y\not\in \arg\min_{y\in\mathcal{R}^{-1}(z)} a(y)$, $y\in Y'$,
      since we will never attain the minimum there.

      Now let $\hat{y}\in \arg\min_{y\in\mathcal{R}^{-1}(z)} a(y)$ and
      $y$ with $y\mathcal{R}z$ and $y\not\in Y'$. Then
      $\theta \sqsubseteq a(y)\ominus a(\hat{y})$ by assumption, which
      implies $a(\hat{y})\oplus \theta \sqsubseteq a(y)$, since
      $a(\hat{y})\sqsubseteq a(y)$ and
      Lemma~\ref{le:mvprop}(\ref{le:mvprop:2}) holds.

      From this we can deduce
      $\min_{y\mathcal{R}z} (a \oplus \theta_{Y'})(y) = a(\hat{y})
      \oplus \theta$. We also have
      $\min_{y\mathcal{R}z} a(y) = a(\hat{y})$ and hence -- since
      $a(\hat{y})\sqsubseteq \comp{\theta}$ (due to
      $\theta\sqsubseteq \delta_a \sqsubseteq \comp{a(\hat{y})}$) and
      Lemma~\ref{le:mvprop}(\ref{le:mvprop:9}) holds --
      $V = (a(\hat{y}) \oplus \theta)\ominus a(\hat{y}) = \theta$.
    \item In the remaining case
      $\arg\min_{y\in\mathcal{R}^{-1}(z)} a(y)\subseteq Y'$ and there
      exist $\hat{y}\in \arg\min_{y\in\mathcal{R}^{-1}(z)} a(y)$,
      $y\not\in Y'$, $y\mathcal{R}z$ such that
      $a(y) \ominus a(\hat{y}) \sqsubset \theta$.

      This implies
      $a(y) \sqsubseteq (a(y)\ominus a(\hat{y})) \oplus a(\hat{y})
      \sqsubset \theta \oplus a(\hat{y})$ since again
      $a(\hat{y})\sqsubseteq \comp{\theta}$ and
      Lemma~\ref{le:mvprop}(\ref{le:mvprop:8}) holds. Hence
      $\min_{y\mathcal{R}z} (a \oplus \theta_{Y'})(y) \sqsubseteq
      a(y)$, which means that
      $V \sqsubseteq a(y)\ominus a(\hat{y}) \sqsubset \theta$.
    \end{itemize}
    Summarizing, for $\theta\sqsubseteq \delta_a$ we observe that
    $V = \theta$ if and only if
    $\arg\min_{y\in\mathcal{R}^{-1}(z)} a(y)\subseteq Y'$ and
    $\theta \sqsubseteq a(y) \ominus a(\hat{y})$ whenever
    $\hat{y}\in \arg\min_{y\in\mathcal{R}^{-1}(z)} a(y)$,
    $y\not\in Y'$ and $y\mathcal{R}z$.

    Hence if $\arg\min_{y\in\mathcal{R}^{-1}(z)} a(y)\subseteq Y'$ we
    have
    \[ \inc{a}{\mins_{\mathcal{R}}}(Y',z) = \min\{a(y) \ominus
      a(\hat{y}) \mid \hat{y}\in \arg\min_{y\in\mathcal{R}^{-1}(z)} a(y),
      y\not\in Y', y\mathcal{R}z\} \cup \{\delta_a\} \]
    otherwise $\inc{a}{\mins_{\mathcal{R}}}(Y',z) = 0$.
    
    The values above are minimal whenever
    $Y' = \arg\min_{y\in\mathcal{R}^{-1}(z)} a(y)$ and thus we have:
    \[ \inc{a}{\mins_{\mathcal{R}}} =
      \min_{z\in\Ybot{Z}{\mins_{\mathcal{R}}(a)}} \{ a(y) \ominus
          a(\hat{y}) \mid y\mathcal{R}z, \hat{y}\in
          \arg\min_{y\in\mathcal{R}^{-1}(z)} a(y), y\not\in
          \arg\min_{y\in\mathcal{R}^{-1}(z)} a(y) \} \cup
          \{\delta_a\}. \] Finally we deduce that
    \[ (\mins_\mathcal{R})_a^\#(Y') =
      (\mins_\mathcal{R})_{a,\inc{a}{\mins_{\mathcal{R}}}}^\#(Y') =
      \{z \in \Ybot{Z}{\mins_\mathcal{R}(a)} \mid
      \arg\min_{y\in\mathcal{R}^{-1}(z)} a(y) \subseteq Y'\}. \]

  \item \emph{Maximum:} let $0\sqsubset \theta\sqsubseteq
    \delta_a$. For all $Y' \subseteq \Ybot{Y}{a}$ we have
    \begin{align*}
      (\maxs_\mathcal{R})_{a,\theta}^\#(Y') & =
      \gamma_{\maxs_\mathcal{R}(a),\theta} \circ \maxs_\mathcal{R}
      \circ \alpha_{a,\theta}(Y')\\
      & = \{z \in \Ybot{Z}{\maxs_\mathcal{R}(a)} \mid \theta
      \sqsubseteq \max_{y\mathcal{R}z} (a \oplus \theta_{Y'})(y)
      \ominus \max_{y\mathcal{R}z} a(y)\}
    \end{align*}
    We observe that 
    \begin{eqnarray*}
      \max_{y\mathcal{R}z} (a \oplus \theta_{Y'})(y) & = &
      \max\{\max_{y\in \arg\max_{y\in\mathcal{R}^{-1}(z)} a(y)}
      (a\oplus \theta_{Y'})(y), \max_{\substack{y\mathcal{R}z, y\in Y'
       \\ y\not\in \arg\max_{y\in\mathcal{R}^{-1}(z)} a(y)}}
      (a(y) \oplus \theta)\}
    \end{eqnarray*}
    We can omit the values of all $y$ with $y\mathcal{R}z$,
    $y\not\in \arg\max_{y\in\mathcal{R}^{-1}(z)} a(y)$, $y\not\in Y'$,
    since we will never attain the maximum there.
    
    We now compute the value $V = \max_{y\mathcal{R}z} (a
    \oplus \theta_{Y'})(y) \ominus \max_{y\mathcal{R}z} a(y)$ and
    consider the following cases:
    \begin{itemize}
    \item Assume that there exists
      $\hat{y}\in \arg\max_{y\in\mathcal{R}^{-1}(z)} a(y)$ where
      $\hat{y}\in Y'$.
      
      Then
      $(a\oplus \theta_{Y'})(\hat{y}) = a(\hat{y}) \oplus\theta
      \sqsupseteq (a\oplus \theta_{Y'})(y) \sqsupseteq a(y)$ for all
      $y\in\mathcal{R}^{-1}(z)$, which implies that
      $\max_{y\mathcal{R}z} (a \oplus \theta_{Y'})(y) = a(\hat{y})
      \oplus \theta$. We also have
      $\max_{y\mathcal{R}z} a(y) = a(\hat{y})$ and hence -- since
      $a(\hat{y})\sqsubseteq \comp{\theta}$ and
      Lemma~\ref{le:mvprop}(\ref{le:mvprop:9}) holds --
      $V = (a(\hat{y}) \oplus \theta)\ominus a(\hat{y}) = \theta$.
    \item Assume that
      $\arg\max_{y\in\mathcal{R}^{-1}(z)} a(y)\cap Y' = \emptyset$.
      Now let $\hat{y}\in \arg\max_{y\in\mathcal{R}^{-1}(z)} a(y)$ and
      $y\not\in \arg\max_{y\in\mathcal{R}^{-1}(z)} a(y)$ with
      $y\mathcal{R}z$ and $y\in Y'$. Then
      \begin{eqnarray*}
        \max_{y\in \arg\max_{y\in\mathcal{R}^{-1}(z)} a(y)} (a\oplus
        \theta_{Y'})(y) & = & a(\hat{y}) \\
        \max_{\substack{y\mathcal{R}z, y\in Y'
            \\ y\not\in \arg\max_{y\in\mathcal{R}^{-1}(z)} a(y)}}
        (a(y) \oplus \theta) & = & a(y') \oplus \theta
      \end{eqnarray*}
      for some value $y'$ with $y'\mathcal{R}z$, $y'\in Y'$,
      $y'\not\in \arg\max_{y\in\mathcal{R}^{-1}(z)} a(y)$, that is
      $a(y')\sqsubset a(\hat{y})$.

      So then either
      $\max_{y\mathcal{R}z} (a \oplus \theta_{Y'})(y) = a(\hat{y})$
      and $V = a(\hat{y})\ominus a(\hat{y}) = 0$. Or
      $\max_{y\mathcal{R}z} (a \oplus \theta_{Y'})(y) = a(y') \oplus
      \theta$ and by
      Lemma~\ref{le:mvprop}(\ref{le:mvprop:11})
      $V = (a(y')\oplus \theta) \ominus a(\hat{y}) \sqsubset \theta$.
    \end{itemize}
    Summarizing, for $\theta\sqsubseteq \delta_a$ we observe that
    $V = \theta$ if and only if
    $\arg\max_{y\in\mathcal{R}^{-1}(z)} a(y)\cap Y'\neq\emptyset$,
    where the latter condition is independent of $\theta$.

    Hence, as in the case of reindexing, we have
    $\inc{a}{\maxs_\mathcal{R}} = \delta_a$. Finally we have
    \[ (\maxs_\mathcal{R})_a^\#(Y') =
      (\maxs_\mathcal{R})_{a,\inc{a}{\maxs_{\mathcal{R}}}}^\#(Y') =
      \{z \in \Ybot{Z}{\maxs_\mathcal{R}(a)} \mid
      \arg\max_{y\in\mathcal{R}^{-1}(z)} a(y) \cap Y' \neq
      \emptyset\}. \]

  \item \emph{Average:} 
    for all
    $0 \sqsubset \theta \sqsubseteq \delta_a$ and
    $Y' \subseteq \Ybot{Y}{a}$ by definition
    \begin{align*}
      (\mathrm{av}_D)_{a,\theta}^\#(Y') & = \gamma_{\mathrm{av}_D(a),\theta} \circ \mathrm{av}_D \circ \alpha_{a,\theta}(Y')\\
      & = \{p \in \Ybot{D}{\mathrm{av}_D(a)} \mid \theta \sqsubseteq \bigoplus_{y\in Y} p(y) \odot
      (a \oplus \theta_{Y'})(y) \ominus \bigoplus_{y\in Y} p(y) \odot
      a(y)\}
    \end{align*}
    We show that this set corresponds to
    $\{p \in \Ybot{D}{\mathrm{av}_D(a)} \mid \mathit{supp}(p) \subseteq Y'\}$.

    Consider
    $p \in \Ybot{D}{\mathrm{av}_D(a)}$ such that $\mathit{supp}(p) \subseteq Y'$.
    Note that clearly $\bigoplus_{y \in Y'} p(y) = 1$. 
    Now we have
    \begin{align*}
      & \bigoplus_{y\in Y} p(y) \odot (a \oplus \theta_{Y'})(y) \ominus \bigoplus_{y\in Y} p(y) \odot a(y)\\
      & = \bigoplus_{y\in Y'} p(y) \odot (a(y) \oplus \theta) \oplus \bigoplus_{y\in Y \backslash Y'} p(y) \odot a(y) \ominus \bigoplus_{y\in Y} p(y) \odot a(y)\\
      & = \bigoplus_{y\in Y'} (p(y) \odot a(y) \oplus p(y) \odot \theta) \oplus \bigoplus_{y\in Y \backslash Y'} p(y) \odot a(y) \ominus \bigoplus_{y\in Y} p(y) \odot a(y)\\
      & \qquad
      \mbox{[by weak distributivity, since for $y \in Y' \subseteq \Ybot{Y}{a}$, $a(y) \sqsubseteq \comp{\delta_a}$]}\\
      & = \bigoplus_{y\in Y'} p(y) \odot \theta \oplus \bigoplus_{y\in Y'} p(y) \odot a(y) \oplus \bigoplus_{y\in Y \backslash Y'} p(y) \odot a(y) \ominus \bigoplus_{y\in Y} p(y) \odot a(y)\\
      & = \bigoplus_{y\in Y'} p(y) \odot \theta \oplus \bigoplus_{y\in Y'} p(y) \odot a(y) \ominus \bigoplus_{y\in Y'} p(y) \odot a(y)\\
      & \qquad
        \mbox{[since, for $y \not\in Y' \supseteq \mathit{supp}(p)$,  $p(y)=0$  and thus $p(y) \odot a(y) = 0$]}\\
      & = (\bigoplus_{y\in Y'} p(y)) \odot \theta \oplus \bigoplus_{y\in Y'} p(y) \odot a(y) \ominus \bigoplus_{y\in Y'} p(y) \odot a(y)\\
      & \qquad
        \mbox{[by weak distributivity, since $p$ is a distribution]}\\
      & = 1 \odot \theta \oplus \bigoplus_{y\in Y'} p(y) \odot a(y) \ominus \bigoplus_{y\in Y'} p(y) \odot a(y)\\
      & \qquad
        \mbox{[since $p$ is a distribution]}\\
      & = \theta \oplus \bigoplus_{y\in Y'} p(y) \odot a(y) \ominus
      \bigoplus_{y\in Y'} p(y) \odot a(y) \\
      & = \theta
    \end{align*}
    In order to motivate the last passage, observe that for all
    $y \in Y' \subseteq \Ybot{Y}{a}$, we have
    $a(y) \sqsubseteq \comp{\delta_a}$, and thus
    $\bigoplus_{y \in Y'} p(y) \odot a(y) \sqsubseteq \bigoplus_{y \in
      Y'} p(y) \odot \comp{\delta_a} = (\bigoplus_{y \in Y'} p(y))
    \odot \comp{\delta_a} = 1 \odot \comp{\delta_a} =
    \comp{\delta_a}$, where the third last passage is motivated by
    weak distributivity. Since $\theta \sqsubseteq \delta_a$, by
    Lemma~\ref{le:mvprop}(\ref{le:mvprop:3}), we have
    $\comp{\delta_a} \sqsubseteq \comp{\theta}$ and thus
    $\bigoplus_{y \in Y'} p(y) \odot a(y) \sqsubseteq \comp{\theta}$.
    In turn, using this fact, 
    Lemma~\ref{le:mvprop}(\ref{le:mvprop:9}) motivates the last equality in the chain above, i.e.,
    $\theta \oplus \bigoplus_{y\in Y'} p(y) \odot a(y) \ominus
    \bigoplus_{y\in Y'} p(y) \odot a(y) = \theta$.

    \smallskip
    
    On the other hand, for all $p \in \Ybot{D}{\mathrm{av}_D(a)}$ such
    that $\mathit{supp}(p) \not\subseteq Y'$, there exists
    $y' \in Y \backslash Y'$ such that $p(y') \neq 0$. Then, we
    have
    \begin{align*}
      & \bigoplus_{y\in Y} p(y) \odot (a \oplus \theta_{Y'})(y) \ominus \bigoplus_{y\in Y} p(y) \odot a(y)\\
      & = \bigoplus_{y\in Y'} p(y) \odot (a(y) \oplus \theta) \oplus \bigoplus_{y\in Y \backslash Y'} p(y) \odot a(y) \ominus \bigoplus_{y\in Y} p(y) \odot a(y)\\
      & = \bigoplus_{y\in Y'} p(y) \odot \theta \oplus \bigoplus_{y\in Y'} p(y) \odot a(y) \oplus \bigoplus_{y\in Y \backslash Y'} p(y) \odot a(y) \ominus \bigoplus_{y\in Y} p(y) \odot a(y)\\
      & \qquad
        \mbox{[by weak distributivity, since for $y \in Y' \subseteq \Ybot{Y}{a}$, $a(y) \sqsubseteq \comp{\delta_a}$]}\\
      & = \bigoplus_{y\in Y'} p(y) \odot \theta \oplus \bigoplus_{y\in Y} p(y) \odot a(y) \ominus \bigoplus_{y\in Y} p(y) \odot a(y)\\
      & \sqsubseteq \bigoplus_{y\in Y'} p(y) \odot \theta\\
      & \qquad 
        \mbox{[by Lemma~\ref{le:mvprop}(\ref{le:mvprop:6})]}\\
      & = \theta \odot \bigoplus_{y\in Y'} p(y)\\
      & \mbox{[by weak distributivity, since $p$ is a distribution]}\\
      & \sqsubset \theta
    \end{align*}
    In order to motivate the last inequality, we proceed as
    follows. We have that $\mathit{supp}(p) \not\subseteq Y'$. Let
    $y_0 \in \mathit{supp}(p) \backslash Y'$. We know that
    $\comp{p(y_0)} \sqsubseteq \bigoplus_{y \in Y \backslash \{y\}} p(y)
    \sqsubseteq \bigoplus_{y \in Y'} p(y)$. Therefore
    $\comp{\bigoplus_{y \in Y'} p(y)} \sqsubseteq p(y_0)
    \neq 0$. Hence
    $\bigoplus_{y \in Y'} p(y) \sqsubset 1$.
  
    The strict inequality above now follows, if we further show that
    given an $x \in \monM$, $x \neq 1$ then
    $\theta \odot x \sqsubset \theta$. Note that $\comp{x} \neq
    0$. Therefore
    $\theta = \theta \odot 1 = \theta \odot (x \oplus \comp{x}) =
    \theta \odot x \oplus \theta \odot \comp{x}$, where the last
    equality follows by weak distributivity. Now
    $\theta \odot \comp{x} \sqsubseteq \comp{x} \sqsubseteq
    \comp{\theta \odot x}$, and thus, by
    Lemma~\ref{le:mvprop}(\ref{le:mvprop:9}), we obtain
    $\theta \odot x = \theta \odot x \oplus \theta \odot \comp{x}
    \ominus \theta \odot \comp{x} = \theta \ominus \theta \odot
    \comp{x} \sqsubset \theta$, as desired. The last passage follows by the fact that $\theta, \comp{x} \neq 0$ and thus $\theta \odot
    \comp{x} \neq 0$.

    Since these results hold for all $\theta\sqsubseteq \delta_a$, we
    have $\inc{a}{\mathrm{av}_D} = \delta^a$.
    
    And finally
    $(\mathrm{av}_D)_{a,\theta}^\#(Y') =
    (\mathrm{av}_D)^{a,\theta}_\#(Y') = \{p \in
    \Ybot{D}{\mathrm{av}_D(a)} \mid \mathit{supp}(p) \subseteq Y'\}$.
    \qedhere
  \end{itemize}
\end{proof}

When a non-expansive function arises as the composition of simpler
ones (see Lemma~\ref{le:comp-non-expansive}) we can obtain the
corresponding approximation by just composing the approximations of
the simpler functions.

\begin{proposition}[composing approximations]
  \label{pr:compose-approximation}
  Let $g : \monM^Y \to \monM^W$ and $h : \monM^W \to \monM^Z$ be
  non-expansive functions. For all $a \in \monM^Y$ we have that
  $(h \circ g)_a^\# = h_{g(a)}^\# \circ g_a^\#$. Analogously
  $(h \circ g)^a_\# = h^{g(a)}_\# \circ g^a_\#$ for the dual case.
\end{proposition}

\begin{proof}
  Here we only consider the primal case, the dual case for
  $(h\circ g)_\#^a$ is analogous.
  
  Let
  $0\sqsubset \theta \sqsubseteq \min
  \{\inc{a}{g},\inc{g(a)}{h}\}$. Then, by
  Theorem~\ref{th:approximation}(\ref{th:approximation:2}) we know
  that
  \[ g_a^\# = g_{a,\theta}^\# = \gamma_{g(a),\theta} \circ g \circ \alpha_{a,\theta} \]
  \[ h_{g(a)}^\# = h_{g(a),\theta}^\# = \gamma_{h(g(a)),\theta} \circ h \circ \alpha_{g(a),\theta} \]
  Now we will prove that
  \[ (h \circ g)_{a,\theta}^\# = h_{g(a),\theta}^\# \circ
    g_{a,\theta}^\# \] First observe that
  $g(\alpha_{a,\theta}(Y')) \in \interval{g(a)}{g(a \oplus \theta)}
  \subseteq \interval{g(a)}{g(a) \oplus \theta}$ for all
  $Y' \subseteq \Ybot{Y}{a}$ by
  Lemma~\ref{le:non-expansiveness}. Applying
  Theorem~\ref{th:approximation}(\ref{th:approximation:2}) on $h$ we
  obtain
  \begin{eqnarray*}
    && (h \circ g)_{a,\theta}^\# = \gamma_{h(g(a)),\theta} \circ h
    \circ g \circ \alpha_{a,\theta}(Y') = h_{g(a),\theta}^\# \circ
    \gamma_{g(a),\theta} \circ g \circ \alpha_{a,\theta}(Y') \\
    & = &
    h_{g(a),\theta}^\# \circ g_{a,\theta}^\#(Y') = h_{g(a)}^{\#}
    \circ g_a^{\#}(Y') 
  \end{eqnarray*}
  Hence all functions $(h \circ g)_{a,\theta}^\#$ are equal and
  independent of $\theta$ and so it must hold that
  $(h \circ g)_{a,\theta}^\# = (h \circ g)_a^\#$. Then from
  Theorem~\ref{th:approximation} we can conclude
  $\min \{\inc{a}{g},\inc{g(a)}{h}\} \sqsubseteq \inc{a}{h\circ g}$. \qedhere
\end{proof}

Furthermore functions can be combined via disjoint union, preserving
non-expansiveness, as follows.

\begin{proposition}[disjoint union of non-expansive functions]
  \label{pr:compose-disjoint}
  Let $f_i : \monM^{Y_i} \to \monM^{Z_i}$, for $i\in I$, be
  non-expansive and such that the sets $Z_i$ are pairwise
  disjoint. The function
  $\biguplus\limits_{i\in I} f_i : \monM^{\bigcup_{i\in I} Y_i} \to
  \monM^{\biguplus_{i\in I} Z_i}$ defined by
  \[ \biguplus_{i\in I} f_i(a)(z) = f_i(a|_{Y_i})(z) \qquad \mbox{if
      $z\in Z_i$} \] is non-expansive.
\end{proposition}

\begin{proof}
  For all $a,b \in \monM^{\bigcup_{i\in I} Y_i}$ we have
  \begin{align*}
    & \norm{\biguplus_{i\in I} f_i(b) \ominus \biguplus_{i\in I} f_i(a)} & \\
    & = \max_{z\in\biguplus_{i\in I} Z_i} (\biguplus_{i\in I} f_i(b)(z) \ominus \biguplus_{i\in I} f_i(a)(z)) & \\
    & = \max_{i \in I} \max_{z\in Z_i} (f_i(b|_{Y_i})(z) \ominus f_i(a|_{Y_i})(z)) & \mbox{[since all $Z_i$ are disjoint]}\\
    & = \max_{i \in I} \norm{f_i(b|_{Y_i}) \ominus f_i(a|_{Y_i})} & \mbox{[by def.\ of norm]}\\
    & \sqsubseteq \max_{i \in I} \norm{b|_{Y_i} \ominus a|_{Y_i}} & \mbox{[since all $f_i$ are non-expansive]}\\
    & = \max_{i \in I} \max_{y\in Y_i} (b(y) \ominus a(y)) & \\
    & = \max_{y\in\bigcup_{i\in I} Y_i} (b(y) \ominus a(y)) & \\
    & = \norm{b \ominus a} & \mbox{[by def.\ of norm]} \tag*{\qedhere}
  \end{align*}
\end{proof}

Also, the corresponding approximation of a disjoint union can be
conveniently assembled from the approximations of its components.

\begin{proposition}[disjoint union and approximations]
  \label{pr:approximation-union-fun}
  The approximations for $\biguplus\limits_{i\in I} f_i$, where
  $f_i : \monM^{Y_i} \to \monM^{Z_i}$ are non-expansive and $Z_i$ are
  pairwise disjoint, have the following form. For all
  $a\colon {\bigcup_{i\in I} Y_i} \to \monM$ and
  $Y' \subseteq \bigcup_{i\in I} Y_i$:
  \[ \big(\biguplus_{i\in I} f_i\big)_a^\#(Y') = \biguplus_{i\in I}
    (f_i)_{a|_{Y_i}}^\#(Y'\cap Y_i) \qquad \big(\biguplus_{i\in I}
    f_i\big)^a_\#(Y') = \biguplus_{i\in I} (f_i)^{a|_{Y_i}}_\#(Y'\cap
    Y_i) \]
\end{proposition}

\begin{proof}
  We just show the statement for the primal case, the dual case is
  analogous. We abbreviate $Y = \bigcup_{i\in I} Y_i$.

  Let $0\sqsubset \theta \sqsubseteq \delta_a$. According to the
  definition of $a$-approximation (Lemma~\ref{le:approximation-2}) we
  have for $Y'\subseteq \Ybot{Y}{a}$:
  \[ \big(\biguplus_{i\in I} f_i\big)_{a,\theta}^\#(Y') = \gamma_{\biguplus\limits_{i\in I}f_i(a),\theta} \circ \biguplus\limits_{i\in I}f_i \circ \alpha_{a,\theta} \]
  \[ (f_i)_{a|_{Y_i},\theta}^\# = \gamma_{f_i(a|_{Y_i}),\theta} \circ
    f_i \circ \alpha_{a|_{Y_i},\theta} \]
  for all $i \in I$. Our first step is prove that
  \[ \gamma_{\biguplus\limits_{i\in I}f_i(a),\theta} \circ \biguplus\limits_{i\in I}f_i \circ \alpha_{a,\theta}(Y') = \biguplus_{i\in I}\gamma_{f_i(a|_{Y_i}),\theta} \circ f_i \circ \alpha_{a|_{Y_i},\theta}(Y'\cap Y_i) \]
  By simply expanding the functions we obtain
  \[ \gamma_{\biguplus\limits_{i\in I}f_i(a),\theta} \circ \biguplus\limits_{i\in I}f_i \circ \alpha_{a,\theta}(Y') = \{z \in Z_i \mid i \in I\ \land\ \theta \sqsubseteq f_i((a \oplus \theta_{Y'})|_{Y_i})(z) \ominus f_i(a|_{Y_i})(z) \} \]
  \[ \biguplus_{i\in I}\gamma_{f_i(a|_{Y_i}),\theta} \circ f_i \circ \alpha_{a|_{Y_i},\theta}(Y'\cap Y_i) = \biguplus_{i \in I} \{z \in Z_i \mid \theta \sqsubseteq f_i(a|_{Y_i} \oplus \theta_{Y' \cap Y_i})(z) \ominus f_i(a|_{Y_i})(z) \} \]
  which are the same set, since for all $i \in I$ clearly $(a \oplus
  \theta_{Y'})|_{Y_i} = a|_{Y_i} \oplus \theta_{Y' \cap Y_i}$.

  This implies
  \[ \big(\biguplus_{i\in I} f_i\big)_{a,\theta}^\#(Y') =
    \biguplus_{i\in I} (f_i)_{a|_{Y_i},\theta}^\#(Y'\cap Y_i). \]
  Whenever $\theta\sqsubseteq \min\limits_{i \in I} \inc{a}{f_i}$,
  this can be rewritten to
  \[ \big(\biguplus_{i\in I} f_i\big)_{a,\theta}^\#(Y') =
    \biguplus_{i\in I} (f_i)_{a|_{Y_i}}^\#(Y'\cap Y_i). \] All functions
  $\big(\biguplus_{i\in I} f_i\big)_{a,\theta}^\#$ are equal and
  independent of $\theta$ and so it must hold that
  $\big(\biguplus_{i\in I} f_i\big)_{a,\theta}^\#$
  $= \big(\biguplus_{i\in I} f_i\big)_a^\#$. Then with
  Theorem~\ref{th:approximation} we can also conclude
  $\min\limits_{i \in I} \inc{a}{f_i}\sqsubseteq
  \inc{a}{\biguplus_{i\in I} f_i}$.
\end{proof}

\end{toappendix}

\begin{table}[t]
  \normalsize
  \caption{\normalsize Basic functions $f\colon \monM^Y\to \monM^Z$ (constant,
    reindexing, minimum, maximum, average), function composition,
    disjoint union and the corresponding approximations
    $f_a^\#\colon \Pow{\Ybot{Y}{a}} \to \Pow{\Ybot{Z}{f(a)}}$,
    $f^a_\#\colon \Pow{\Ytop{Y}{a}} \to
    \Pow{\Ytop{Z}{f(a)}}$. 
  }
  \label{tab:basic-functions-approximations}
  \emph{Notation:}
  $\mathcal{R}^{-1}(z) = \{y\in Y \mid y\mathcal{R}z\}$,
  $\mathit{supp}(p) = \{y\in Y\mid p(y) \sqsupset 0\}$ for
  $p\in \mathcal{D}(Y)$, \\
  $\arg\min_{y\in Y'} a(y)$, resp.\ $\arg\max_{y\in Y'} a(y)$,
  the set of elements where $a|_{Y'}$ \\
  reaches the minimum, resp.\ the maximum,
  for $Y' \subseteq Y$ and $a\in\monM^Y$

  \smallskip

  \begin{center}
    \begin{tabular}{|l|l|l|}
      \hline function $f$ & definition of $f$ & $f_a^\#(Y')$ (above),
      $f_\#^a(Y')$ (below) \\
      \hline\hline
      $c_k$ & $f(a) = k$ & $\emptyset$ \\
      ($k\in\monM^Z$) && $\emptyset$ \\ \hline
      $u^*$ & $f(a) = a \circ u$ & $u^{-1}(Y')$ \\
      ($u\colon Z\to Y$) & & $u^{-1}(Y')$ \\
      \hline $\mins_\mathcal{R}$ &
      $f(a)(z) = \min\limits_{y\mathcal{R}z} a(y)$ &
      $\{z \in \Ybot{Z}{f(a)} \mid
      \arg\!\!\!\!\!\min\limits_{y\in\mathcal{R}^{-1}(z)}\!\!a(y)\subseteq Y'\}$ \\
      ($\mathcal{R}\subseteq Y\times Z$) & &
      $\{z \in \Ytop{Z}{f(a)} \mid
      \arg\!\!\!\!\!\min\limits_{y\in\mathcal{R}^{-1}(z)}\!\!a(y)\cap Y'\!\neq
      \emptyset\}$\mystrutbl \\ \hline $\maxs_\mathcal{R}$ &
      $f(a)(z) = \max\limits_{y\mathcal{R}z} a(y)$ &
      $\{z \in \Ybot{Z}{f(a)} \mid
      \arg\!\!\!\!\!\max\limits_{y\in\mathcal{R}^{-1}(z)}\!\!a(y)\cap Y'\!\neq
      \emptyset\}$ \\
      ($\mathcal{R}\subseteq Y\times Z$) &&
      $\{z \in \Ytop{Z}{f(a)} \mid
      \arg\!\!\!\!\!\max\limits_{y\in\mathcal{R}^{-1}(z)}\!\!a(y)\subseteq Y'\}$
      \mystrutbl \\
      \hline \mystrutab$\mathrm{av}_D$ \quad &
      $f(a)(p) = \bigoplus\limits_{y\in Y} p(y) \odot a(y)$ &
      $\{p \in \Ybot{D}{f(a)} \mid \mathit{supp}(p) \subseteq
      Y'\}$ \\
      ($Z = D \subseteq \mathcal{D}(Y)$) & &
      $\{p \in \Ytop{D}{f(a)} \mid \mathit{supp}(p) \subseteq Y'\}$ \\
      \hline\hline
      $h\circ g$ & $f(a) = h(g(a))$ & $h_{g(a)}^\# \circ
      g_a^\#(Y')$ \\
      ($g\colon \monM^Y\to \monM^W$, && $h^{g(a)}_\# \circ g^a_\#(Y')$ \\
      $h\colon \monM^W\to \monM^Z$) && \\ \hline
      $\biguplus\limits_{i\in I} f_i$ \quad $I$ finite &
      $f(a)(z) = f_i(a|_{Y_i})(z)$ &
      $\biguplus_{i\in I} (f_i)_{a|_{Y_i}}^\#(Y'\cap Y_i)$ \\
      ($f_i\colon \monM^{Y_i}\to \monM^{Z_i}$, &
      ($z\in Z_i$) &
      $\biguplus_{i\in I} (f_i)^{a|_{Y_i}}_\#(Y'\cap Y_i)$ 
      \\
      $Y = \bigcup\limits_{i\in I} Y_i$, %
      $Z = \biguplus\limits_{i\in I} Z_i$)\mystrutbl & & \\
      \hline
    \end{tabular}
  \end{center}
\end{table}

We can then prove the desired results (non-expansiveness and approximation) for the basic building blocks and their composition (all schematically reported in Table~\ref{tab:basic-functions-approximations}).

\begin{theoremrep}
  \label{th:bigtable}
  All basic functions in Def.~\ref{de:basic-functions}
  are
  non-expansive. Furthermore non-expansive functions are closed under
  composition and disjoint union. The approximations are the ones
  listed in the third column of Table~\ref{tab:basic-functions-approximations}.
\end{theoremrep}

\begin{proof}
  Follows directly from Propositions~\ref{pr:fct-nonexpansive},
  \ref{pr:associated-approximations},
  \ref{pr:compose-approximation}, \ref{pr:compose-disjoint},
  \ref{pr:approximation-union-fun} and
  Lemma~\ref{le:comp-non-expansive}.
\end{proof}

\begin{toappendix}
We can also specify the maximal decrease respectively increase that
is propagated (here we are using the notation of
Lemma~\ref{le:approximation-2}).

\begin{corollary}
  \label{cor:iota-examples}
  Let $f\colon \monM^Y \to \monM^Z$, $a\in \monM^Y$ and $\inc{a}{f}$
  be defined as in Lemma~\ref{le:approximation-2}. In the dual view we
  have
  $\inc{f}{a} = \min \{ \inc{f}{a}(Y',z) \mid Y' \subseteq Y\ \land\ z
  \in Z\ \land\ \inc{f}{a}(Y',z) \neq 0\} \cup \{ \delta^a\}$, where
  the set
  $\{ \theta \sqsubseteq \delta^a \mid z \in f^{a,\theta}_\#(Y') \}$
  has a maximum for each $z\in \Ytop{Z}{f(a)}$ and
  $Y'\subseteq \Ytop{Y}{a}$, that we denote by $\inc{f}{a}(Y',z)$.

  We consider the basic functions from
  Def.~\ref{de:basic-functions}, function composition
  as in Lemma~\ref{le:comp-non-expansive} and disjoint union as in
  Proposition~\ref{pr:compose-disjoint} and give the corresponding
  values for $\inc{a}{f}$ and $\inc{f}{a}$.

  For greatest fixpoints (primal case) we obtain:
  \begin{itemize}
  \item
    $\inc{a}{c_k} = \inc{a}{u^*} = \inc{a}{\max_\mathcal{R}} =
    \inc{a}{\mathrm{av}_D} = \delta^a$
  \item $\inc{a}{\min_\mathcal{R}}= \min\limits_{z\in
    \Ybot{Z}{\min_\mathcal{R}(a)}} \{ a(y)\ominus a(\hat{y}) \mid$ \\
    \mbox{}\hspace{4cm}$y\mathcal{R}z, y\notin
    \arg\min_{y\in\mathcal{R}^{-1}(z)} a(y), \hat{y}\in
    \arg\min_{y\in\mathcal{R}^{-1}(z)} a(y) \} \cup \{\delta^a\}$
  \item $\inc{a}{g\circ f} \sqsupseteq
    \min\{\inc{a}{f},\inc{f(a)}{g}\}$
  \item
    $\inc{a}{\biguplus_{i\in I} f_i} = \min_{i\in I}
    \inc{a|_{Y_i}}{f_i}$
  \end{itemize}
  For least fixpoints (dual case) we obtain:
  \begin{itemize}
  \item
    $\inc{c_k}{a} = \inc{u^*}{a} = \inc{\min_\mathcal{R}}{a} =
    \inc{\mathrm{av}_D}{a} = \delta_a$
  \item $\inc{\max_\mathcal{R}}{a} = \min\limits_{z\in
    \Ytop{Z}{\min_\mathcal{R}(a)}} \{ a(\hat{y})\ominus a(y) \mid$ \\
    \mbox{}\hspace{4cm}$y\mathcal{R}z, \hat{y}\in
    \arg\max_{y\in\mathcal{R}^{-1}(z)} a(y), y\notin
    \arg\max_{y\in\mathcal{R}^{-1}(z)} a(y) \} \cup \{\delta_a\}$
  \item $\inc{g\circ f}{a} \sqsupseteq
    \min\{\inc{f}{a},\inc{g}{f(a)}\}$
  \item
    $\inc{\biguplus_{i\in I} f_i}{a} = \min_{i\in I}
    \inc{f_i}{a|_{Y_i}}$
  \end{itemize}
\end{corollary}

\begin{proof}
  The values $\inc{a}{f}$ can be obtained by inspecting the proofs of
  Propositions~\ref{pr:associated-approximations},
  \ref{pr:compose-approximation} and
  \ref{pr:compose-disjoint}.

  It only remains to show that
  $\iota := \inc{a}{\biguplus_{i\in I} f_i} \sqsubseteq \min_{i\in I}
  \inc{a|_{Y_i}}{f_i}$ (cf. Proposition~\ref{pr:compose-disjoint}),
  which means showing
  $\iota \sqsubseteq \inc{a|_{Y_i}}{f_i}$
  for every $i\in I$. We abbreviate $\iota_i := \inc{a|_{Y_i}}{f_i}$.

  If $\iota\sqsupset \iota_i$ for some $i\in I$, we will
  find a $z\in \Ybot{Z_i}{f_i(a)}$ and $Y'\subseteq \Ybot{Y}{a}$, such
  that
  $z\in (f_i)^\#_{a|_{Y_i},\iota_i}(Y' \cap Y_i) =
  (f_i)^\#_{a|_{Y_i}}(Y' \cap Y_i)$ but
  $z\notin (f_i)_{a|_{Y_i},\iota}^\#(Y' \cap Y_i)$ by definition
  (cf. Lemma~\ref{le:approx-max-theta}). This is a contradiction since
  \begin{eqnarray*}
    z\in \biguplus_{i\in I} (f_i)^\#_{a|_{Y_i}}(Y' \cap Y_i) =
    \big(\biguplus_{i\in I} f_i\big)_a^\#(Y') = \big(\biguplus_{i\in
      I} f_i\big)_{a,\iota}^\#(Y') = \biguplus_{i\in I}
    (f_i)_{a|_{Y_i},\iota}^\#(Y'\cap Y_i)
  \end{eqnarray*}
  and since $z\in Z_i$,
  $z\not\in (f_i)_{a|_{Y_i},\iota}^\#(Y' \cap Y_i)$ and cannot be
  contained in the union.
  
  The arguments for the values $\inc{f}{a}$ in the dual case are
  analogous. \qedhere
\end{proof}
\end{toappendix}

\section{Applications}
\label{se:applications}

\subsection{Termination probability}
\label{sec:termination-probability}

We start by making the example from the introduction
(Section~\ref{se:introduction}) more formal. Consider a Markov chain
$(S, T, \suc)$, as defined in the introduction
(Fig.~\ref{fig:markov-chain}), where we restrict the codomain of
$\suc\colon S\backslash T \to \mathcal{D}(S)$ to
$D\subseteq \mathcal{D}(S)$, where $D$ is finite (to ensure that all
involved sets are finite). Furthermore let
$\mathcal{T}\colon [0,1]^S\to [0,1]^S$ be the function
(Fig.~\ref{fig:markov-chain}) whose least fixpoint $\mu\mathcal{T}$
assigns to each state its termination probability.

\begin{lemmarep}
  \label{le:char-TermProb}
  The function $\mathcal{T}$ can be written as
  \[ \mathcal{T} = (\suc^* \circ \mathrm{av}_D) \uplus c_k  \]
  where $k\colon T \to [0,1]$ is the constant function $1$ defined
  only on terminal states.
\end{lemmarep}

\begin{proof}
Let $t\colon S\to [0,1]$. For $s\in T$ we have
\begin{align*}
&((\suc^* \circ \mathrm{av}_D) \uplus c_k)(t)(s) \\
&= c_k(t)(s) &\mbox{[since $s\in T$]} \\
&= k(s) = 1 &\mbox{[by definition of $c_k$ and $k$]} \\
&= \mathcal{T}(t)(s) &\mbox{[since $s\in T$]}
\end{align*}
For $s\notin T$ we have
\begin{align*}
&((\suc^* \circ \mathrm{av}_D) \uplus c_k)(t)(s) \\
&= \suc^* \circ \mathrm{av}_D(t)(s) &\mbox{[since $s \notin T$]} \\
&= \mathrm{av}_D(t)(\suc(s)) &\mbox{[by definition of reindexing]} \\
&= \sum_{s'\in S} \suc(s)(s') \cdot t(s') &\mbox{[by definition of $\mathrm{av}_D$]} \\
&= \mathcal{T}(t)(s) &\mbox{[since $s \notin T$]} \tag*{\qedhere}
\end{align*}
\end{proof}

From this representation and Theorem~\ref{th:bigtable} it
is obvious that $\mathcal{T}$ is non-expansive.

\begin{lemmarep}
  \label{le:associated-function-TermProb}
  Given a function $t\colon S\to [0,1]$, the $t$-approximation for
  $\mathcal{T}$ in the dual sense is
  $\mathcal{T}_\#^t\colon \Pow{\Ytop{S}{t}} \to
  \Pow{\Ytop{S}{\mathcal{T}(t)}}$ with
  \[
  \mathcal{T}_\#^t(S') = \{s \in \Ytop{S}{\mathcal{T}(t)} \mid s \notin T \land \mathit{supp}(\suc(s)) \subseteq S'\}.
  \]
\end{lemmarep}
  
\begin{proof}
  In the following let $t \colon S \to [0,1]$ and $S' \subseteq \Ytop{S}{t}$.
  By Lemma~\ref{le:char-TermProb} we know that $\mathcal{T} = (\suc^* \circ \mathrm{av}_D) \uplus c_k$, then by Propositions~\ref{pr:approximation-union-fun}, \ref{pr:compose-approximation}, and~\ref{pr:associated-approximations} we have
  \begin{align*}
  \mathcal{T}_\#^t(S') &= ((\suc^* \circ \mathrm{av}_D) \uplus c_k)_\#^t(S') \\
  &= (\suc^* \circ \mathrm{av}_D)_\#^{t}(S') \cup (c_k)_\#^{t}(S') \\
  &= (\suc^*)_\#^{\mathrm{av}_D(t)} \circ (\mathrm{av}_D)_\#^{t}(S') \cup (c_k)_\#^{t}(S') \\
  &= \{s \in \Ytop{S \backslash T}{\suc^*(\mathrm{av}_D(t))} \mid \suc(s) \in \{q \in \Ytop{D}{\mathrm{av}_D(t)} \mid \mathit{supp}(q) \subseteq S'\} \} \cup \emptyset \\
  &= \{s \in \Ytop{S \backslash T}{\suc^*(\mathrm{av}_D(t))} \mid \suc(s) \in \Ytop{D}{\mathrm{av}_D(t)} \land \mathit{supp}(\suc(s)) \subseteq S'\}
  \end{align*}
  Observe that actually for all $s \in \Ytop{S \backslash T}{\suc^*(\mathrm{av}_D(t))}$ it always holds that $\suc(s) \in \Ytop{D}{\mathrm{av}_D(t)}$. In fact, since $s \in \Ytop{S \backslash T}{\suc^*(\mathrm{av}_D(t))}$ we must have that $\suc^*(\mathrm{av}_D(t))(s) = \mathrm{av}_D(t)(\suc(s)) \neq 0$, and thus $\suc(s) \in \{q \in D \mid \mathrm{av}_D(t)(q) \neq 0\} = \Ytop{D}{\mathrm{av}_D(t)}$. Therefore, we have that
  \begin{align*}
  &\{s \in \Ytop{S \backslash T}{\suc^*(\mathrm{av}_D(t))} \mid \suc(s) \in \Ytop{D}{\mathrm{av}_D(t)} \land \mathit{supp}(\suc(s)) \subseteq S'\} \\
  &= \{s \in \Ytop{S \backslash T}{\suc^*(\mathrm{av}_D(t))} \mid \mathit{supp}(\suc(s)) \subseteq S'\}
  \end{align*}
  Finally, the set above is the same as
  \[ \{s \in \Ytop{S}{\mathcal{T}(t)} \mid s \notin T \land \mathit{supp}(\suc(s)) \subseteq S'\} = \{s \in \Ytop{S \backslash T}{\mathcal{T}(t)} \mid \mathit{supp}(\suc(s)) \subseteq S'\} \]
  because, for all $s \in S \backslash T$, hence $s \notin T$, we have that $\mathcal{T}(t)(s) = \sum_{s'\in S} \suc(s)(s')\cdot t(s') = \suc^*(\mathrm{av}_D(t))(s)$, and so $\Ytop{S \backslash T}{\mathcal{T}(t)} = \Ytop{S \backslash T}{\suc^*(\mathrm{av}_D(t))}$.
\end{proof}

At this point we have all the ingredients needed to formalise the
application presented in the introduction. We refrain from repeating
the same example, but rather present a new example that allows us to
illustrate the question of the largest decrease for a fixpoint that
still guarantees a pre-fixpoint (the dual problem is treated in
Proposition~\ref{de:greatest-increase-set}).

\begin{exa}
  \label{ex:markov-no-greatest}
  Consider the following Markov chain where $S = \{x_1,x_2,x_3\}$ are
  non-terminal states. The least fixpoint of the underlying fixpoint
  function $\mathcal{T}$ is clearly the constant $0$, since no state
  can reach a terminal state.
 \begin{center}
    \scalebox{0.85}
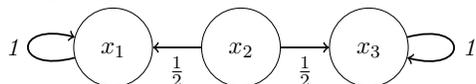
{
      \begin{tikzpicture}
        \node (x1) at (0,0) [circle,draw] {\begin{tabular}{c}
            $x_1$
          \end{tabular} };
                  \node (x2) at (2,0) [circle,draw] {\begin{tabular}{c}
            $x_2$
          \end{tabular} };
                            \node (x3) at (4,0) [circle,draw] {\begin{tabular}{c}
            $x_3$
          \end{tabular} };
          
        \draw[->,thick,   loop left] (x1) to node [left]  {1} (x1);
        \draw[->,thick,   loop right] (x3) to node [right]  {1} (x3);
        \draw[->,thick] (x2) to node [below]  {$\frac{1}{2}$} (x1);
        \draw[->,thick] (x2) to node [below]  {$\frac{1}{2}$} (x3);

      \end{tikzpicture}
    }
  \end{center}
  Now consider the function $t\colon S\to[0,1]$
  defined by $t(x_1) = 0.1$, $t(x_2) = 0.5$ and $t(x_3) = 0.9$. This
  is also a fixpoint of $\mathcal{T}$.

  Observe that $\mathcal{T}^t_\#(S) = S$ and thus, clearly,
  $\nu \mathcal{T}^t_\# = S$.  According to (the dual of)
  Def.~\ref{de:greatest-increase-set} we have $\delta^t(S) = 0.1$ and
  thus, by (the dual of) Proposition~\ref{pr:increases}, the function
  $t' = t \ominus (0.1)_S$, with $t'(x_1)=0$, $t'(x_2)=0.4$, and
  $t'(x_3) = 0.8$, is a pre-fixpoint. Indeed,
  $\mathcal{T}(t')(x_1) =0$, $\mathcal{T}(t')(x_2) =0.4$ and
  $\mathcal{T}(t')(x_3) =0.8$.

  This is not the largest decrease producing a pre-fixpoint. In fact,
  we can choose $\theta = 0.9$, greater than $\delta^{t}(S)$
  and we have that
  $t \ominus \theta_{S}$ is the constant $0$, i.e., the least fixpoint
  of $\mathcal{T}$.  However, if we take
  $\theta' = 0.5 \sqsubset \theta$, then $t \ominus \theta'_{S}$ is
  not a pre-fixpoint. In fact $(t \ominus \theta'_{S})(x_2) = 0$,
  while $\mathcal{T}(t\ominus \theta'_S)(x_2) =0.2$. This means that
  the set of decreases (beyond $\delta^{t}(S)$) producing a
  pre-fixpoint is not downward-closed and hence the largest decrease
  cannot be found by binary search, while, as already mentioned, a
  binary search will work for decreases below $\delta^{t}(S)$.
\end{exa}

It is well-known that the function $\mathcal{T}$ can be tweaked in
such a way that it has a unique fixpoint, coinciding with
$\mu\mathcal{T}$, by determining all states which cannot reach a
terminal state and setting their value to zero~\cite{bk:principles-mc}. Hence
fixpoint iteration from above does not really bring us any added value
here. It does however make sense to use the proof rule in order to
guarantee lower bounds via post-fixpoints.

Furthermore, termination probability is a special case of the
considerably more complex stochastic games that will be studied in
Section~\ref{se:ssgs}, where the trick of modifying the function is
not applicable.

\subsection{Behavioural metrics for probabilistic automata  and metric transition systems}
\label{sec:prob-automata}

We now consider behavioural metrics for probabilistic automata, which
involve both non-deterministic branching as well as probabilistic
branching. In addition, state labels can be taken from a metric space
in order to capture the fact that there can be a lower bound for the
distance related to some intrinsic features of the states.
As we will discuss, the model is sufficiently general to capture, as
instances, various kinds of probabilistic automata
in the literature~(e.g.,~\cite{bblmtv:prob-bisim-distance-automata}) as well as metric transition
systems~\cite{afs:linear-branching-metrics}.

We first consider the Kantorovich and Hausdorff
liftings and the corresponding approximations, which will play a major
role in the treatment of probabilistic automata.

\paragraph*{Kantorovich lifting.} The Kantorovich (also known as
Wasserstein) lifting converts a metric on $X$ to a metric on
probability distributions over $X$. Actually, as it commonly happens,
we will define the lifting for general distance functions on $[0,1]$,
not restricting to metrics.

In order to ensure finiteness of all the sets involved, we restrict to
$D\subseteq \mathcal{D}(X)$, some finite set of probability
distributions over $X$. A \emph{coupling} of $p,q\in D$ is a
probability distribution $c\in \mathcal{D}(X\times X)$ whose left and
right marginals are $p,q$, i.e.,
$p(x_1) = m_c^L(x_1) := \sum_{x_2\in X} c(x_1,x_2)$ and
$q(x_2) = m_c^R(x_2) := \sum_{x_1\in X} c(x_1,x_2)$. The set of all
couplings of $p,q$, denoted by $\Omega(p,q)$, forms a polytope with
finitely many vertices~\cite{pc:computational-ot}. The set of all
polytope vertices that are obtained by coupling any $p,q\in D$ is also
finite and is denoted by
$\mathit{VP}_D\subseteq \mathcal{D}(X\times X)$.

The Kantorovich lifting is defined as
$\mathcal{K} : [0,1]^{X\times X} \to [0,1]^{D\times D}$ where
\[ \mathcal{K}(d)(p,q) = \min_{c\in\Omega(p,q)} \sum_{(x_1,x_2)\in
    X\times X} c(x_1,x_2) \cdot d(x_1,x_2). \] The coupling $c$ can be
interpreted as the optimal transport plan to move goods from suppliers
to customers~\cite{v:optimal-transport}.
Below we provide an alternative characterisation, which shows
non-expansiveness of $\mathcal{K}$ and allows one to derive its
approximations.

\begin{lemmarep}
  Let $u : \mathit{VP}_D\to D\times D$, $u(c) = (m_c^L,m_c^R)$. 
  Then
  \[ \mathcal{K} = \mins_u \circ \mathrm{av}_{\mathit{VP}_D}
  \]
  where
  $\mathrm{av}_{\mathit{VP}_D}\colon [0,1]^{X\times
    X}\to[0,1]^{\mathit{VP}_D}$,
  $\min_u\colon [0,1]^{\mathit{VP}_D} \to [0,1]^{D\times D}$.
\end{lemmarep}

\begin{proof}
  It holds that $u^{-1}(p,q) = \Omega(p,q)\cap \textit{VP}_D$ for
  $p,q\in D$. Furthermore note it is sufficient to consider as
  couplings the vertices, i.e., the elements of $\mathit{VP}_D$, since
  the minimum is always attained there~\cite{pc:computational-ot}.
  
  Hence we obtain for $d\colon X\times X\to [0,1]$, $p,q\in D$:
  \begin{align*}
    \mins_u(\mathrm{av}_{\mathit{VP}_D}(d))(p,q) & =
    \min_{c\in\Omega(p,q)\cap \mathit{VP}_D}
    \mathrm{av}_{\mathit{VP}_D}(d)(c) \\
    & = \min_{c\in\Omega(p,q)\cap \mathit{VP}_D} \sum_{x_1,x_2\in
      X\times X} c(x_1,x_2) \cdot d(x_1,x_2) \\
    & = \min_{c\in\Omega(p,q)} \sum_{x_1,x_2\in X\times X}
    c(x_1,x_2) \cdot d(x_1,x_2) \\
    & = \mathcal{K}(d)(p,q) \tag*{\qedhere}
  \end{align*}
\end{proof}

We next present the approximation of the Kantorovich lifting in the
dual sense. Intuitively, given a distance function $d$ and a relation
$M$ on $X$, it characterises those pairs $(p,q)$ of distributions whose
distance in the Kantorovich metric decreases by a constant when we
decrease the distance $d$ for all pairs in $M$ by the same constant.

\begin{lemmarep}
  \label{le:associated-function-kantorovich}
  Let $d\colon X\times X\to [0,1]$. The approximation for the
  Kantorovich lifting $\mathcal{K}$ in the dual sense is
  $\mathcal{K}_\#^d\colon \Pow{\Ytop{X\times X}{d}} \to \Pow{\Ytop{D\times
    D}{\mathcal{K}(d)}}$ with
  \begin{eqnarray*}
    \mathcal{K}_\#^d(M) & = & \{(p,q)\in \Ytop{D\times D}{\mathcal{K}(d)}
    \mid \exists c\in
    \Omega(p,q), \mathit{supp}(c)\subseteq M, \\
    && \qquad \sum_{u,v\in S} d(u,v)\cdot c(u,v) =
    \mathcal{K}(d)(p,q)\}.
  \end{eqnarray*}
\end{lemmarep}

\begin{proof}
  Let $d\colon X\times X\to [0,1]$ and $M\subseteq \Ytop{X\times X}{d}$. Then we have:
  \begin{eqnarray*}
    \mathcal{K}_\#^d(M) & = &
    (\mins_u)_\#^{\mathrm{av}_{\mathit{VP}_D}(d)}
    ((\mathrm{av}_{\mathit{VP}_D})_\#^d(M)) 
  \end{eqnarray*}
  where
  \begin{eqnarray*}
    (\mathrm{av}_{\mathit{VP}_D})_\#^d\colon \Pow{\Ytop{X\times X}{d}} \to
    \Pow{\Ytop{\mathit{VP}_D}{\mathrm{av}_{\mathit{VP}_D}(d)}}
    \\
    (\mins_u)_\#^{\mathrm{av}_{\mathit{VP}_D}(d)}\colon
    \Pow{\Ytop{\mathit{VP}_D}{\mathrm{av}_{\mathit{VP}_D}(d)}}
    \to \Pow{\Ytop{D\times D}{\mathcal{K}(d)}}
  \end{eqnarray*}
  We are using the approximations associated to non-expansive
  functions, given in
  Proposition~\ref{pr:associated-approximations}, and obtain:
  \begin{align*}
    \mathcal{K}_\#^d(M) & = \{(p,q)\in \Ytop{D\times D}{\mathcal{K}(d)}
    \mid \arg\min_{c\in u^{-1}(p,q)}\mathrm{av}_{\mathit{VP}_D}(d)(c) \cap
    (\mathrm{av}_{\mathit{VP}_D})_\#^d(M) \neq \emptyset\} \\
    & = \{(p,q)\in \Ytop{D\times D}{\mathcal{K}(d)}\mid \exists
    c\in\Omega(p,q), c\in
    (\mathrm{av}_{\mathit{VP}_D})_\#^d(M), \\
    & \qquad \mathrm{av}_{\mathit{VP}_D}(d)(c) =
    \min_{c'\in\Omega(p,q)} \mathrm{av}_{\mathit{VP}_D}(d)(c')\} \\
    & = \{(p,q)\in \Ytop{D\times D}{\mathcal{K}(d)}\mid \exists
    c\in\Omega(p,q), c\in
    (\mathrm{av}_{\mathit{VP}_D})_\#^d(M), \\
    & \qquad \mathrm{av}_{\mathit{VP}_D}(d)(c) =
    \mathcal{K}(d)(p,q) \} \\
    & = \{(p,q)\in \Ytop{D\times D}{\mathcal{K}(d)}\mid \exists
    c\in\Omega(p,q),
    \mathit{supp}(c)\subseteq M, \\
    & \qquad \sum_{u,v\in S} d(u,v)\cdot c(u,v) = \mathcal{K}(d)(p,q)
    \} \tag*{\qedhere}
  \end{align*}
\end{proof}

\paragraph*{Hausdorff lifting.} Given a metric $d$ on a finite set
$X$, the Hausdorff lifting of $d$ provides a metric on the powerset
$\Pow{X}$. As for the Kantorovich lifting, we lift distance functions
that are not necessarily metrics.  The Hausdorff lifting is given by a
function
$\mathcal{H} : \monM^{X\times X} \to \monM^{\Pow{X}\times \Pow{X}}$
where
\[ \mathcal{H}(d)(X_1,X_2) = \max\{\max_{x_1\in X_1} \min_{x_2\in X_2}
  d(x_1,x_2),\max_{x_2\in X_2} \min_{x_1\in X_1} d(x_1,x_2)\}. \]
An alternative characterisation of the Hausdorff lifting due to
M\'emoli~\cite{m:wasserstein}, also observed
in~\cite{bbkk:coalgebraic-behavioral-metrics}, is more convenient for
our purposes.  Let $u : \Pow{X\times X}\to \Pow{X}\times \Pow{X}$ be
defined by $u(C) = (\pi_1[C],\pi_2[C])$, where $\pi_1,\pi_2$ are the
projections $\pi_i : X\times X\to X$ and
$\pi_i[C] = \{\pi_i(c)\mid c\in C\}$. Then
$\mathcal{H}(d)(X_1,X_2) = \min \{ \max_{(x_1,x_2) \in C} d(x_1,x_2)
\mid C \subseteq X \times X \ \land\ u(C) = (X_1,X_2)\}$, which can be seen to correspond to the Wasserstein distance with $C$ playing the role of couplings.
Relying on this characterisation, we can obtain the result below,
from which we deduce that $\mathcal{H}$ is non-expansive and
construct its approximation as the composition of the corresponding
functions from Table~\ref{tab:basic-functions-approximations}.

\begin{lemmarep}
  It holds that $\mathcal{H} = \mins_u\circ \maxs_\in$ where
  $\max_\in\colon \monM^{X\times X}\to\monM^{\Pow{X\times X}}$, with
  $\mathrel{\in}\ \subseteq (X\times X)\times\Pow{X\times X}$ the
  ``is-element-of''-relation on $X\times X$, and
  $\min_u\colon \monM^{\Pow{X\times X}} \to \monM^{\Pow{X}\times
    \Pow{X}}$.
\end{lemmarep}

\begin{proof}
  Let for $d : X\times X\to\monM$, $X_1,X_2\subseteq X$. Then we have
  \begin{eqnarray*}
    && \mins_u(\maxs_\in(d))(X_1,X_2) \\
    & = & \min_{u(C)=(X_1,X_2)} (\maxs_\in(d))(C) =
    \min_{u(C)=(X_1,X_2)} \max_{(x_1,x_2)\in C} a(x_1,x_2)
  \end{eqnarray*}
  which is exactly the definition of the Hausdorff lifting
  $\mathcal{H}(d)(X_1,X_2)$ via couplings, due to
  M\'emoli~\cite{m:wasserstein}.
\end{proof}

We next determine the approximation of the Hausdorff lifting in the
dual sense. Intuitively, given a distance function $d$ and a relation
$R$ on $X$, such function characterises those pairs $(X_1,X_2)$,
$X_1,X_2\subseteq X$, whose distance in the Hausdorff metric
decreases by a constant when we decrease the distance $d$ for all
pairs in $R$ by the same constant.

\begin{lemmarep}
  \label{le:associated-function-hausdorff}
  The approximation for the Hausdorff lifting $\mathcal{H}$ in the
  dual sense is as follows. Let $d\colon X\times X\to \monM$, then
  $\mathcal{H}_\#^d\colon \Pow{\Ytop{X\times X}{d}} \to
  \Pow{\Ytop{\Pow{X}\times \Pow{X}}{\mathcal{H}(d)}}$ with 
  \begin{eqnarray*}
    \mathcal{H}_\#^d(R) & = & \{(X_1,X_2)\in \Ytop{\Pow{X}\times
      \Pow{X}}{\mathcal{H}(d)} \mid \\
    && \qquad \forall x_1\in X_1 \big(\min_{x_2'\in X_2} d(x_1,x_2')
    = \mathcal{H}(d)(X_1,X_2)\,\Rightarrow\, \exists x_2\in
    X_2\colon \\
    && \qquad\qquad\qquad\qquad (x_1,x_2)\in R \land d(x_1,x_2) =
    \mathcal{H}(d)(X_1,X_2)\big)
    \mathop{\land} \\
    && \qquad \forall x_2\in X_2 \big( \min_{x_1'\in X_1}
    d(x_1',x_2) = \mathcal{H}(d)(X_1,X_2)\,\Rightarrow\, \exists
    x_1\in
    X_1\colon \\
    && \qquad\qquad\qquad\qquad (x_1,x_2)\in R \land d(x_1,x_2) =
    \mathcal{H}(d)(X_1,X_2)\big)\}
  \end{eqnarray*}
\end{lemmarep}

\begin{proof}
  Let $d\colon X\times X\to \monM$ and
  $R\subseteq \Ytop{X\times X}{d}$. Then we have:
  \begin{eqnarray*}
    \mathcal{H}_\#^d(R) & = & (\mins_u)_\#^{\max_\in(d)}
    ((\maxs_\in)_\#^d(R))
  \end{eqnarray*}
  where
  \begin{eqnarray*}
    (\maxs_\in)_\#^d\colon \Pow{\Ytop{X\times X}{d}} \to
    \Pow{\Ytop{\Pow{X\times X}}{\max_\in(d)}}
    \\
    (\mins_u)_\#^{\max_\in(d)}\colon
    \Pow{\Ytop{\Pow{X\times X}}{\max_\in(d)}} \to
    \Pow{\Ytop{\Pow{X}\times \Pow{X}}{\mathcal{H}(d)}}
  \end{eqnarray*}
  We are using the approximations associated to non-expansive
  functions, given in
  Proposition~\ref{pr:associated-approximations}, and obtain:
  \begin{eqnarray*}
    \mathcal{H}_\#^d(R)
    & = & \{ (X_1,X_2)\in\Ytop{\Pow{X}\times\Pow{X}}{\mathcal{H}(d)}
    \mid \arg\min\limits_{C\in u^{-1}(X_1,X_2)}\max\nolimits_\in(d)(C)
    \cap (\maxs_\in)_\#^d(R) \neq \emptyset\} \\
    & = & \{ (X_1,X_2)\in\Ytop{\Pow{X}\times\Pow{X}}{\mathcal{H}(d)}
    \mid \exists C\subseteq X\times X, u(C) = (X_1,X_2), \\
    && \qquad C\in (\maxs_\in)_\#^d(R), \mathrm{max}_\in(d)(C) =
    \min_{u(C')=(X_1,X_2)} \mathrm{max}_\in(d)(C')\} \\
    & = & \{ (X_1,X_2)\in\Ytop{\Pow{X}\times\Pow{X}}{\mathcal{H}(d)}
    \mid \exists C\subseteq X\times X, u(C) = (X_1,X_2), \\
    && \qquad C\in (\maxs_\in)_\#^d(R), \max d[C] =
    \min_{u(C')=(X_1,X_2)} \max d[C']\} \\
    & = & \{
    (X_1,X_2)\in\Ytop{\Pow{X}\times\Pow{X}}{\mathcal{H}(d)}\mid
    \exists C\subseteq X\times X, u(C) = (X_1,X_2),  \\
    && \qquad \arg\max\limits_{(y_1,y_2)\in C}d(y_1,y_2) \subseteq R, \max d[C] =
    \mathcal{H}(d)(X_1,X_2) \}
  \end{eqnarray*}
  We show that this is equivalent to the characterisation in the
  statement of the lemma.
  \begin{itemize}
  \item Assume that for all $x_1\in X_1$ such that
    $\min_{x_2'\in X_2} d(x_1,x_2') = \mathcal{H}(d)(X_1,X_2)$, there
    exists $x_2\in X_2$ such that $(x_1,x_2)\in R$ and
    $d(x_1,x_2) = \mathcal{H}(d)(X_1,X_2)$ (and vice versa).

    We define a set $C_m$ that contains all such pairs $(x_1,x_2)$,
    obtained from this guarantee. Now let $x_1\not\in
    \pi_1[C_m]$. Then necessarily
    $\min_{x_2'\in X_2} d(x_1,x_2') < \mathcal{H}(d)(X_1,X_2)$
    (because the minimal distance to an element of $X_2$ cannot
    exceed the Hausdorff distance of the two sets). Construct another
    set $C'$ that contains all such $(x_1,x_2)$ where $x_2$ is an
    argument where the minimum is obtained. Also add elements
    $x_2\not\in \pi_2[C_m]$ and their corresponding partners to $C'$.

    The $C = C_m\cup C'$ is a coupling for $X_1,X_2$, i.e.,
    $u(C) = (X_1,X_2)$. Furthermore
    $\arg\max_{(y_1,y_2)\in C}d(y_1,y_2) = C_m \subseteq R$ and
    $\max d[C] = \max d[C_m] = \mathcal{H}(d)(X_1,X_2)$.
  \item Assume that there exists $C\subseteq X\times X$,
    $u(C) = (X_1,X_2)$, $\arg\max_{(y_1,y_2)\in C}d(y_1,y_2) \subseteq R$,
    $\max d[C] = \mathcal{H}(d)(X_1,X_2)$.

    Now let $x_1\in X_1$ such that
    $\min_{x_2'\in X_2} d(x_1,x_2') = \mathcal{H}(d)(X_1,X_2)$. Since
    $C$ is a coupling of $X_1,X_2$, there exists $x_2\in X_2$ such
    that $(x_1,x_2)\in C\subseteq R$. It is left to show that
    $d(x_1,x_2) = \mathcal{H}(d)(X_1,X_2)$, which can be done as
    follows:
    \[
      \mathcal{H}(d)(X_1,X_2) = \min_{x_2'\in X_2} d(x_1,x_2') \le
      d(x_1,x_2) \le \max d[C] = \mathcal{H}(d)(X_1,X_2).
    \]
    For an $x_2\in X_2$ such that
    $\min_{x_1'\in X_1} d(x_1',x_2) = \mathcal{H}(d)(X_1,X_2)$ the
    proof is analogous.
    \qedhere
  \end{itemize}
\end{proof}

\paragraph*{Probabilistic automata}
We have now all the tools needed to discuss probabilistic automata. Let $(L, d_L)$ be a fixed metric space, which will be used for labelling states.

A \emph{probabilistic automaton} is a tuple
$\mathcal{A} = (S,\suc,\ell)$, where $S$ is a non-empty finite set
of states,
$\suc\colon S\to \Pow{\mathcal{D}(S)}$ assigns finite sets of
probability distributions to states and $\ell\colon S\to L$ is a
labelling function. (In the following we again replace
$\mathcal{D}(S)$ by a finite subset $D$.)

The \emph{probabilistic bisimilarity pseudo-metrics} is the least
fixpoint of the function $\mathcal{M}\colon$
$[0,1]^{S\times S} \to [0,1]^{S\times S}$ where for
$d\colon S\times S\to[0,1]$, $s,t\in S$:
\[
  \mathcal{M}(d)(s,t) = \max \{ d_L(\ell(s),\ell(t)), \ \mathcal{H}(\mathcal{K}(d))(\suc(s),\suc(t))\}
\]
where $\mathcal{H}$ is the Hausdorff lifting (for
$\monM=[0,1]$) and $\mathcal{K}$ is the Kantorovich lifting defined earlier.

The fixpoint function $\mathcal{M}$ can be expressed as the composition of more basic non-expansive functions and thus, by Theorem~\ref{th:bigtable}, it is
non-expansive itself.

\begin{lemmarep}[decomposing $\mathcal{M}$]
  \label{le:M-decomposition}
  The fixpoint function $\mathcal{M}$ for probabilistic bisimilarity
  pseudo-metrics can be written as:
  \[ \mathcal{M} = \maxs_\rho \circ (((\suc\times\suc)^*\circ
    \mathcal{H}\circ \mathcal{K}) \uplus ((\ell\times\ell)^* \circ c_{d_L})) \]
  where
  $\rho\colon (S\times S)\uplus (S\times S) \to (S\times S)$ with
  $\rho((s,t),i) = (s,t)$.
\end{lemmarep}

\begin{proof}
  In fact, given $d\colon S\times S\to [0,1]$ and $s,t \in S$, we have
  \begin{align*}
    & \maxs_\rho((((\suc\times\suc)^*\circ \mathcal{H}\circ \mathcal{K})
      \uplus ((\ell\times\ell)^* \circ c_{d_L}))(d))(s,t) \\
     = & \max\{(\suc\times\suc)^*\circ \mathcal{H}\circ
       \mathcal{K})(d)(s,t), ((\ell\times\ell)^* \circ d_L)(s,t) \} \\
     = & \max\{ \mathcal{H}(\mathcal{K}(d)(\suc(s),\suc(t)), d_L(\ell(s),\ell(t)) \} \\
     = & \mathcal{M}(d)(s,t) \tag*{\qedhere}
  \end{align*}
\end{proof}

As discussed below, whenever $d_L$ is discrete, this specializes to the
probabilistic automata of~\cite{bblmtv:prob-bisim-distance-automata}
and whenever the probability distributions are Dirac distributions we
obtain metric transition systems~\cite{afs:linear-branching-metrics}.

The above decomposition also helps in determining the approximation of
$\mathcal{M}$.

\begin{lemmarep}[approximating $\mathcal{M}$]
  \label{le:M-approximation}
  Let $d\colon S\times S \to [0,1]$. The approximation for
  $\mathcal{M}$ in the dual sense is
  $\mathcal{M}_\#^d\colon \Pow{\Ytop{S\times S}{d}} \to
  \Pow{\Ytop{S\times S}{\mathcal{M}(d)}}$ with
  \[
    \begin{array}{ll}
      \mathcal{M}_\#^d(X) = \{ (s,t) \in
      \Ytop{S\times S}{\mathcal{M}(d)} \mid
      & d_L(\ell(s),\ell(t)) < \mathcal{H}(\mathcal{K}(d))(\eta(s),\eta(t)) \\
      & \land\ (\eta(s),\eta(t))
      \in \mathcal{H}^{\mathcal{K}(d)}_\# \circ \mathcal{K}^d_\#(X) \}
    \end{array}
  \]
\end{lemmarep}

\begin{proof}
  Let $d\colon S\times S \to [0,1]$ and
  $X \subseteq \Ytop{S\times S}{d}$.
  We abbreviate
  $g = (\suc\times\suc)^* \circ \mathcal{H} \circ \mathcal{K}
  \colon [0,1]^{S\times S} \to [0,1]^{S\times S}$
  and $j = (\ell\times\ell)^* \circ c_{d_L} \colon
  [0,1]^{S\times S} \to [0,1]^{S\times S}$, so that
  $\mathcal{M} = \max_\rho\circ (g \uplus j)$.
  Thus we obtain
  \[ \mathcal{M}_\#^d (X) = (\max\nolimits_\rho)_\#^{(g \uplus j)(d)}
  \circ (g \uplus j)_\#^d(X) \]
  Since $c_{d_L} \colon [0,1]^{S\times S} \to [0,1]^{L\times L}$
  is a constant function and
  $((\ell \times \ell)^*)_\#^{c_{d_L}(d)} = (\ell\times\ell)^{-1}$,
  we deduce that
  \[ j_\#^d(X) = ((\ell \times \ell)^*)_\#^{c_{d_L}(d)} \circ (c_{d_L})_\#^d(X) = (\ell\times\ell)^{-1}(\emptyset) = \emptyset \]
  On the other hand
  \[ g_\#^d = ((\eta \times \eta)^*)_\#^{\mathcal{H}(\mathcal{K}(d))} \circ
    \mathcal{H}_\#^{\mathcal{K}(d)} \circ \mathcal{K}_\#^d\]
  where
  \begin{align*}
    \mathcal{H}_\#^{\mathcal{K}(d)} \circ \mathcal{K}_\#^d
    &\colon \Pow{\Ytop{S\times S}{d}}\to
    \Pow{\Ytop{\Pow{S}\times \Pow{S}}{\mathcal{H}(\mathcal{K}(d))}} \\
    ((\eta \times \eta)^*)_\#^{\mathcal{H}(\mathcal{K}(d))} &\colon
    \Pow{\Ytop{\Pow{S}\times \Pow{S}}{\mathcal{H}(\mathcal{K}(d))}}\to
    \Pow{\Ytop{S\times S}{g(d)}}
  \end{align*}
  We recall that
  $((\eta \times \eta)^*)_\#^{\mathcal{H}(\mathcal{K}(d))}= (\eta\times \eta)^{-1}$,
  and hence
  \[ (s,t) \in g_\#^d(X) \Leftrightarrow (\eta(s),\eta(t))
    \in \mathcal{H}_\#^{\mathcal{K}(d)} \circ \mathcal{K}_\#^d(X)\]
  Lastly, we obtain
  \begin{align*}
    \mathcal{M}_\#^d (X)& = (\max\nolimits_\rho)_\#^{(g \uplus j)(d)}
    \circ (g \uplus j)^d_\#(X) \\
    & = (\max\nolimits_\rho)_\#^{(g \uplus j)(d)}
    (g^d_\#(X) \uplus j^d_\#(X)) \\
    & = \{ (s,t) \in \Ytop{S\times S}{\mathcal{M}(d)} \mid
    \arg\max_{y\in \rho^{-1}(s,t)}
    (g \uplus j)(d)(y) \subseteq g_\#^d(X) \times \{0\} \}
  \end{align*}
  Recalling that
  $\rho^{-1}(s,t) = \{((s,t),0), ((s,t),1)\}$,
  the inclusion
  \[ \arg\max_{y\in \rho^{-1}(s,t)}(g \uplus j)(d)(y)
    \subseteq g_\#^d(X) \times \{0\} \]
  can only hold if
  $g(d)(s,t) > j(d)(s,t)$ (and hence the maximum is
  achieved by $g(d)$ instead of $j(d)$) and additionally
  $((s,t),0) \in g_\#^d(X) \times \{0\}$. Thus
  \begin{align*}
    \mathcal{M}_\#^d (X) &= \{ (s,t)\in \Ytop{S\times S}{\mathcal{M}(d)}
    \mid j(d)(s,t) < g(d)(s,t)\ \land\ (s,t) \in g^d_\#(X) \} \\
    & = \{ (s,t)\in \Ytop{S\times S}{\mathcal{M}(d)} \mid
    d_L(\ell(s),\ell(t)) < \mathcal{H}(\mathcal{K}(d))(\eta(s),\eta(t)) \\
    & \qquad\qquad\qquad\qquad\qquad\quad \land\ (\eta(s),\eta(t)) \in
    \mathcal{H}_\#^{\mathcal{K}(d)} \circ \mathcal{K}_\#^d(X) \} \tag*{\qedhere}
  \end{align*}
\end{proof}

\paragraph{Comparison with~\cite{bblmtv:prob-bisim-distance-automata}}
The paper~\cite{bblmtv:prob-bisim-distance-automata} describes the
first method for computing behavioural distances over probabilistic automata.
Although the behavioural distance arises as a least
fixpoint, it is in fact better, even the only known method, to iterate
from above, in order to reach this least fixpoint. This is done by
guessing and improving couplings, similarly to what happens for
strategy iteration discussed later in Section~\ref{se:ssgs}. A major
complication, faced in~\cite{bblmtv:prob-bisim-distance-automata}, is
that the procedure can get stuck at a fixpoint which is not the least
and one has to determine that this is the case and decrease the
current candidate.
This is done by relying on an adaptation of the notion of self-closed
relation from~\cite{f:game-metrics-markov-decision}, and next we argue
that this is closely related to the theory developed in
Section~\ref{se:proof-rules}.
In fact this was our inspiration to
generalise this technique to a more general setting.

We next establish a formal correspondence with our results. First note
that the probabilistic automata considered
in~\cite{bblmtv:prob-bisim-distance-automata} are a special case of
those defined above, where the metric on the set of state labels is
required to be discrete. Hence states with different labels are
necessarily at distance~$1$.

Let 
$\mathcal{A} = (S,\suc,\ell)$ be a fixed probabilistic automaton and
let us assume that, as in~~\cite{bblmtv:prob-bisim-distance-automata},
the metric space of labels $(L, d_L)$ is discrete.

Assume that $d$ is a fixpoint of $\mathcal{M}$, i.e.,
$d = \mathcal{M}(d)$. In order to check whether
$d = \mu f$,~\cite{bblmtv:prob-bisim-distance-automata} adapts the
notion of a self-closed relation
from~\cite{f:game-metrics-markov-decision}.

\begin{defiC}[{\cite{bblmtv:prob-bisim-distance-automata}}]
  \label{def:self-closed-relation}
  A relation $M\subseteq S\times S$ is \emph{self-closed} with respect to $d =
  \mathcal{M}(d)$ if, whenever $s\,M\,t$, then
  \begin{itemize}
  \item $\ell(s) = \ell(t)$ and $d(s,t) > 0$,
  \item if $p\in \suc(s)$ and $d(s,t) = \min_{q'\in \suc(t)}
    \mathcal{K}(d)(p,q')$, then there exists $q\in\suc(t)$ and $c\in
    \Omega(p,q)$ such that $d(s,t) = \sum_{u,v\in S} d(u,v)\cdot
    c(u,v)$ and $\mathit{supp}(c) \subseteq M$,
  \item if $q\in\suc(t)$ and $d(s,t) = \min_{p'\in \suc(s)}
    \mathcal{K}(d)(p',q)$, then there exists $p\in\suc(s)$ and $c\in
    \Omega(p,q)$ such that $d(s,t) = \sum_{u,v\in S} d(u,v)\cdot
    c(u,v)$ and $\mathit{supp}(c) \subseteq M$.
  \end{itemize}
\end{defiC}

The largest self-closed relation, denoted by $\approx_d$, can be shown
to be empty if and only if
$d = \mu f$~\cite{bblmtv:prob-bisim-distance-automata}.
This has an immediate correspondence with our results since we can prove an intimate connection between self-closed relations and post-fixpoints of the approximation of $\mathcal{M}$.

\begin{propositionrep}
  Let $d\colon S\times S\to [0,1]$ where $d = \mathcal{M}(d)$. Then
  $\mathcal{M}_\#^d\colon \Pow{\Ytop{S\times S}{d}}\to \Pow{\Ytop{S\times S}{d}}$, where
  $\Ytop{S\times S}{d} = \{(s,t)\in S\times S \mid d(s,t) > 0\}$.

  Then $M$ is a self-closed relation with respect to $d$ if and only if
  $M\subseteq \Ytop{S\times S}{d}$ and $M$ is a post-fixpoint of $\mathcal{M}_\#^d$.
\end{propositionrep}

\begin{proof}
  First note that whenever $M$ is self-closed, it holds that
  $d(s,t) > 0$ for all $(s,t) \in M$ and hence
  $M \subseteq \Ytop{S\times S}{d}$.

  Observe that whenever
  $\ell(s) \neq \ell(t)$, we would have
  $d_L(\ell(s),\ell(t)) = 1 \ge \mathcal{H}(\mathcal{K}(d))(\eta(s),\eta(t))$.
  On the other hand, when $\ell(s) = \ell(t)$, instead, we have
  $d_L(\ell(s),\ell(t)) = 0 < \mathcal{H}(\mathcal{K}(d))(\eta(s),\eta(t))$,
  since $\mathcal{H}(\mathcal{K}(d))(\eta(s),\eta(t)) > 0$
  for all $(s,t) \in \Ytop{S\times S}{d}$.
  So, by Lemma~\ref{le:M-approximation}, we obtain that
  \begin{align*}
    \mathcal{M}_\#^d(M) & = \{ (s,t) \in \Ytop{S\times S}{d} \mid
      d_L(\ell(s),\ell(t)) < \mathcal{H}(\mathcal{K}(d))(\eta(s),\eta(t)) \\
    & \qquad\qquad\qquad\qquad\quad\ \land\ (\eta(s),\eta(t))
      \in \mathcal{H}^{\mathcal{K}(d)}_\# \circ \mathcal{K}^d_\#(M) \} \\
    & = \{ (s,t) \in \Ytop{S\times S}{d} \mid
    \ell(s) = \ell(t)\ \land\ (\eta(s),\eta(t)) \in
      \mathcal{H}^{\mathcal{K}(d)}_\# \circ \mathcal{K}^d_\#(M) \} \\
    & = \{ (s,t) \in S\times S \mid
    d(s,t) > 0\ \land\ \ell(s) = \ell(t)\ \land\ (\eta(s),\eta(t))
      \in \mathcal{H}^{\mathcal{K}(d)}_\# \circ \mathcal{K}^d_\#(M) \}
  \end{align*}

  Using the characterisation of the associated approximation of
  the Hausdorff lifting in
  Lemma~\ref{le:associated-function-hausdorff}, we obtain that this
  is equivalent to
  \begin{center}
    for all $p\in\suc(s)$, whenever $\min_{q'\in\suc(t)}
    \mathcal{K}(d)(p,q') =
    \mathcal{H}(\mathcal{K}(d))(\suc(s),\suc(t))$, then there exists
    $q\in \suc(t)$ such that $(p,q)\in\mathcal{K}_\#^d(M)$ and
    $\mathcal{K}(d)(p,q) =
    \mathcal{H}(\mathcal{K}(d))(\suc(s),\suc(t))$ (and vice versa),
  \end{center}
  assuming that $\ell(s)=\ell(t)$ (this is a
  requirement in the definition of $\mathcal{M}_\#^d(M)$), since then we have 
  $\mathcal{H}(\mathcal{K}(d))(\suc(s),\suc(t)) = d(s,t) > 0$ and
  hence $(\suc(s),\suc(t)) \in \Ytop{\Pow{D}\times
  \Pow{D}}{\mathcal{H}(\mathcal{K}(d))}$.

  Since also $d=\mathcal{M}(d)$, the condition above can be rewritten to
  \begin{center}
    for all $p\in\suc(s)$, whenever
    $\min_{q'\in\suc(t)} \mathcal{K}(d)(p,q') = d(s,t)$, then there
    exists $q\in \suc(t)$ such that $(p,q)\in\mathcal{K}_\#^d(M)$ and
    $\mathcal{K}(d)(p,q) = d(s,t)$ (and vice versa).
  \end{center}
  From Lemma~\ref{le:associated-function-kantorovich} we know that
  $(p,q)\in \mathcal{K}_\#^d(M)$ iff $\mathcal{K}(d)(p,q) > 0$ and
  there exists $c\in\Omega(p,q)$ such that
  $\mathit{supp}(c)\subseteq M$ and
  $\sum_{u,v\in S} c(u,v)\cdot d(u,v) = \mathcal{K}(d)(p,q)$. We
  instantiate the condition above accordingly and obtain
  \begin{center}
    for all $p\in\suc(s)$, whenever
    $d(s,t) = \min_{q'\in\suc(t)} \mathcal{K}(d)(p,q')$, then there
    exists $q\in\suc(t)$ such that there exists $c\in\Omega(p,q)$ with
    $\mathit{supp}(c)\subseteq M$,
    $\mathcal{K}(d)(p,q) = \sum_{u,v\in S} c(u,v)\cdot d(u,v)$ and
    $\mathcal{K}(d)(p,q) = d(s,t)$ (and vice versa).
  \end{center}
  The two last equalities can be simplified to
  $d(s,t) = \sum_{u,v\in S} c(u,v)\cdot d(u,v)$, since
  \[ \mathcal{K}(d)(p,q) \le \sum_{u,v\in S} c(u,v)\cdot d(u,v) =
    d(s,t) = \min_{q'\in\suc(t)} \mathcal{K}(d)(p,q') \le
    \mathcal{K}(d)(p,q) \]
  and hence $\mathcal{K}(d)(p,q) = d(s,t)$ can be inferred from the
  remaining conditions.

  We finally obtain the following equivalent characterisation:
  \begin{center}
    for all $p\in\suc(s)$, whenever
    $d(s,t) = \min_{q'\in\suc(t)} \mathcal{K}(d)(p,q')$, then there
    exists $q\in\suc(t)$ such that there exists $c\in\Omega(p,q)$ with
    $\mathit{supp}(c)\subseteq M$,
    $d(s,t) = \sum_{u,v\in S} c(u,v)\cdot d(u,v)$ (and vice versa).
  \end{center}
  Hence we obtain that
  $(\eta(s),\eta(t))\in \mathcal{H}^{\mathcal{K}(d)}_\# \circ \mathcal{K}^d_\#(M)$ 
  is equivalent to the the
  second and third item of Def.~\ref{def:self-closed-relation}
  (under the assumption that $\ell(s)=\ell(t)$), while the first item
  is covered by the other conditions ($d(s,t) > 0$ and
  $\ell(s)=\ell(t)$) in the characterisation of $\mathcal{M}_\#^d(M)$.
\end{proof}

We observe that other kinds of probabilistic automata, e.g., those
originally introduced by Rabin~\cite{Rab:PA}, where transitions rather
than states are labelled and some states are marked as final, i.e., the transition relation is of the kind $\eta : S\to
\{0,1\} \times \mathcal{D}(S)^L$ or their non-deterministic variant, can be easily cast in our framework.

\paragraph{Branching Distances for Metric Transition Systems}
\label{sec:mts}

We observe that also metric transition systems (MTS) and their
(symmetrical) branching distances, as studied in
\cite{afs:linear-branching-metrics,fl:quantitative-spectrum-journal}, live in our framework.
In fact, a \emph{metric transition system} over some metric space
$(L, d_L)$ is essentially a probabilistic automaton as defined above,
where the probabilistic component is dropped, i.e.,
$\mathcal{A} = (S, \eta, \ell)$ where $S$ is the set of states,
$\eta : S \to \Pow{S}$ is the transition function and $\ell : S \to L$
a labelling function.

Clearly, a metric transition system $\mathcal{A} = (S, \eta, \ell)$
can be formally seen as a special probabilistic automaton.
Given a state $s \in S$, let $\beta_s$ denote the
Dirac distribution, assigning probability $1$ to $s$ and $0$ to
all other states. Then we can ``transform'' the  transition relation
$\suc \colon S\to \Pow{S}$ into $\suc' \colon S\to \Pow{\mathcal{D}(S)}$,
defining $\suc'(s) = \{ \beta_t \mid t \in \suc(s)\}$. Using this
observation, Lemma~\ref{le:M-approximation} and the fact that for a distance $d : S \times S \to \interval{0}{1}$ and a pair of states $s, t \in S$, it holds $\mathcal{K}(d)(\beta_s, \beta_t) = d(s,t)$,
we obtain the approximation:
\[
  \begin{array}{ll}
    \mathcal{M}_\#^d(X) = \{ (s,t)\in \Ytop{S\times
    S}{\mathcal{M}(d)} \mid
    & d_L(\ell(s),\ell(t)) < \mathcal{H}(d)(\eta(s),\eta(t))\\
    & \land\ (\eta(s),\eta(t))
      \in \mathcal{H}^{d}_\#(X) \}
  \end{array}
\]

\begin{exa}
  We consider the MTS depicted below, where the metric space of labels
  is the real interval $[0,1]$ with the Euclidean distance $d_L(r,s)= |r-s|$.

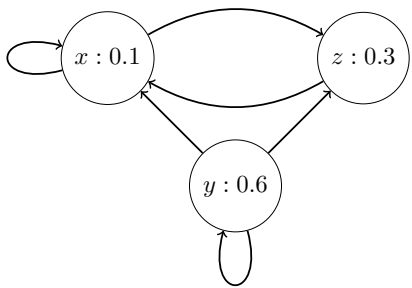
  \begin{center}
    \scalebox{0.85}{
      \begin{tikzpicture}[inner sep=0pt]
        \node (x) at (0,0) [circle,draw] {\begin{tabular}{c}
            $x:0.1$
          \end{tabular} };
        \node (y) at (2,-2) [circle,draw] {\begin{tabular}{c}
            $y:0.6$
          \end{tabular} };
        \node (z) at (4,0) [circle,draw] {\begin{tabular}{c}
            $z:0.3$
          \end{tabular} };

        \draw[->,thick, loop left] (x) to node [below]  {} (x);
        \draw[->,thick, loop below] (y) to node [below]  {} (y);
        \draw[->,thick] (y) to node [below]  {} (x);
        \draw[->,thick] (y) to node [below]  {} (z);
        \draw[->,thick, bend left] (x) to node [below]  {} (z);
        \draw[->,thick, bend left] (z) to node [below]  {} (x);
      \end{tikzpicture}
    }
  \end{center}
  Here, $\eta(x) = \{x,z\}$, $\eta(y) = \{x,y,z\}$ and
  $\eta(z) = \{x\}$. Additionally we have $\ell(x) = 0.1$, $\ell(y) = 0.6$
  and $\ell(z) = 0.3$ resulting in $d_L(\ell(x),\ell(y)) = 0.5$,
  $d_L(\ell(x),\ell(z)) = 0.2$ and $d_L(\ell(y),\ell(z)) = 0.3$. The least fixpoint of
  $\mathcal{M}$ is a pseudo-metric $\mu \mathcal{M}$ given by
  $\mu \mathcal{M}(x,y) = \mu \mathcal{M}(y,z)= 0.5$ and
  $\mu \mathcal{M}(x,z)= 0.3$. (Since $\mu \mathcal{M}$ is a
  pseudo-metric, the remaining entries are fixed:
  $\mu \mathcal{M}(u,u) = 0$ and
  $\mu \mathcal{M}(u,v) = \mu \mathcal{M}(v,u)$ for all
  $u,v\in \{x,y,z\}$.)

  Now consider the pseudo-metric $d$ with
  $d(x,y) = d(x,z) = d(y,z) = 0.5$. This is also a fixpoint of
  $\mathcal{M}$. Note that
  $\mathcal{H}(d)(\eta(x),\eta(y)) = \mathcal{H}(d)(\eta(x),\eta(z)) =
  \mathcal{H}(d)(\eta(y),\eta(z)) = 0.5$. Let us use our technique in
  order to verify that $d$ is not the least fixpoint of $\mathcal{M}$,
  by showing that $\nu \mathcal{M}_\#^{d} \neq \emptyset$.

  We start fixpoint iteration with the approximation
  $\mathcal{M}_\#^{d}$ from the top element $\Ytop{S \times S}{d}$,
  which is given by the symmetric closure\footnote{We denote the
  symmetric closure of a relation $R$ by $\symmclose{R}$.} of
  $\{(x,y),(x,z),(y,z)\}$ (since reflexive pairs do not contain
  slack).

  We first observe that the pairs
  $(x,y), (y,x) \notin \mathcal{M}_\#^d
  (\symmclose{\{(x,y),(x,z),(y,z)\}})$ since
  $d_L(\ell(x),\ell(y)) = 0.5 \not< \mathcal{H}(d)(\eta(x),\eta(y)) =
  0.5$.  Next,
  $(y,z), (z,y)\notin \mathcal{M}_\#^d (\symmclose{\{(x,z),(y,z)\}})$
  since  it holds
  $(\eta(y),\eta(z))\not\in
  \mathcal{H}_\#^d(\symmclose{\{(x,z),(y,z)\}})$.  In order to see
  this, consider the approximation of the Hausdorff lifting in
  Lemma~\ref{le:associated-function-hausdorff} and note that for
  $y\in \eta(y)$ we have
  $\min_{u \in \eta(z)} d(y,u) = 0.5 =
  \mathcal{H}(d)(\eta(y),\eta(z))$, but
  $(y,x) \notin \symmclose{\{ (x,z),(y,z)\}}$ (where $x$ is the only
  element in $\eta(z)$).

  The pairs $(x,z), (z,x)$ on the other hand satisfy all conditions
  and hence
  \[ \nu \mathcal{M}_\#^{d} = \symmclose{\{(x,z)\}} =
    \mathcal{M}_\#^{d}(\symmclose{\{(x,z)\}}) \neq \emptyset \]

  Thus we conclude that $d$ is not the least fixpoint, but, according
  to Proposition~\ref{pr:increases}, we can decrease the value of $d$ in
  the positions $(x,z), (z,x)$ and obtain a pre-fixpoint from which we
  can continue fixpoint iteration.
\end{exa}

\subsection{Bisimilarity}
\label{sec:bisimilarity}

In order to define standard bisimilarity we use a variant
$\mathcal{G}$ of the Hausdorff lifting $\mathcal{H}$ defined before,
where $\max$ and $\min$ are swapped. More precisely,
$\mathcal{G} : \monM^{X\times X} \to \monM^{\Pow{X}\times \Pow{X}}$ is
defined, for $d \in \monM^{X\times X}$, by
\begin{center}
$\mathcal{G}(d)(X_1,X_2) = \max \{ \min_{(x_1,x_2) \in C} d(x_1,x_2)
\mid C \subseteq X \times X \ \land\ u(C) = (X_1,X_2)\}$.
\end{center}

\begin{toappendix}
  \begin{lemma}
    \label{le:exchanged-hausdorff}
    The approximation for the adapted Hausdorff lifting $\mathcal{G}$
    in the primal sense is as follows. Let
    $a\colon X\times X\to \{0,1\}$, then
    $\mathcal{G}^\#_a\colon\Pow{\Ybot{X\times X}{a}}\to
    \Pow{\Ybot{\Pow{X}\times \Pow{X}}{a}} $ with 
    \begin{eqnarray*}
      \mathcal{G}_a^\#(R) & = & \{(X_1,X_2)\in \Ybot{\Pow{X}\times
        \Pow{X}}{\mathcal{H}(a)} \mid \\
      && \qquad\qquad\quad \forall x_1\in X_1\exists x_2\in
      X_2\colon \big((x_1,x_2)\not\in \Ybot{X\times X}{a} \lor (x_1,x_2)\in R\big)\\
      && \qquad\qquad \mathop{\land} \forall x_2\in X_2\exists x_1\in
      X_1\colon \big((x_1,x_2)\not\in \Ybot{X\times X}{a} \lor
      (x_1,x_2)\in R\big) \}
    \end{eqnarray*}
  \end{lemma}

\begin{proof}
  We rely on the characterisation of $\mathcal{H}_\#^a$ (dual case) of
  Lemma~\ref{le:associated-function-hausdorff} and we examine the
  case where $\monM = \{0,1\}$. In this case, whenever we have
  $(X_1,X_2)\in \Ytop{\Pow{X}\times\Pow{X}}{\mathcal{H}(a)}$ it must
  necessarily hold that $\mathcal{H}(a)(X_1,X_2) = 1$. Hence, the
  first part of the conjunction simplifies to:
  \[ \forall x_1\in X_1 \big(\min_{x_2'\in X_2} a(x_1,x_2') =
    1\,\Rightarrow\, \exists x_2\in X_2\colon (x_1,x_2)\in R \land
    a(x_1,x_2) = 1\big), \] from which we can omit $a(x_1,x_2) = 1$
  from the conclusion, since this holds automatically. Furthermore
  $\min_{x'_2\in X_2} a(x_1,x'_2) = 1$ can be rewritten to
  $\forall x_2\in X_2\colon a(x_1,x_2) = 1$. This gives us:
  \begin{eqnarray*}
    && \forall x_1\in X_1 \big(\lnot \forall x_2\in X_2\colon
    a(x_1,x_2) = 1 \mathop{\lor} \exists x_2\in X_2\colon (x_1,x_2)\in
    R \big) \\
    & \equiv & \forall x_1\in X_1 \big(\exists x_2\in X_2\colon
    a(x_1,x_2) = 0 \mathop{\lor} \exists x_2\in X_2\colon (x_1,x_2)\in
    R\big) \\
    & \equiv & \forall x_1\in X_1 \exists x_2\in X_2 \big( (x_1,x_2)
    \not\in \Ytop{X\times X}{a} \mathop{\lor} (x_1,x_2)\in R\big).
  \end{eqnarray*}
  Since this characterisation is independent of the order, we can
  replace $\Ytop{X\times X}{a}$ by $\Ybot{X\times X}{a}$ and obtain a
  characterizing condition for $\mathcal{G}_a^\#$ (primal
  case).
\end{proof}
\end{toappendix}
Now we can define the fixpoint function for bisimilarity and its
corresponding approximation. For simplicity we consider unlabelled
transition systems, but it would be straightforward to handle labelled
transitions.

Let $X$ be a finite set of states and $\suc : X\to\Pow{X}$ a function
that assigns a set of successors $\suc(x)$ to a state $x\in X$. The
fixpoint function for bisimilarity
$\mathcal{B} : \{0,1\}^{X\times X} \to \{0,1\}^{X\times X}$ can be expressed by using the
Hausdorff lifting $\mathcal{G}$ with $\monM = \{0,1\}$.

\begin{lemmarep}
  Bisimilarity on $\suc$ is the greatest fixpoint of
  $\mathcal{B} = (\suc\times\suc)^* \circ \mathcal{G}$.
\end{lemmarep}

\begin{proof}
  Let for $a : X\times X\to \{ 0,1\}$, $x,y\in X$. Then we have
  \begin{eqnarray*}
  (\suc\times\suc)^* \circ \mathcal{G}(a)(x,y) 
  & = & \mathcal{G} (a)(\suc(x),\suc(y)) \\
  & = & \maxs_u(\mins_\in(a))(\suc(x),\suc(y)) \\
  & = & \max_{u(C) = (\suc(x),\suc(y))} (\mins_\in^{X \times X}(a))(C) \\
  & = & \max_{u(C) = (\suc(x),\suc(y))} \min_{(x',y') \in C} a(x',y')
  \end{eqnarray*}
  Now we prove that this, indeed, corresponds with the standard bisimulation function, i.e.\ $\max_{u(C) = (\suc(x),\suc(y))} \min_{(x',y') \in C} a(x',y') = 1$ if and only if for all $x' \in \suc(x)$ there exists $y' \in \suc(y)$ such that $a(x',y') = 1$ and vice versa.
  For the first implication, assume that $\max_{u(C) = (\suc(x),\suc(y))} \min_{(x',y') \in C} a(x',y') = 1$. This means that there exists $C \subseteq X \times X$ such that $u(C) = (\pi_1(C), \pi_2(C)) = (\suc(x),\suc(y))$ and $\min_{(x',y') \in C} a(x',y') = 1$. Then we have two cases. Either $C = \emptyset$, which means that $\suc(x) = \suc(y) = \emptyset$, that is, $x$ and $y$ have no successors, and so the bisimulation property vacuously holds. Otherwise, $C \neq \emptyset$, and we must have $a(x',y') = 1$ for all $(x',y') \in C$. Then, since $(\pi_1(C), \pi_2(C)) = (\suc(x),\suc(y))$, for all $x' \in \suc(x)$ there must exists $y' \in \suc(y)$ such that $(x',y') \in C$, and thus $a(x',y') = 1$. Vice versa, for all $y' \in \suc(y)$ there must exists $x' \in \suc(x)$ such that $(x',y') \in C$, and thus $a(x',y') = 1$. So the bisimulation property holds.

  For the other implication, assume that for all $x' \in \suc(x)$ there exists $y' \in \suc(y)$ such that $a(x',y') = 1$ and call $c_1(x')$ such a $y'$. Vice versa, assume also that for all $y' \in \suc(y)$ there exists $x' \in \suc(x)$ such that $a(x',y') = 1$ and call $c_2(y')$ such a $x'$. This means that for all $x' \in \suc(x)$ and $y' \in \suc(y)$, we have $a(x',c_1(x')) = a(c_2(y'),y') = 1$. Now let $C' = \{(x',y') \in \suc(x) \times \suc(y) \mid c_1(x') = y' \lor x' = c_2(y')\}$. Since we assumed that for all $x' \in \suc(x)$ there exists $y' \in \suc(y)$ such that $c_1(x') = y'$, we must have that $\pi_1(C') = \suc(x)$. The same holds for all $y' \in \suc(y)$, thus $\pi_2(C') = \suc(y)$. Therefore, we know that $u(C') = (\suc(x),\suc(y))$, and we can conclude by showing that $a(x',y') = 1$ for all $(x',y') \in C'$, in which case also $\max_{u(C) = (\suc(x),\suc(y))} \min_{(x',y') \in C} a(x',y') = 1$. By definition of $C'$ either $c_1(x') = y'$ or $x' = c_2(y')$, or both, must hold. Assume the first one holds, the other case is similar. Then, we can immediately conclude since by hypothesis we know that $a(x',c_1(x')) = 1$.

  Since we proved that the function $\mathcal{B}$ is the same of the standard bisimulation function, then its greatest fixpoint $\nu \mathcal{B}$ is the bisimilarity on $\suc$.
\end{proof}

Since we are interested in the greatest fixpoint, we are working in
the primal sense. Bisimulation relations are represented by their
characteristic functions $a\colon X\times X\to \{0,1\}$, in fact the
corresponding relation can be obtained by taking the complement of
$\Ybot{X\times X}{a} = \{(x_1,x_2)\in X_1\times X_2 \mid a(x_1,x_2) =
0\}$.

\begin{lemmarep}
  \label{le:associated-function-bisimilarity}
  Let $a\colon X\times X\to \{0,1\}$. The approximation for the
  bisimilarity function $\mathcal{B}$ in the primal sense is
  $\mathcal{B}_a^\#\colon \Pow{\Ybot{X\times X}{a}} \to
  \Pow{\Ybot{X\times X}{\mathcal{B}(a)}}$ with
  \begin{eqnarray*}
    \mathcal{B}_a^\#(R) & = & \{(x_1,x_2)\in \Ybot{X\times X}{\mathcal{B}(a)} \mid \\
    && \quad \forall y_1\in \suc (x_1) \exists y_2\in
    \suc(x_2)\big( (y_1,y_2)\not\in \Ybot{X\times X}{a} \lor (y_1,y_2)\in R)
    \big)  \\
    &&  \  \land \forall y_2\in \suc (x_2) \exists y_1\in
    \suc(x_1)\big( (y_1,y_2)\not\in \Ybot{X\times X}{a} \lor (y_1,y_2)\in R)
    \big \}
  \end{eqnarray*}
\end{lemmarep}

\begin{proof}
  From Lemma~\ref{le:associated-function-hausdorff} we know that
  \begin{eqnarray*}
    \mathcal{G}_a^\#\colon \Ybot{X\times X}{a} & \to &
    \Ybot{\Pow{X}\times\Pow{X}}{\mathcal{G}(a)} \\
    \mathcal{G}_a^\#(R) & = & \{(X_1,X_2)\in \Ybot{\Pow{X}\times
    \Pow{X}}{\mathcal{G}(a)} \mid \\
    && \quad \forall x_1\in X_1\exists x_2\in
    X_2\colon \big((x_1,x_2)\not\in \Ybot{X\times X}{a} \lor (x_1,x_2)\in R\big)\\
    && \mathop{\land} \forall
    x_2\in X_2\exists x_1\in X_1\colon \big((x_1,x_2)\not\in \Ybot{X\times X}{a} \lor (x_1,x_2)\in R\big) \}.
  \end{eqnarray*}
  Furthermore 
  \begin{eqnarray*}
     ((\suc\times\suc)^*)_{\mathcal{G}(a)}^\# \colon
     \Ybot{\Pow{X}\times\Pow{X}}{\mathcal{G}(a)} & \to & \Ybot{X\times X}{\mathcal{B}(a)} \\
     ((\suc\times\suc)^*)_{\mathcal{G}(a)}^\#(Q) & = &
     (\suc\times\suc)^{-1}(Q) 
  \end{eqnarray*}
  Composing these functions we obtain:
  \begin{align*}
    \mathcal{B}_a^\#\colon \Ybot{X\times X}{a} & \to
    \Ybot{X\times X}{\mathcal{B}(a)} \\
    \mathcal{B}_a^\#(R) & = (\suc\times\suc)^{-1}(\{(Y_1,Y_2)\in
    \Ybot{\Pow{X}\times \Pow{X}}{\mathcal{G}(a)} \mid \\
    & \qquad \forall y_1\in Y_1\exists y_2\in
    Y_2\colon \big((y_1,y_2)\not\in \Ybot{X\times X}{a} \lor (y_1,y_2)\in R\big)\\
    & \quad\, \mathop{\land} \forall y_2\in Y_2\exists y_1\in Y_1\colon
    \big((y_1,y_2)\not\in \Ybot{X\times X}{a} \lor (y_1,y_2)\in R\big) \}) \\
    & = \{(x_1,x_2)\in
    \Ybot{X\times X}{\mathcal{B}(a)} \mid \\
    & \quad \forall y_1\in \suc(x_1)\exists y_2\in
    \suc(x_2)\colon \big((y_1,y_2)\not\in \Ybot{X\times X}{a} \lor
    (y_1,y_2)\in R\big)\\
    & \mathop{\land} \forall y_2\in \suc(x_2)\exists y_1\in
    \suc(x_1)\colon \big((y_1,y_2)\not\in \Ybot{X\times X}{a} \lor
    (y_1,y_2)\in R\big) \}. \tag*{\qedhere}
  \end{align*}
\end{proof}

We conclude this section by discussing how this view on bisimilarity
can be useful: first, it again opens up the possibility to compute
bisimilarity -- a greatest fixpoint -- by iterating from below,
through smaller fixpoints. This could potentially be useful if it is
easy to compute the least fixpoint of $\mathcal{B}$ inductively and
continue from there.

Furthermore, we obtain a technique for witnessing non-bisimilarity of
states. While this can also be done by exhibiting a distinguishing
modal formula~\cite{hm:hm-logic,c:automatically-explaining-bisim} or
by a winning strategy for the spoiler in the bisimulation game~\cite{s:bisim-mc-other-games}, to our
knowledge there is no known method that does this directly, based on
the definition of bisimilarity.

With our technique we can witness non-bisimilarity of two
states $x_1,x_2\in X$ by presenting a pre-fixpoint $a$ (i.e.,
$\mathcal{B}(a) \le a$) such that $a(x_1,x_2) = 0$ (equivalent to
$(x_1,x_2)\in \Ybot{X\times X}{a}$) and $\nu \mathcal{B}_a^\# = \emptyset$,
since this implies $\nu \mathcal{B}(x_1,x_2) \le a(x_1,x_2) = 0$ by
our proof rule.

There are two issues to discuss: first, how can we characterise a
pre-fixpoint of $\mathcal{B}$ (which is quite unusual, since
bisimulations are post-fixpoints)? In fact, the condition
$\mathcal{B}(a) \le a$ can be rewritten to: for all
$(x_1,x_2)\in \Ybot{X\times X}{a}$ 
there exists $y_1\in \suc(x_1)$ such that for all $y_2\in \suc(x_2)$
we have $(y_1,y_2)\in \Ybot{X\times X}{a}$ (\emph{or} vice versa).
Second, at first sight it does not seem as if we gained anything since
we still have to do a fixpoint computation on relations. However, the
carrier set is $\Ybot{X\times X}{a}$, i.e., a set of non-bisimilarity
witnesses and this set can be small even though $X$ might
be large, since $a$ might have value $0$ only on a small subset of
$X\times X$.

\begin{exa}
  \label{ex:bisim-1}
  We consider the transition system depicted below.
  
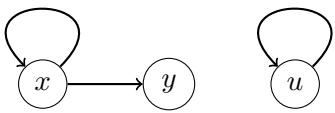
\begin{center}
    \begin{tikzpicture}[->]
      \node (x) [circle,draw]{$x$};
      \node [right=of x] (y) [circle,draw] {$y$};
      \node [right=of y] (u) [circle,draw] {$u$};
      \draw [->,thick] (x) to (y);
      \draw [->,thick,loop] (x) to (x);
      \draw [->,thick,loop] (u) to (u);
    \end{tikzpicture}
  \end{center}
  Our aim is to construct a witness
  showing that $x,u$ are not bisimilar. This witness is a function
  $a\colon X\times X\to \{0,1\}$ with $a(x,u) = 0 = a(y,u)$ and for
  all other pairs the value is $1$.
  Hence $\Ybot{X\times X}{a=\mathcal{B}(a)} = \Ybot{X\times X}{a} =
  \{(x,u),(y,u)\}$ and it is easy to check that $a$ is a pre-fixpoint
  of $\mathcal{B}$ and that $\nu \mathcal{B}_a^* = \emptyset$: we
  iterate over $\{(x,u),(y,u)\}$ and first remove $(y,u)$ (since $y$
  has no successors) and then $(x,u)$. This implies that
  $\nu \mathcal{B} \le a$ and hence $\nu \mathcal{B}(x,u) = 0$, which
  means that $x,u$ are not bisimilar.
\end{exa}

\begin{exa}
  \label{ex:bisim-2}
  We modify Example~\ref{ex:bisim-1} and consider a function $a$ where
  $a(x,u) = 0$ and all other values are $1$. Again $a$ is a
  pre-fixpoint of $\mathcal{B}$ and $\nu \mathcal{B} \le a$ (since
  only reflexive pairs are in the bisimilarity). However
  $\nu \mathcal{B}_a^* \neq \emptyset$, since $\{(x,u)\}$ is a
  post-fixpoint. This is a counterexample to completeness discussed
  after Theorem~\ref{th:soundness}.

  Intuitively speaking, the states $y,u$ over-approximate and claim that
  they are bisimilar, although they are not. (This is permissible for
  a pre-fixpoint.) This tricks $x,u$ into thinking that there is some
  wiggle room and that one can increase the value of $(x,u)$. This is
  true, but only because of the limited, local view, since the
  ``true'' value of $(y,u)$ is $0$.
\end{exa}

\section{Simple stochastic games}
\label{se:ssgs}

In this section we show how our techniques can be applied to simple
stochastic games~\cite{condon92,c:algorithms-ssg}. In particular, we
present two novel algorithms based on strategy iteration and discuss
some runtime results.

\subsection{Introduction to simple stochastic games.}

A simple stochastic game is a state-based two-player game where the
two players, Min and Max, each own a subset of states they control,
for which they can choose the successor. The system also contains sink
states with an assigned payoff and averaging states which randomly
choose their successor based on a given probability distribution. The
goal of Min is to minimise and the goal of Max to maximise the payoff.

Simple stochastic games are an important type of games that subsume
parity games and the computation of behavioural distances for
probabilistic automata (cf. Section~\ref{sec:prob-automata}, \cite{bblmtv:prob-bisim-distance-automata}). The associated decision
problem (if both players use their best strategies, is the
  expected payoff of Max greater than $\frac{1}{2}$?) is known to lie
in $\mathsf{NP}\cap \mathsf{coNP}$, but it is an open question
whether it is contained in $\mathsf{P}$. There are known randomised
subexponential algorithms~\cite{bv:randomized-algorithms-games}.

It has been shown that it is sufficient to consider positional
strategies, i.e., strategies where the choice of the player is only
dependent on the current state. The expected payoffs for each state
form a so-called value vector and can be obtained as the least
solution of a fixpoint equation (see below).

\smallskip
 
A \emph{simple stochastic game} is given by a finite set $V$ of nodes,
partitioned into $\mathit{MIN}$, $\mathit{MAX}$, $\mathit{AV}$
(average) and $\mathit{SINK}$, and the following data:
$\suc_{\min} : \mathit{MIN}\to \Pow{V}$,
$\suc_{\max} : \mathit{MAX}\to \Pow{V}$ (successor functions for Min
and Max nodes), $\suc_{\mathrm{av}} : \mathit{AV}\to D$ (probability
distributions, where $D\subseteq \mathcal{D}(V)$ finite) and
$w : \mathit{SINK}\to [0,1]$ (weights of sink nodes).

The fixpoint function $\mathcal{V}\colon [0,1]^V\to [0,1]^V$ is
defined below for $a\colon V\to [0,1]$ and $v\in V$:
\begin{align*}
  \mathcal{V}(a)(v) = \begin{cases}
    \min_{v' \in \suc_{\min}(v)} a(v') &v\in \mathit{MIN}\\
    \max_{v' \in \suc_{\max}(v)} a(v') &v\in \mathit{MAX}\\
    \sum_{v' \in V} \suc_\mathrm{av}(v)(v')\cdot a(v') &v\in \mathit{AV} \\
    \ell(v) &v\in \mathit{SINK}
  \end{cases}
\end{align*}

The \emph{least} fixpoint of $\mathcal{V}$ specifies the average
payoff for all nodes when Min and Max play optimally. In an infinite
game the payoff is $0$. In order to avoid infinite games and
guarantee uniqueness of the fixpoint, many authors~\cite{hk:nonterminating-stochastic-games,c:algorithms-ssg,tvk:strategy-improvement-ssg}
restrict to stopping games, which are guaranteed to terminate
for every pair of Min/Max-strategies. Here we deal with general
games where more than one fixpoint may exist. Such a scenario has
been studied in~\cite{kkkw:value-iteration-ssg}, which considers
value iteration to under- and over-approximate the value vector. The
over-approximation faces challenges with cyclic dependencies, similar
to the vicious cycles described earlier. Here we focus on
strategy iteration, which is usually less efficient than value
iteration, but yields a precise result instead of approximating it.

\begin{exa}
  \label{ex:ssg}
  We consider the game depicted below. Here $\min$ is a Min node with
  $\suc_{\min}(\min) = \{\textbf{1},\mathrm{av}\}$, $\max$ is a
  Max node with
  $\suc_{\max}(\max) = \{\bm{\varepsilon},\mathrm{av}\}$,
  $\textbf{1}$ is a sink node with payoff 1, $\bm{\varepsilon}$ is a
  sink node with some small payoff $\varepsilon\in (0,1)$ and
  $\mathrm{av}$ is an average node which transitions to both $\min$
  and $\max$ with probability~$\frac{1}{2}$. This game is not stopping.
  
  Min should choose $\mathrm{av}$ as successor since a payoff of $1$
  is bad for Min. Given this choice of Min, Max should not declare
  $\mathrm{av}$ as successor since this would create an infinite play
  and hence the payoff is $0$. Therefore Max has to choose
  $\bm{\varepsilon}$ and be content with a payoff of $\varepsilon$,
  which is achieved from all nodes different from $\bm{1}$.

  \begin{center}
    \scalebox{0.85}
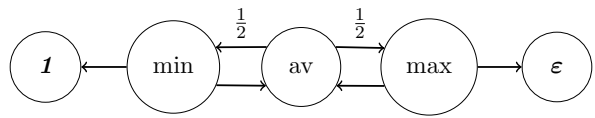
{
      \begin{tikzpicture}
        \node (S1) at (1,2) [circle,draw] {\begin{tabular}{c}
            $\textbf{1}$
          \end{tabular} };
        \node (MIN) at (3,2) [circle,draw] {\begin{tabular}{c}
            $\min$
          \end{tabular} };
        \node (AV) at (5,2) [circle,draw] {\begin{tabular}{c}
            $\mathrm{av}$
          \end{tabular} };

        \node (S0) at (9,2) [circle,draw] {\begin{tabular}{c}
            $\bm{\varepsilon}$
          \end{tabular} };
        \node (MAX) at (7,2) [circle,draw] {\begin{tabular}{c}
            $\max$
          \end{tabular} };

        \draw[->,thick] (MIN) to node [below]  {} (S1);
        \draw[->,thick] (MAX) to node [below]  {} (S0);

        \draw [thick]
        (MIN.337) edge[auto=right,->] node {} (AV.208)  ; 

        \draw [thick]
        (MAX.203) edge[auto=right,->] node {} (AV.332)  ; 

        \draw [thick]
        (AV.30) edge[auto=right,->] node [above] {$\frac{1}{2}$} (MAX.157)  ; 

        \draw [thick]
        (AV.150) edge[auto=right,->] node [above] {$\frac{1}{2}$} (MIN.25)  ; 

      \end{tikzpicture}
    }
  \end{center}
\end{exa}

In order to be able to determine the approximation of $\mathcal{V}$
and to apply our techniques, we consider the following equivalent
definition.

\begin{lemmarep}
  $\mathcal{V} = (\suc_{\min}^*\circ\mins_\in) \uplus
  (\suc_{\max}^*\circ\maxs_\in) \uplus
  (\suc_\mathrm{av}^*\circ\mathrm{av}_D) \uplus c_w$, where
  $\mathrel{\in}\ \subseteq V\times\Pow{V}$ is the
  ``is-element-of''-relation on $V$.
\end{lemmarep}

\begin{proof}
Let $a\colon V \to [0,1]$. For $v\in \mathit{MAX}$ we have
\[ \mathcal{V}(a)(v) = (\suc^*_{\max} \circ \maxs_\in)(a)(v) = \maxs_\in(a)(\suc_{\max}(v)) = \max_{v'\in \suc_{\max} (v)}a(v'). \]
For $v\in \mathit{MIN}$ we have
\[ \mathcal{V}(a)(v) = (\suc^*_{\min} \circ \mins_\in)(a)(v) = \mins_\in(a)(\suc_{\min}(v)) = \min_{v'\in \suc_{\min} (v)}a(v'). \]
For $v\in \mathit{AV}$ we have
\[  \mathcal{V}(a)(v) = (\suc^*_{\mathrm{av}}\circ \mathrm{av}_D)(a)(v) = \mathrm{av}_D(a)(\suc_\mathrm{av}(v)) = \sum_{v'\in V} \suc_\mathrm{av} (v)(v') \cdot a(v').\]
For $v\in \mathit{SINK}$ we have
 $\mathcal{V}(a)(v) = c_w(a)(v) =  w(v)$.
\end{proof}

As a composition of non-expansive functions, $\mathcal{V}$ is
non-expansive as well. Since we are interested in the least fixpoint
we work in the dual sense and obtain the following
approximation, which intuitively says: we can decrease a value at node
$v$ by a constant only if, in the case of a Min node, we decrease the
value of one successor where the minimum is reached, in the case of a
Max node, we decrease the values of all successors where the maximum
is reached, and in the case of an average node, we decrease the values
of all successors. %

\begin{lemmarep}
  \label{le:approximation-ssg}
  Let $a\colon V\to [0,1]$. The approximation for the value iteration
  function $\mathcal{V}$ in the dual sense is
  $\mathcal{V}_\#^a\colon \Pow{\Ytop{V}{a}} \to \Pow{\Ytop{V}{\mathcal{V}(a)}}$
  with
  \begin{eqnarray*}
    \mathcal{V}_\#^a(V') & = & \{v\in \Ytop{V}{\mathcal{V}(a)} \mid \big( v
    \in \mathit{MIN} \land \arg\min_{v'\in\suc_{\min}(v)} a(v')
    \cap V'\neq \emptyset \big) \mathop{\lor} \\
    && \quad \big( v \in \mathit{MAX} \land
    \arg\max_{v'\in\suc_{\max}(v)} a(v') \subseteq V'\big) \lor
    \big( v \in \mathit{AV} \land \mathit{supp}(\suc_{\mathrm{av}}(v))
    \subseteq V' \big) \}
  \end{eqnarray*}
\end{lemmarep}

\begin{proof}
  Let $a\colon V\to [0,1]$ and $V'\subseteq \Ytop{V}{a}$. By Proposition
  \ref{pr:approximation-union-fun} we have:
  \begin{align*}
    \mathcal{V}_\#^a(V') = &\big( \mathit{MIN} \cap (\suc^*_{\min}
    \circ \mins_{\in})^a_\# (V') \big) \cup
    \big( \mathit{MAX} \cap (\suc^*_{\max} \circ \maxs_{\in})^a_\#(V') \big) \cup \\
    &\big( \mathit{AV} \cap (\suc^*_{\mathrm{av}} \circ
    \mathrm{av}_D)^a_\# (V') \big) \cup \big( \mathit{SINK} \cap
    (c_w)_{\#}^a(V') \big)
  \end{align*}
  It holds that $({\suc^*_{\min}})_\#^{\mins_{\in}(v)}= \suc^{-1}_{\min}$,
  $({\suc^*_{\max}})_\#^{\maxs_{\in}(v)}= \suc^{-1}_{\max}$ and
  $({\suc^*_{\mathrm{av}}})_\#^{\mathrm{av}_D(v)}= \suc^{-1}_{\mathrm{av}}$.
Using previous results (Proposition~\ref{pr:associated-approximations}) we deduce
  \begin{align*}
    &v\in (\suc^*_{\min} \circ \mins_{\in})^a_\# (V') \Leftrightarrow \suc_{\min}(v) \in (\mins_{\in})^a_\# (V') \Leftrightarrow \arg\min_{v'\in\suc_{\min}(v)} a(v') \cap V'\neq \emptyset \\
    &v\in (\suc^*_{\max} \circ \maxs_{\in})^a_\# (V') \Leftrightarrow \suc_{\max}(v) \in (\maxs_{\in})^a_\# (V') \Leftrightarrow \arg\max_{v'\in\suc_{\max}(v)} a(v') \subseteq V'\\
    &v\in (\suc^*_{\mathrm{av}} \circ {\mathrm{av}_D})_\#^a (V')
    \Leftrightarrow \suc_{\mathrm{av}}(v) \in (\mathrm{av}_D)^a_\#
    (V') \Leftrightarrow \mathit{supp}(\suc_{\mathrm{av}}(v))
    \subseteq V'
  \end{align*}
  Lastly $(c_w)^a_\#(V') = \emptyset$ for any $V'\subseteq V$ since
  $c_w$ is a constant function which
  concludes the proof.
\end{proof}

\subsection{Strategy iteration from above and below.}

We describe two algorithms based on the idea of strategy iteration,
first introduced by Hoffman and Karp
in~\cite{hk:nonterminating-stochastic-games}, that are novel, as far
as we know. The first iterates to the least fixpoint from above and
uses the techniques described in Section~\ref{se:proof-rules} in
  order not to get stuck at a larger fixpoint. The second iterates
from below: the role of our results is not directly visible in the
code of the algorithm, but its non-trivial correctness proof is based
on the proof rule introduced earlier.

We first recap the underlying notions.  A Min-strategy is a mapping
$\tau\colon \mathit{MIN}\to V$ such that
$\tau(v)\in \suc_\mathrm{min}(v)$ for every $v\in
\mathit{MIN}$. Following such a strategy, Min decides to always leave
a node $v$ via $\tau(v)$. Analogously, a Max-strategy is a function $\sigma\colon \mathit{MAX}\to V$.
Fixing a strategy for either player induces a modified value function. If $\tau$ is a Min-strategy, we obtain 
$\mathcal{V}_\tau$ which is defined exactly as $\mathcal{V}$ but for
$v\in \mathit{MIN}$ where we set $\mathcal{V}_\tau(a)(v) =
a(\tau(v))$. Analogously, for $\sigma$ a Max-strategy, $\mathcal{V}_\sigma$ is obtained by setting
$\mathcal{V}_\sigma(a)(v) = a(\sigma(v))$ when
$v\in\mathit{MAX}$. If both players fix their strategies, the game
reduces to a Markov chain.
It is easy to see that fixing a strategy for one player produces an under- or over-approximation, depending on the player, of the function $\mathcal{V}$.

\begin{lemmarep}
  \label{le:Vmin>V>Vmax}
  For every Max-strategy $\sigma$ and every Min-strategy $\tau$ it holds that
  $\mathcal{V}_{\sigma} \leq \mathcal{V} \leq \mathcal{V}_{\tau}$.
\end{lemmarep}

\begin{proof}
  Given any $a\colon V\to [0,1]$ and $v\in V$, we have
  \begin{align*}
    \mathcal{V}_{\sigma}(a)(v) &= \begin{cases}
      \min_{v' \in \suc_{\min}(v)} a(v') &v\in \mathit{MIN} \\
      a(\sigma(v')) &v\in \mathit{MAX} \\
      \sum_{v' \in V} a(v')\cdot \suc_\mathrm{av}(v)(v') &v\in \mathit{AV} \\
      c_w(v) &v\in \mathit{SINK}
    \end{cases} \\
    &\leq \begin{cases}
      \min_{v' \in \suc_{\min}(v)} a(v') &v\in \mathit{MIN} \\
      \max_{v' \in \suc_{\max}(v)} a(v') &v\in \mathit{MAX} \\
      \sum_{v' \in V} a(v')\cdot \suc_\mathrm{av}(v)(v') &v\in \mathit{AV} \\
      c_w(v) &v\in \mathit{SINK}
    \end{cases} \\
    & = \mathcal{V}(a)(v)
  \end{align*}
  The same proof idea can be applied to show
  $\mathcal{V}\leq \mathcal{V}_{\tau}$.
\end{proof}

In order to describe our algorithms we also need the notion of a
\emph{switch}. Assume that $\tau$ is a Min-strategy and let $a$ be a
(pre-)fixpoint of $\mathcal{V}_\tau$. Min can now potentially improve her
strategy for nodes $v\in\mathit{MIN}$ where
$\min_{v'\in \suc_{\min}(v)} a(v') < a(\tau(v))$, called
\emph{switch nodes}. This results in a Min-strategy
$\tau' = \mathit{sw}_{\min}(\tau,a)$, where\footnote{If the
  minimum is achieved in several nodes, Min simply chooses one of
  them. However, she will only switch if this strictly improves the
  value.}
$\tau'(v) = \text{arg}\min_{v'\in \suc_{\min}(v)} a^{(i)}(v')$ for a
switch node $v$ and $\tau'$, $\tau$ agree otherwise. Also,
$\mathit{sw}_{\max}(\sigma,a)$ is defined analogously for Max
strategies.

\begin{figure}[t]
  \captionsetup[subfigure]{justification=centering}
  \begin{subfigure}[t]{0.55\textwidth}
    \centering
    \begin{tcolorbox}[width=0.95\textwidth, colframe=white]
      \small
      \textbf{Determine $\mu \mathcal{V}$ (from above)}\\[-2mm]
      \hrule \ \\[-5mm]
      \begin{enumerate}[1.]
      \item Guess a Min-strategy $\tau^{(0)}$, $i:=0$
      \item \label{it:sia-2} $a^{(i)}:=\mu \mathcal{V}_{\tau^{(i)}}$
      \item \label{it:sia-3}
        $\tau^{(i+1)}:=\mathit{sw}_{\min}(\tau^{(i)},a^{(i)})$
      \item If $\tau^{(i+1)} \neq \tau^{(i)}$ then $i:=i+1$ and goto~\ref{it:sia-2}
      \item Compute $V' = \nu \mathcal{V}_\#^a$,
        where $a = a^{(i)}$. 
      \item If $V' = \emptyset$ then stop and return $a^{(i)}$.\\
        Otherwise 
        $a^{(i+1)} := a - (\inc{\mathcal{V}}{a}(V'))_{V'}$,
        $\tau^{(i+2)}:=\mathit{sw}_{\min}(\tau^{(i)},a^{(i+1)})$,
        $i:=i+2$, goto~\ref{it:sia-2}
      \end{enumerate}
    \end{tcolorbox}
    \caption{Strategy iteration from above}
    \label{fig:sia}
  \end{subfigure}%
  \begin{subfigure}[t]{0.45\textwidth}
    \centering
    \begin{tcolorbox}[width=0.95\textwidth, colframe=white]
      \small
      \textbf{Determine $\mu \mathcal{V}$ (from below)}\\[-2mm]
      \hrule \ \\[-5mm]
      \begin{enumerate}[1.]
      \item Guess a Max-strategy $\sigma^{(0)}$, $i:=0$
      \item \label{it:sib-2} $a^{(i)}:=\mu \mathcal{V}_{\sigma^{(i)}}$
      \item $\sigma^{(i+1)}:=\mathit{sw}_{\max}(\sigma^{(i)},a^{(i)})$
      \item If $\sigma^{(i+1)} \neq \sigma^{(i)}$ then\\
        \mbox{} \quad $i:=i+1$ and
        goto~\ref{it:sib-2}\\
        Otherwise stop and return $a^{(i)}$.
      \end{enumerate}
    \end{tcolorbox}
    \caption{Strategy iteration from below}
    \label{fig:sib}
  \end{subfigure}
  \caption{Strategy iteration from above and below}
\end{figure}

Now strategy iteration from above works as described in
Fig.~\ref{fig:sia}.
The computation of
$\mu\mathcal{V}_{\tau^{(i)}}$ in the second step intuitively means
that Max chooses his best answering strategy and we compute the least
fixpoint based on this answering strategy. At some point no further
switches are possible and we have reached a fixpoint $a$, which need
not yet be the least fixpoint. Hence we use the techniques from
Section~\ref{se:proof-rules} to decrease $a$ and obtain a
new pre-fixpoint $a^{(i+1)}$, from which we can continue. The correctness of this
procedure partially follows from Theorem~\ref{th:fixpoint-sound-compl}
and Proposition~\ref{pr:increases}, however we also need to show the
following: first, we can compute
$a^{(i)} = \mu \mathcal{V}_{\tau^{(i)}}$ efficiently by solving a
linear program (cf. Lemma~\ref{le:linear-program}) by
adapting~\cite{condon92}. Second, the chain of the $a^{(i)}$
decreases, which means that the algorithm will eventually terminate
(cf. Theorem~\ref{th:ssg-terminate-correct}).

Strategy iteration from below is given in Fig.~\ref{fig:sib}. At
first sight, the algorithm looks simpler than strategy iteration from
above, since we do not have to check whether we have already reached
$\nu \mathcal{V}$, reduce and continue from there. However, in this
case the computation of $\mu \mathcal{V}_{\sigma^{(i)}}$ via a linear
program is more involved (cf. Lemma~\ref{le:linear-program}), since we
have to pre-compute (via greatest fixpoint iteration over $\Pow{V}$)
the nodes where Min can force a cycle based on the current strategy of
Max, thus obtaining payoff $0$.

This algorithm does not directly use our technique but we can use our
proof rules to prove the correctness of the algorithm
(Theorem~\ref{th:ssg-terminate-correct}). In particular, the proof that
the sequence $a^{(i)}$ increases is quite involved: we have to show
that
$a^{(i)} = \mu \mathcal{V}_{\sigma^{(i)}} \le \mu
\mathcal{V}_{\sigma^{(i+1)}} = a^{(i+1)}$. This could be done by
showing that $\mu \mathcal{V}_{\sigma^{(i+1)}}$ is a pre-fixpoint of
$\mathcal{V}_{\sigma^{(i)}}$, but there is no straightforward way to
do this. Instead, we prove this fact using our proof rules, by showing
that $a^{(i)}$ is below the least fixpoint of
$\mathcal{V}_{\sigma^{(i+1)}}$.

The algorithm generalises strategy iteration by Hoffman and
Karp~\cite{hk:nonterminating-stochastic-games}. Note that we cannot
simply adapt their proof, since we do not assume that the game is
stopping, which is a crucial ingredient. In the case where the game is
stopping, the two algorithms coincide, meaning that we also provide an
alternative correctness proof in this situation, while other
correctness proofs~\cite{condon92} are based on linear algebra and
inverse matrices.

\begin{lemmarep}
  \label{le:linear-program}
  The least fixpoints of $\mathcal{V}_\tau$ and $\mathcal{V}_\sigma$
  can be determined by solving linear programs.
\end{lemmarep}

\begin{proof}
  We adapt the linear programs found in the literature on simple
  stochastic games (see e.g.~\cite{condon92}).

  The least fixpoint $a = \mu \mathcal{V}_\tau$ can be determined by
  solving the following linear program:
  \begin{align*}
    \min &\sum_{v\in V}a(v)\\
    &a(v) = a(\tau(v)) &v\in \mathit{MIN} \\
    &a(v) \geq a(v') &\forall v'\in \suc_{\max}(v),v\in \mathit{MAX} \\
    &a(v) = \sum_{v' \in V} a(v')\cdot \suc_\mathrm{av}(v)(v')  &v\in \mathit{AV} \\
    &a(v) = w(v) &v\in \mathit{SINK}
  \end{align*}
  By having $a(v) \geq a(v')$ for all $v'\in \suc_{\max}(v)$ and
  $v\in\mathit{MAX}$ we guarantee
  $a(v) = \max_{v'\in\suc_{\max}(v)}$ $a^{(i)}(v')$ since we
  minimise. The minimisation also guarantees computation of the least
  fixpoint (in particular, nodes that lie on a cycle will get a value
  of $0$). Hence, the linear program correctly characterises
  $\mu \mathcal{V}_\tau$.

  \bigskip
  
  Given a strategy $\sigma$ for Max, we can determine
  $a = \mu \mathcal{V}_\sigma$ by solving the following linear program:
  \begin{align*}
    \max &\sum_{v\in V}a(v)\\
    &a(v) = 0 &v \in C_{\sigma} \\
    &a(v) \leq a(v')  &\forall v'\in \suc_{\min}(v), v\in
    \mathit{MIN}, v \not\in C_{\sigma} \\
    &a(v) = a(\sigma (v)) &v\in \mathit{MAX}, v
    \not\in C_{\sigma} \\
    &a(v) = \sum_{v' \in V} a(v')\cdot
    \suc_\mathrm{av}(v)(v') &v\in \mathit{AV}, v \not\in C_{\sigma} \\
    &a(v) = w(v) &v\in \mathit{SINK}
  \end{align*}
  The set $C_\sigma$ contains those nodes which will
  guarantee a non-terminating play if Min plays optimally, given the
  fixed Max-strategy $\sigma$. 

  The set $C_\sigma$ can again be computed via fixpoint-iteration by
  computing the greatest fixpoint of $c_\sigma$ via Kleene iteration
  on $\Pow{V}$ from above:
  \begin{eqnarray*}
    c_\sigma \colon \Pow{V} & \to & \Pow{V} \\
    c_\sigma(V') & = & \{ v\in V\mid (v\in \mathit{MIN} \land
    \suc_{\min}(v)\cap V' \neq \emptyset) \lor (v\in \mathit{MAX}
    \land \sigma(v)\in V') \\
    && \qquad \mathop{\lor} (v\in \mathit{AV} \land
    \mathit{supp}(\suc_{\mathrm{av}}(v))\subseteq V') \}
  \end{eqnarray*}
  It is easy to see that $C_\sigma = \nu c_\sigma$ contains all those
  nodes from which Min can force a non-terminating play and hence
  achieve payoff $0$. (Note that there are further nodes that
  guarantee payoff $0$ -- namely sinks with that payoff and nodes
  which can reach such sinks -- but those will obtain value $0$ in any
  case.)
  
  We now show that this linear program computes
  $\mu \mathcal{V}_\sigma$: first, by requiring $a(v) \leq a(v')$ for
  all $v\in\mathit{MIN}$, $v'\in \suc_{\min}(v)$, we guarantee
  $a(v) = \min_{v'\in \suc_{\min}} a(v')$ since we maximise. Hence we
  obtain the greatest fixpoint of the following function
  $\mathcal{V}'_\sigma\colon [0,1]^V \to [0,1]^V$:
  \begin{align*}
    \mathcal{V}'_\sigma(a)(v)& = \begin{cases}
      0 &v\in C_\sigma \\
      \sum_{v' \in V} a(v')\cdot \suc_\mathrm{av}(v)(v') &v\in \mathit{AV},
      v\not\in C_\sigma\\
      a(\sigma(v)) &v\in \mathit{MAX},
      v\not\in C_\sigma \\
      \min_{v' \in \suc_{\min}(v)} a(v') &v\in \mathit{MIN},
      v\not\in C_\sigma \\
      w(v) &v\in\mathit{SINK} \\
    \end{cases}
  \end{align*}

  It is easy to show that the least fixpoints of $\mathcal{V}'_\sigma$
  and $\mathcal{V}_\sigma$ agree, i.e., $\mu \mathcal{V}'_\sigma$ and
  $\mu \mathcal{V}_\sigma$:
  \begin{itemize}
  \item $\mu \mathcal{V}'_\sigma \le \mu \mathcal{V}_\sigma$ can be
    shown by observing that $\mathcal{V}'_\sigma \le
    \mathcal{V}_\sigma$.
  \item $\mu \mathcal{V}_\sigma \le \mu \mathcal{V}'_\sigma$ can be
    shown by proving that $\mu \mathcal{V}'_\sigma$ is a pre-fixpoint
    of $\mathcal{V}_\sigma$, which can be done via a straightforward
    case analysis.

    We have to show
    $\mathcal{V}_\sigma(\mu \mathcal{V}'_\sigma)(v) \le \mu
    \mathcal{V}'_\sigma(v)$ for all $v\in V$. We only spell out the
    case where $v\in \mathit{AV}$, the other cases are similar. In
    this case either $v\not\in C_\sigma$, which means that
    \[ \mathcal{V}_\sigma(\mu \mathcal{V}'_\sigma)(v) =
      \mathcal{V}'_\sigma(\mu \mathcal{V}'_\sigma)(v) = \mu
      \mathcal{V}'_\sigma(v). \] If instead $v\in C_\sigma$, we have
    that $\mathit{supp}(\suc_{\mathrm{av}}(v)) \subseteq C_\sigma$ and
    so $\mu \mathcal{V}'_\sigma(v') = 0$ for all
    $v'\in \mathit{supp}(\suc_{\mathrm{av}}(v))$. Hence
    \[ \mathcal{V}_\sigma(\mu \mathcal{V}'_\sigma)(v) =
      \sum_{v'\in V} \suc_{\mathrm{av}}(v)(v')\cdot
      \mu \mathcal{V}'_\sigma(v') = 0 = \mu \mathcal{V}'_\sigma(v) \]
  \end{itemize}

  If we can now show that $\mathcal{V}'_\sigma$ has a unique fixpoint, we
  are done. The argument for this goes as follows: assume that this
  function has another fixpoint $a'$ different from
  $\mu\mathcal{V}'_\sigma$. Clearly
  $\Ytop{V}{a'} \cap C_\sigma = \emptyset$, where
  $\Ytop{V}{a'} = \{v\in V\mid a'(v) \neq 0\} $. Hence, if we compare
  $(\mathcal{V}'_\sigma)_\#^a\colon
  \Pow{\Ytop{V}{a}}\to\Pow{\Ytop{V}{\mathcal{V}'_\sigma(a)}}$ (defined
  analogously to Lemma~\ref{le:approximation-ssg}) and $c_\sigma$
  above, we observe that
  $(\mathcal{V}'_\sigma)_\#^{a'} \subseteq
  c_\sigma|_{\Pow{\Ytop{V}{a'}}}$. (Both functions coincide, apart from
  their treatment of nodes $v\in\mathit{MIN}$, where $c_\sigma(V')$
  contains $v$ whenever one of its successors is contained in $V'$,
  whereas $(\mathcal{V}'_\sigma)_\#^{a'}(V')$ additionally requires
  that the value of this successor is minimal.) Since $a'$ is not the
  least fixpoint we have by Theorem~\ref{th:fixpoint-sound-compl} that
  \[ \emptyset \neq \nu (\mathcal{V}'_\sigma)_\#^{a'} \subseteq \nu
    (c_\sigma|_{\Pow{\Ytop{V}{a'}}}) \subseteq \nu c_\sigma = C_\sigma. \]
  This is a contradiction, since
  $\Ytop{V}{a'} \cap C_\sigma = \emptyset$ as observed above.

  This shows that $\mathcal{V}'_\sigma$ has a unique fixpoint and
  completes the proof.
  Note that if we do not explicitly require that the values of all
  nodes in $C_\sigma$ are $0$, $\mathcal{V}'_\sigma$ will potentially
  have several fixpoints and the linear program would not characterise
  the least fixpoint.
\end{proof}

\begin{theorem}
  \label{th:ssg-terminate-correct}
  Strategy iteration from above and below both terminate and compute
  the least fixpoint of $\mathcal{V}$.
\end{theorem}

\begin{proof}
  \noindent\emph{Strategy iteration from above:}
  
  We start by showing the following: Given any $a^{(i)}$ and a new
  switched Min-strategy $\tau^{(i+1)}$, i.e.,
  $\tau^{(i+1)} = \mathit{sw}_{\min}(\tau^{(i)},a^{(i)})$, then
  $a^{(i)}$ is a pre-fixpoint of $\mathcal{V}_{\tau^{(i+1)}}$. By
  choice of $\tau^{(i+1)}$ we have
  \begin{align*}
    \mathcal{V}_{\tau^{(i+1)}}(a^{(i)})(v)& =
    \begin{cases}
      a^{(i)}(\tau^{(i+1)}(v)) &v\in \mathit{MIN} \\
      \max_{v' \in \suc_{\max}(v)} a^{(i)}(v') &v\in \mathit{MAX} \\
      \sum_{v' \in V} a^{(i)}(v')\cdot \suc_\mathrm{av}(v)(v') &v\in \mathit{AV} \\
      w(v) &v\in \mathit{SINK}
    \end{cases} \\
    &= \begin{cases}
      \min_{v' \in \suc_{\min}(v)} a^{(i)}(v') &v\in \mathit{MIN} \\
      \max_{v' \in \suc_{\max}(v)} a^{(i)}(v') &v\in \mathit{MAX} \\
      \sum_{v' \in V} a^{(i)}(v')\cdot \suc_\mathrm{av}(v)(v') &v\in \mathit{AV} \\
      w(v) &v\in \mathit{SINK}
    \end{cases} \\
    &= \mathcal{V}(a^{(i)})(v) 
  \end{align*}
  By Lemma~\ref{le:Vmin>V>Vmax}
  we know that $\mathcal{V}\leq \mathcal{V}_{\tau^{(i)} }$, and since
  $a^{(i)}$ is a fixpoint of $\mathcal{V}_{\tau^{(i)}}$ we conclude
  \begin{align*}
    \mathcal{V}_{\tau^{(i+1)}}(a^{(i)})(v) = \mathcal{V}(a^{(i)})(v)
    \leq \mathcal{V}_{\tau^{(i)}}(a^{(i)})(v)=a^{(i)}(v)
  \end{align*}
  Thus we have $a^{(i+1)} \leq a^{(i)}$ (by Knaster-Tarski, since
  $a^{(i)}$ is a pre-fixpoint of $\mathcal{V}_{\tau^{(i+1)}}$ and
  $a^{(i+1)}$ is its least fixpoint). Furthermore we know that
  $a^{(i)}$ is not a fixpoint of $\mathcal{V}_{\tau^{(i+1)}}$
  (otherwise we could not have performed a switch) and hence
  $a^{(i+1)}$ is strictly smaller than $a^{(i)}$ for at least one
  input. Since there are only finitely many strategies we will
  eventually stop switching and reach a fixpoint $a = a^{(j)}$ for an
  index $j$.

  Then, if $V' = \nu \mathcal{V}_\#^a = \emptyset$ then $a$ is the
  least fixpoint and we conclude.

  Otherwise, we determine
  $a^{(j+1)} = a - (\inc{\mathcal{V}}{a}(V'))_{V'}$. By
  Proposition~\ref{pr:increases} (dual version), $a^{(j+1)}$ is a
  pre-fixpoint of $\mathcal{V}$. Now Min will choose her best strategy
  $\tau = \tau^{(j+2)} = \mathit{sw}_{\min}(\tau^{(i)},a^{(i+1)})$ and
  we continue computing $a^{(j+2)} = \mu \mathcal{V}_{\tau^{(j+2)}}$.
  First,
  observe that since $a^{(j+1)}$ is a pre-fixpoint of $\mathcal{V}$,
  it is also a pre-fixpoint of $\mathcal{V}_{\tau^{(j+2)}}$. In fact,
  $\mathcal{V}$ and $\mathcal{V}_{\tau^{(j+1)}}$ coincide on all nodes
  $v\not\in \mathit{MIN}$. If $v\in\mathit{MIN}$, we have
  \begin{eqnarray*}
    \mathcal{V}_{\tau^{(j+2)}}(a^{(j+1)})(v) & = &
    a^{(j+1)}(\tau^{(j+2)}(v)) \\
    & = & \min_{v'\in\suc_{\min}(v)}
      a^{(j+1)}(v') = \mathcal{V}(a^{(j+1)})(v) \le a^{(j+1)}(v). 
  \end{eqnarray*}
  Hence it follows by Knaster-Tarski that
  $a^{(j+2)} = \mu \mathcal{V}_{\tau^{(j+2)}} \le a^{(j+1)}$. In turn,
  $a^{(j+1)} < a^{(j)}$ since $V'$ is non-empty and hence also
  $a^{(j+2)} < a^{(j)}$ (where $<$ on tuples means means $\leq$ in all
  components and $<$ in at least one component.)

  This means that the chain $a^{(i)}$ is strictly
  descending. Hence, at each iteration we obtain a new strategy and, since the number of strategies is finite, the iteration will eventually stop.

  Hence the algorithm terminates and stops at the least fixpoint of
  $\mathcal{V}$.

  \bigskip
  
  \noindent\emph{Strategy iteration from below:}

  We start as follows: Assume $a$ is the least fixpoint of
  $\mathcal{V}_{\sigma}$, i.e.\ $a = \mu\mathcal{V}_\sigma$ and
  $\sigma'$ the new best strategy for Max obtained by switching with
  respect to $a$, i.e., $\sigma'=\mathit{sw}_{\max}(\sigma,a)$. We have
  to show that $a'= \mu \mathcal{V}_{\sigma '}$ lies above $a$
  ($a' \ge a$). Here we use our proof rules (see
  Theorem~\ref{th:soundness}) and show the following:
  \begin{itemize}
  \item First, observe that $a$ is a post-fixpoint of
    $\mathcal{V}_{\sigma '}$. For any $v\in V$ we have
    \begin{align*}
      a (v) = \mathcal{V}_{\sigma}(a)
      (v) &=
      \begin{cases}
        \min_{v'\in \suc_{\min}(v)} a(v') &v\in \mathit{MIN} \\
        a(\sigma(v)) &v\in \mathit{MAX} \\
        \sum_{v' \in V} a (v')\cdot \suc_\mathrm{av}(v)(v')
        &v\in \mathit{AV} \\
        w(v) &v\in \mathit{SINK}
      \end{cases} \\
      &\leq
      \begin{cases}
        \min_{v'\in \suc_{\min}(v)} a(v') &v\in \mathit{MIN} \\
        \max_{v'\in \suc_{\max}(v)} a(v')
        &v\in \mathit{MAX} \\
        \sum_{v' \in V} a (v')\cdot \suc_\mathrm{av}(v)(v')
        &v\in \mathit{AV} \\
        w(v) &v\in \mathit{SINK} \end{cases} \\
      &=
      \begin{cases}
        \min_{v'\in \suc_{\min}(v)} a(v') &v\in \mathit{MIN} \\
        a(\sigma'(v)) &v\in \mathit{MAX} \\
        \sum_{v' \in V} a (v')\cdot \suc_\mathrm{av}(v)(v')
        &v\in \mathit{AV} \\
        w(v) &v\in \mathit{SINK} \end{cases} \\
      &= \mathcal{V}_{\sigma'_{\max}}(a) (v)
    \end{align*}

  \item Next we show that
    $\nu(\mathcal{V}_{\sigma'})_*^{a} = \emptyset$, thus proving that
    $a \le \mu \mathcal{V}_{\sigma'} = a'$ by
    Theorem~\ref{th:soundness}. Note that
    $(\mathcal{V}_{\sigma'})_*^{a}\colon
    \Ytop{V}{a=\mathcal{V}_{\sigma'}(a)} \to
    \Ytop{V}{a=\mathcal{V}_{\sigma'}(a)}$, i.e., it restricts to those
    elements of $a$ where $a$ and $\mathcal{V}_{\sigma'}(a)$ coincide.

    Whenever $v\in\mathit{MAX}$ is a node where the strategy has been
    ``switched'' with respect to $a$, we have
    \[ \mathcal{V}_{\sigma'}(a)(v) = a(\sigma'(v)) > a(\sigma(v)) =
      a(v). \] The first equality above is true by the definition of
    $\mathcal{V}_{\sigma'}$ and the last equality holds since $a$ is a
    fixpoint of $\mathcal{V}_\sigma$.
    So if $v$ is a switch node, it holds that
    $v\not\in \Ytop{V}{a=\mathcal{V}_\sigma(a)}$. By contraposition if
    $v\in \Ytop{V}{a=\mathcal{V}_\sigma(a)}$, $v$ cannot be a switch node.

    We next show that $(\mathcal{V}_\sigma)_*^{a}$,
    $(\mathcal{V}_{\sigma'})_*^{a}$ agree on
    $\Ytop{V}{a=\mathcal{V}_{\sigma'}(a)} \subseteq \Ytop{V}{a} =
    \Ytop{V}{a=\mathcal{V}_\sigma(a)}$ (remember that $a$ is a
    fixpoint of $\mathcal{V}_\sigma$). It holds that
    \begin{eqnarray*}
      (\mathcal{V}_\sigma)_*^a(V') & = &
        \gamma^{\mathcal{V}_\sigma(a),\iota}(\mathcal{V}_\sigma(
        \alpha^{a,\iota}(V'))) \\
      (\mathcal{V}_{\sigma'})_*^a(V') & = &
        \gamma^{\mathcal{V}_{\sigma'}(a),\iota}(\mathcal{V}_{\sigma'}(
        \alpha^{a,\iota}(V'))) \cap
        \Ytop{V}{a=\mathcal{V}_{\sigma'}(a)}
    \end{eqnarray*}
    for a suitable constant $\iota$ and if we choose $\iota$ small
    enough we can use the same constant in both cases. Now let
    $v\in \Ytop{V}{a=\mathcal{V}_{\sigma'}(a)}$: by definition it
    holds that
    $v\in (\mathcal{V}_\sigma)_*^a(V') =
    \gamma^{\mathcal{V}_\sigma(a),\iota}(\mathcal{V}_\sigma(
    \alpha^{a,\iota}(V')))$ if and only if
    $\mathcal{V}_\sigma(\alpha^{a,\iota}(V'))(v) \ominus
    \mathcal{V}_\sigma(a)(v) \ge \iota$. Since, by the considerations
    above, $v$ is not a switch node,
    $\mathcal{V}_\sigma(b)(v) = \mathcal{V}_{\sigma'}(b)(v)$ for all
    $b$ and we can replace $\mathcal{V}_\sigma$ by
    $\mathcal{V}_{\sigma'}$, resulting in the equivalent statement
    $v\in
    \gamma^{\mathcal{V}_{\sigma'}(a),\iota}(\mathcal{V}_{\sigma'}(
    \alpha^{a,\iota}(V')))$, also equivalent to
    $v\in (\mathcal{V}_{\sigma'})_*^a(V')$.

    Thus
    $\nu (\mathcal{V}_{\sigma'})_*^{a} \subseteq \nu
    (\mathcal{V}_\sigma)_*^{a} = \emptyset$.

  \end{itemize}

  Hence we obtain an ascending sequence $a^{(i)}$. Furthermore,
  whenever we perform a switch, we know that $a^{(i)}$ is not a
  fixpoint of $\mathcal{V}_{\sigma^{(i+1)}}$ (otherwise we could not
  have performed a switch) and hence $a^{(i+1)}$ is strictly larger
  than $a^{(i)}$ for at least one input. Since there are only finitely
  many strategies we will eventually stop switching and reach the least
  fixpoint. 
\end{proof}

\begin{exa}
  The previous Example~\ref{ex:ssg} is well suited to explain our two algorithms.

  Starting with strategy iteration from above, we may guess
  $\tau^{(0)}(\min) = \textbf{1}$. In this case, Max would choose
  $\mathrm{av}$ as successor and we would reach a fixpoint, where
  each node except for $\bm{\varepsilon}$ is associated with a payoff
  of $1$. Next, our algorithm would detect the vicious cycle formed by
  $\min$, $\mathrm{av}$ and $\max$. We can reduce the values in this
  vicious cycle and reach the correct payoff values for each node.

  For strategy iteration from below assume that
  $\sigma^{(0)}(\max) = \mathrm{av}$. Given this strategy of Max,
  Min can force the play to stay in a cycle formed by $\min$,
  $\mathrm{av}$ and $\max$. Thus, the payoff achieved by the Max
  strategy $\sigma^{(0)}$ and an optimal play by Min would be $0$ for
  each of these nodes. In the next iteration Max switches and
  chooses $\bm{\varepsilon}$ as successor, i.e.
  $\sigma^{(1)}(\max) =\bm{\varepsilon}$, which results in the
  correct values.
\end{exa}
  
\subsection{Runtime results}

We implemented strategy iteration from above and from below -- in the
following abbreviated by SIA and SIB -- and classical Kleene iteration
(KI) in MATLAB. In Kleene iteration we terminate with a tolerance of
$10^{-14}$, i.e., we stop if the change from one iteration to the next
is below this value.

In order to test the algorithms we created random stochastic games
with $n$ nodes, where each Max, Min respectively average node has a
maximal number of $m$ successors. For each node we choose randomly one
of the four types of nodes. Sink nodes are given a random weight
uniformly in $[0,1]$. Max and Min nodes are randomly assigned to
successors and for an average nodes we assign a random number to each
of its successors, followed by normalisation to obtain a probability
distribution.

We performed 1000 runs with different randomly created systems for
each value of $n$ and $m = \frac{n}{2}$. Table \ref{table:runtime_results_1} shows the
runtimes in seconds and the number of iterations. Also, for
SIB, we display the number of nodes with a payoff of $0$ (for an
optimal play of Min) and the number of times SIA got stuck at any
other fixpoint which is not $\mu\mathcal{V}$ (all numbers --
runtime, iterations, etc. -- are summed
up over all
$1000$ runs).

\smallskip
\begin{table}
\begin{center}
  \begin{tabular}{|r||r|r|r||r|r|r||r|r|}
    \hline & \multicolumn{3}{c||}{runtime (seconds)} &
      \multicolumn{3}{c||}{number of iterations} &
      \multicolumn{1}{c|}{number nodes} & \multicolumn{1}{c|}{number of} \\
      \cline{1-7} $~~~n~~~$ & $~~$KI$~~$ & $~~$SIA$~~$ & $~~$SIB$~~$ &
      $~~$KI$~~$&
      $~~$SIA$~~$ & $~~$SIB$~~$ & $~~$payoff 0$~~$ & $~~$other fp$~~$ \\
      \hline \hline
      10 & 0.59 & 20.28 & 18.49 & 47302 & 2259 & 2152 & 2439 & 508 \\
      \hline
      20 & 1.05 & 31.71 & 25.96 & 30275 & 3620 & 3018 & 4714 & 743 \\
      \hline
      30 & 2.03 & 35.98 & 29.77 & 27361 & 3881 & 3275 & 7268 & 771 \\
      \hline
      40 & 3.77 & 38.84 & 32.67 & 26999 & 3850 & 3296 & 9806 & 756 \\
      \hline
      50 & 5.31 & 38.09 & 31.85 & 26604 & 3799 & 3215 & 12573 & 734 \\
      \hline
      60 & 7.63 & 40.33 &34.37 & 26467 & 3737 & 3218 & 15151 & 727 \\
      \hline
      70 & 10.77 & 45.00 & 37.50 & 26569 & 3751 & 3154 & 17473 & 751 \\
      \hline
      80 & 15.38 & 54.89 & 46.72 &26179 & 3713 & 3105 & 20031 & 752 \\
      \hline
      90 & 16.07 & 52.21 & 43.52 & 26401 & 3695 & 3083 & 22390 & 777 \\
      \hline
      100 & 19.46 & 60.29 & 50.88 & 26464 & 3654 & 3062 & 25163 & 751 \\
    \hline
  \end{tabular}
\end{center}

\caption{Experimental results for KI, SIA, SIB on randomly generated SSGs with weight of sink nodes in [0,1].} \label{table:runtime_results_1}

\end{table}

\smallskip

Note that SIB always performs slightly better than SIA. Moreover KI
neatly beats both of them.  Here we need to remember that KI only
converges to the solution and it is known that the rate of convergence
can be exponentially slow~\cite{c:algorithms-ssg}.

Note that the linear optimisation problems are quite costly to solve,
especially for large systems. Thus additional iterations are
substantially more costly compared to KI. Observe also that SIA has to
perform more iterations than SIB, which explains the slightly higher
runtime.

The number of nodes with a payoff of 0 seems to grow linearly with the
number of nodes in the system. The number of times SIA gets stuck at a
fixpoint different from $\mu\mathcal{V}$ however seems to be
independent of the system size and comparatively small.

We performed a second comparison (see Table \ref{table:runtime_results_2}), where we assigned to sink nodes a value
in {0, 1}, which is often done for simple stochastic games.

\smallskip

\begin{table}
\begin{center}
  \begin{tabular}{|r||r|r|r||r|r|r||r|r|r|r|r|}
    \hline & \multicolumn{3}{c||}{runtime (seconds)} &
      \multicolumn{3}{c||}{number of iterations} &
      \multicolumn{1}{c|}{number nodes} & \multicolumn{1}{c|}{number of} \\
      \cline{1-7} $~~~n~~~$ & $~~$KI$~~$ & $~~$SIA$~~$ & $~~$SIB$~~$ &
      $~~$KI$~~$ &
      $~~$SIB$~~$ & $~~$SIA$~~$ & $~~$payoff 0$~~$ & $~~$other fp$~~$ \\
      \hline \hline
      10 & 0.36 & 14.51 & 14.58 & 42547 & 1703 & 1702 & 5484 & 219 \\
      \hline
      20 & 1.00 & 19.85 & 19.98 & 29515 & 2385 & 2478 & 8168 &137 \\
      \hline
      30 & 1.97 & 20.45 & 20.77 & 27643 & 2367 & 2469 & 11502 & 33 \\
      \hline
      40 & 3.30 & 20.13 & 20.94 & 26761 & 2306 & 2383 & 14989 & 12 \\
      \hline
      50 & 4.96 & 20.24 & 20.94 & 26562 & 2253 & 2306 & 18821& 2 \\
      \hline
      60 & 6.87 & 20.57 & 21.19 & 26560 & 2176 & 2227 & 22573 & 0 \\
      \hline
      70 &9.14 & 21.95 & 22.35 & 26146 & 2142 & 2186 & 26260& 0 \\
      \hline
      80 & 11.73 & 24.69 & 24.94 &26235 & 2084 & 2131 & 30136 & 0 \\
      \hline
      90 & 14.73 & 28.90 & 28.71 & 26330 & 2066 & 2091 & 33930& 0 \\
      \hline
      100 & 18.22 & 34.75 & 34.84 & 26227 & 2051 & 2068 & 37496 & 0 \\
      \hline
  \end{tabular}
\end{center}

\caption{Experimental results for KI, SIA, SIB on randomly generated SSGs with weight of sink nodes in \{0,1\}.} \label{table:runtime_results_2}

\end{table}

\smallskip

Here, SIA performs very similar to SIB. The SIA approach seems to
suffer, since Max can easily find himself in a situation where he can
never reach a 1-sink, since only half of the sink nodes are of this
kind.  Additionally for these systems a significantly larger number of
nodes have a payoff of 0 and SIA is less likely to get stuck at a
fixpoint different from $\mu\mathcal{V}$. These factors seem to be
correlated since it is now ``harder'' for Min to choose a bad
successor (with a value greater than 0).

\section{Conclusion}
\label{ref:conclusion}

It is well-known that several computations in the context of system
verification can be performed by various forms of fixpoint iteration
and it is worthwhile to study such methods at a high level of
abstraction, typically in the setting of complete lattices and
monotone functions. Going beyond the classical results by
Tarski~\cite{t:lattice-fixed-point}, combination of fixpoint iteration
with approximations~\cite{CC:TLA,bkp:abstraction-up-to-games-fixpoint}
and with up-to techniques~\cite{p:complete-lattices-up-to} has proven
to be successful. Here we treated a more specific setting, where the
carrier set consists of functions from a finite set into an MV-chain
and the fixpoint functions are non-expansive (and hence monotone), and
introduced a novel technique to obtain upper bounds for greatest and
lower bounds for least fixpoints, also providing associated
algorithms. Such techniques are applicable to a wide range of examples
and so far they have been studied only in quite specific scenarios,
such as
in~\cite{bblmtv:prob-bisim-distance-automata,f:game-metrics-markov-decision,kkkw:value-iteration-ssg}.

In the future we plan to lift some of the restrictions of our
approach.
First, an extension to an infinite domain $Y$ would of course be
desirable, but since several of our results currently depend on
finiteness, such a generalisation does not seem to be easy. The
restriction to total orders, instead, seems easier to lift: in
particular, if the partially ordered MV-algebra $\bar{\monM}$ is of
the form $\monM^I$ where $I$ is a finite index set and $\monM$ an
MV-chain. (E.g., finite Boolean algebras are of this type.)  In this
case, our function space is
$\bar{\monM}^Y = \big(\monM^I\big)^{\raisebox{-2pt}{\scriptsize
    $Y$}} \cong \monM^{Y \times I}$ and we have reduced to the setting
presented in this paper. This will allow us to handle featured
transition systems~\cite{ccpshl:simulation-product-line-mc} for
compactly specifying software product lines in a single transition
system. There, transitions are equipped with boolean formulas that
specify for which products (or features) a transition can be taken.

There are several other application examples that did not fit into
this paper, but that can also be handled by our approach, for instance
coalgebraic behavioural
metrics~\cite{bbkk:coalgebraic-behavioral-metrics}.  While here we
introduced strategy iteration techniques for simple stochastic games,
we also want to check whether we can provide an improvement to value
iteration techniques, combining our approach
with~\cite{kkkw:value-iteration-ssg}. In this
context it is also interesting to consider more generic approaches to
strategy iteration, which is done for simple stochastic games in
\cite{ABS:GSIM} and in a lattice-theoretical settting in
\cite{bekp:lattice-strategy-iteration-arxiv}. The latter also
considers energy games~\cite{bcdgr:algorithms-mean-payoff-games} as a
new instance of our framework.

We also plan to study whether some examples can be handled with other
types of Galois connections: here we used an additive variant, but
looking at multiplicative variants (multiplication by a constant
factor) might also be fruitful.

\smallskip

\emph{Acknowledgements:} We are grateful to Ichiro Hasuo for making us
aware of stochastic games as an application domain. Furthermore we would
like to thank Timo Matt and Matthias Kuntz for their help with
experiments and implementation.

\bibliographystyle{alphaurl}
\bibliography{references}

\newcommand{\etalchar}[1]{$^{#1}$}
\begin{thebibliography}{KKKW18}

\bibitem[AdMS21]{ABS:GSIM}
David Auger, Xavier~Badin de~Montjoye, and Yann Strozecki.
\newblock A generic strategy improvement method for simple stochastic games.
\newblock In {\em {MFCS}}, volume 202 of {\em LIPIcs}, pages 12:1--12:22.
  Schloss Dagstuhl - Leibniz-Zentrum f{\"{u}}r Informatik, 2021.

\bibitem[BBKK18]{bbkk:coalgebraic-behavioral-metrics}
Paolo Baldan, Filippo Bonchi, Henning Kerstan, and Barbara K{\"o}nig.
\newblock Coalgebraic behavioral metrics.
\newblock {\em Logical Methods in Computer Science}, 14(3), 2018.
\newblock Selected Papers of the 6th Conference on Algebra and Coalgebra in
  Computer Science (CALCO 2015).

\bibitem[BBL{\etalchar{+}}19]{bblmtv:prob-bisim-distance-automata}
Giorgio Bacci, Giovanni Bacci, Kim~G. Larsen, Radu Mardare, Qiyi Tang, and
  Franck van Breugel.
\newblock Computing probabilistic bisimilarity distances for probabilistic
  automata.
\newblock In {\em Proc. of CONCUR '19}, volume 140 of {\em {LIPIcs}}, pages
  9:1--9:17. Schloss Dagstuhl -- Leibniz Center for Informatics, 2019.

\bibitem[BBLM17]{bblm:on-the-fly-exact-journal}
Giorgio Bacci, Giovanni Bacci, Kim~G. Larsen, and Radu Mardare.
\newblock On-the-fly exact computation of bisimilarity distances.
\newblock {\em Logical Methods in Computer Science}, 13(2:13):1--25, 2017.

\bibitem[BCD{\etalchar{+}}11]{bcdgr:algorithms-mean-payoff-games}
Lubos Brim, Jakub Chaloupka, Laurent Doyen, Raffaella Gentilini, and
  Jean-Fran{\,c}ois Raskin.
\newblock Faster algorithms for mean-payoff games.
\newblock {\em Formal Methods in System Design}, 38(2):97--118, 2011.

\bibitem[BEKP22]{bekp:lattice-strategy-iteration-arxiv}
Paolo Baldan, Richard Eggert, Barbara K\"onig, and Tommaso Padoan.
\newblock A lattice-theoretical view of strategy iteration, 2022.
\newblock arXiv:2207.09872.
\newblock URL: \url{https://arxiv.org/abs/2207.09872}.

\bibitem[BK08]{bk:principles-mc}
Christel Baier and Joost-Pieter Katoen.
\newblock {\em Principles of Model Checking}.
\newblock MIT Press, 2008.

\bibitem[BKP20]{bkp:abstraction-up-to-games-fixpoint}
Paolo Baldan, Barbara K\"onig, and Tommaso Padoan.
\newblock Abstraction, up-to techniques and games for systems of fixpoint
  equations.
\newblock In {\em Proc. of CONCUR '20}, volume 171 of {\em {LIPIcs}}, pages
  25:1--25:20. Schloss Dagstuhl -- Leibniz Center for Informatics, 2020.
\newblock \href {https://doi.org/10.4230/LIPIcs.CONCUR.2020.25}
  {\path{doi:10.4230/LIPIcs.CONCUR.2020.25}}.

\bibitem[BV05]{bv:randomized-algorithms-games}
Henrik Bj{\"o}rklund and Sergei Vorobyov.
\newblock Combinatorial structure and randomized subexponential algorithms for
  infinite games.
\newblock {\em Theoretical Computer Science}, 349(3):347--360, 2005.

\bibitem[CC77]{cc:ai-unified-lattice-model}
Patrick Cousot and Radhia Cousot.
\newblock Abstract interpretation: A unified lattice model for static analysis
  of programs by construction or approximation of fixpoints.
\newblock In {\em Proc. of POPL '77 (Los Angeles, California)}, pages 238--252.
  ACM, 1977.

\bibitem[CC00]{CC:TLA}
Patrick Cousot and Radhia Cousot.
\newblock Temporal abstract interpretation.
\newblock In Mark~N. Wegman and Thomas~W. Reps, editors, {\em Proc. of POPL
  '00}, pages 12--25. {ACM}, 2000.

\bibitem[CCP{\etalchar{+}}12]{ccpshl:simulation-product-line-mc}
Maxime Cordy, Andreas Classen, Gilles Perrouin, Pierre-Yves Schobbens, Patrick
  Heymans, and Axel Legay.
\newblock Simulation-based abstractions for software product-line model
  checking.
\newblock In {\em Proc. of ICSE '12 (International Conference on Software
  Engineering)}, pages 672--682. IEEE, 2012.

\bibitem[Cle90]{c:automatically-explaining-bisim}
Rance Cleaveland.
\newblock On automatically explaining bisimulation inequivalence.
\newblock In {\em Proc. of CAV '90}, pages 364--372. Springer, 1990.
\newblock {LNCS} 531.

\bibitem[Con90]{c:algorithms-ssg}
Anne Condon.
\newblock On algorithms for simple stochastic games.
\newblock In {\em Advances In Computational Complexity Theory}, volume~13 of
  {\em DIMACS Series in Discrete Mathematics and Theoretical Computer Science},
  pages 51--71, 1990.

\bibitem[Con92]{condon92}
Anne Condon.
\newblock The complexity of stochastic games.
\newblock {\em Information and Computation}, 96(2):203--224, 1992.
\newblock \href {https://doi.org/10.1016/0890-5401(92)90048-K}
  {\path{doi:10.1016/0890-5401(92)90048-K}}.

\bibitem[dFS09]{afs:linear-branching-metrics}
Luca {de Alfaro}, Marco Faella, and Mari{\"e}lle Stoelinga.
\newblock Linear and branching system metrics.
\newblock {\em IEEE Transactions on Software Engineering}, 35(2):258--273,
  2009.

\bibitem[FL14]{fl:quantitative-spectrum-journal}
Uli Fahrenberg and Axel Legay.
\newblock The quantitative linear-time-branching-time spectrum.
\newblock {\em Theoretical Computer Science}, 538:54--69, 2014.

\bibitem[Fu12]{f:game-metrics-markov-decision}
Hongfei Fu.
\newblock Computing game metrics on {M}arkov decision processes.
\newblock In {\em Proc. of ICALP '12, Part~II}, pages 227--238. Springer, 2012.
\newblock {LNCS} 7392.

\bibitem[GRS00]{GRS:MAIC}
Roberto Giacobazzi, Francesco Ranzato, and Francesca Scozzari.
\newblock Making abstract interpretations complete.
\newblock {\em Journal of the ACM}, 47(2):361--416, 2000.

\bibitem[HM85]{hm:hm-logic}
Matthew Hennessy and Robin Milner.
\newblock Algebraic laws for nondeterminism and concurrency.
\newblock {\em Journal of the ACM}, 32:137--161, 1985.

\bibitem[KH66]{hk:nonterminating-stochastic-games}
Richard~M. Karp and Alan~J. Hoffman.
\newblock On nonterminating stochastic games.
\newblock {\em Management Science}, 12(5):359--370, 1966.

\bibitem[KKKW18]{kkkw:value-iteration-ssg}
Edon Kelmendi, Julia Kr{\"a}mer, Jan K{\v r}et{\'{\i}}nsk{\'{y}}, and
  Maximilian Weininger.
\newblock Value iteration for simple stochastic games: Stopping criterion and
  learning algorithm.
\newblock In {\em Proc. of CAV '18}, pages 623--642. Springer, 2018.
\newblock {LNCS} 10981.

\bibitem[M{\'e}m11]{m:wasserstein}
Facundo M{\'e}moli.
\newblock {G}romov-{W}asserstein distances and the metric approach to object
  matching.
\newblock {\em Foundations of Computational Mathematics}, 11(4):417--487, 2011.

\bibitem[Mun07]{Mun:MV}
Daniele Mundici.
\newblock {MV}-algebras. {A} short tutorial.
\newblock Available at
  \url{http://www.matematica.uns.edu.ar/IXCongresoMonteiro/Comunicaciones/Mundici_tutorial.pdf},
  2007.

\bibitem[Mun11]{m:lukasiewicz-mv}
Daniele Mundici.
\newblock {\em Advanced {\L}ukasiewicz calculus and {MV}-algebras}, volume~35
  of {\em Trends in Logic}.
\newblock Springer, 2011.

\bibitem[NNH10]{NNH:PPA}
Flemming Nielson, Hanne~R. Nielson, and Chris Hankin.
\newblock {\em Principles of Program Analysis}.
\newblock Springer, 2010.

\bibitem[PC20]{pc:computational-ot}
Gabriel Peyr{\'e} and Marco Cuturi.
\newblock Computational optimal transport, 2020.
\newblock arXiv:1803.00567.
\newblock URL: \url{https://arxiv.org/abs/2009.14817}.

\bibitem[Pou07]{p:complete-lattices-up-to}
Damien Pous.
\newblock Complete lattices and up-to techniques.
\newblock In {\em Proc. of APLAS '07}, pages 351--366. Springer, 2007.
\newblock {LNCS} 4807.

\bibitem[Rab63]{Rab:PA}
Michael~O. Rabin.
\newblock Probabilistic automata.
\newblock {\em Information and Control}, 6(3):230--245, 1963.

\bibitem[RVAK11]{tvk:strategy-improvement-ssg}
Tripathi Rahul, Elena Valkanova, and V.S. Anil~Kumar.
\newblock On strategy improvement algorithms for simple stochastic games.
\newblock {\em Journal of Discrete Algorithms}, 9:263--278, 2011.

\bibitem[San11]{San:IntroBisCoind}
Davide Sangiorgi.
\newblock {\em Introduction to Bisimulation and Coinduction}.
\newblock Cambridge University Press, 2011.

\bibitem[Sti97]{s:bisim-mc-other-games}
Colin Stirling.
\newblock Bisimulation, model checking and other games.
\newblock Notes for Mathfit instructional meeting on games and computation,
  Edinburgh, June 1997.
\newblock URL: \url{http://homepages.inf.ed.ac.uk/cps/mathfit.pdf}.

\bibitem[Tar55]{t:lattice-fixed-point}
Alfred Tarski.
\newblock A lattice-theoretical fixpoint theorem and its applications.
\newblock {\em Pacific Journal of Mathematics}, 5:285--309, 1955.

\bibitem[Vil09]{v:optimal-transport}
C{\'e}dric Villani.
\newblock {\em Optimal Transport -- Old and New}, volume 338 of {\em A Series
  of Comprehensive Studies in Mathematics}.
\newblock Springer, 2009.

\end{thebibliography}

\appendix

\end{document}